\documentclass[10pt, oneside]{article}

\input ./math_headers.sty

\usepackage[backend=bibtex,
            isbn=false,
            doi=false,
            maxbibnames=5,
            giveninits=true,
            style=alphabetic,
            citestyle=alphabetic]{biblatex}
\renewbibmacro{in:}{}
\let\OLDthebibliography\thebibliography
\renewcommand\thebibliography[1]{
  \OLDthebibliography{#1}
  \setlength{\parskip}{0pt}
  \setlength{\itemsep}{4pt plus 0.3ex}
}
\bibliography{GL-KW}

\usepackage{rotating}
\usepackage{multirow}
\definecolor{light-gray}{gray}{0.9}
\newenvironment{psmallmatrix}
  {\left(\begin{smallmatrix}}
  {\end{smallmatrix}\right)}

\newcommand{\del}{\partial}
\def\d{{\rm d}}

\renewcommand{\flat}{\mathrm{Flat}}

\newcommand{\Oloc}{\mathcal{O}_{\mathrm{loc}}}

\begin{document}

\title{Quantum Geometric Langlands Categories\\ from $\mc N=4$ Super Yang--Mills Theory}
\author{Chris Elliott and Philsang Yoo}
\date{\today}

\maketitle

\begin{abstract}
We describe the family of supersymmetric twists of $\mc N=4$ super Yang--Mills theory using derived algebraic geometry, starting from holomorphic Chern--Simons theory on $\mc N=4$ super twistor space.  By considering an ansatz for categorical geometric quantization of the family of further twists of a fixed holomorphic twist, we give a quantum field-theoretic synthesis of the categories of twisted D-modules occurring in the quantum geometric Langlands correspondence.
\end{abstract}

\section{Introduction}

This paper aims to explore the occurrence of the quantum geometric Langlands correspondence in topologically twisted 4d gauge theory using the Batalin--Vilkovisky (BV) formalism: an approach to classical and quantum field theory that is a physical realization of the mathematical theory of derived geometry.  We extend earlier work \cite{EY1, EY2} in which we studied the derived algebraic geometry of twists of 4d $\mc N=4$ supersymmetric gauge theory; namely, we describe the geometric quantization of the derived stack assigned to the 4-manifold $\CC \times C$ for each twist of $\mc N=4$ super Yang--Mills theory, where $C$ is a compact Riemann surface. We argue that for certain twists this procedure produces the categories $\text{D}_\kappa(\bun_G(C))$ of twisted D-modules appearing in the quantum geometric Langlands correspondence.  In particular, we discuss the meaning of the parameter $\kappa$ in terms of derived algebraic geometry.  It is worth noting that the parametrization we obtain is different from the canonical parameter $\Psi$ appearing in the original work of Kapustin and Witten \cite{KapustinWitten}, and takes a much simpler form. In particular, our analysis does not incorporate the coupling constant in the Yang--Mills action functional.

The connections between the geometric Langlands program and supersymmetric gauge theory in four dimensions originated in work of Kapustin and Witten \cite{KapustinWitten}.  They considered a family of topological twists of 4d $\mc N=4$ super Yang--Mills theory parameterized by a value $\Psi \in \bb{CP}^1$ called the canonical parameter (depending both on the complexified coupling constant of the theory and on the choice of twisting supercharge).  There is a duality of quantum field theories called S-duality that relates the $\mc N=4$ theories with gauge group $G$ and its Langlands dual group $G^\vee$; this descends to a duality between the topologically twisted theories with canonical parameter $\Psi$ and $-1/r\Psi$, where $r$ is the lacing number of $G$.

The connection to the geometric Langlands correspondence comes when we consider the action of S-duality on certain boundary conditions in these 4d topological field theories. That is, if we compactify the 4d field theories along a compact Riemann surface $C$, then we obtain twists of an $\mc N=(2,2)$ supersymmetric sigma model whose target is the Hitchin moduli space on $C$ (in fact, it canonically admits the action of the $\mc N=(4,4)$ supersymmetric algebra). Kapustin and Witten argue that along appropriate loci in the $\bb{CP}^1$ family these twists yield either A- or B-models into the Hitchin moduli space with particular symplectic or complex structures. Then they go on to describe the categories of branes in these 2d topological field theories, and observe that versions of the categories appearing in the geometric Langlands correspondence are associated to S-dual theories.

A significant amount of Kapustin and Witten's physical analysis is devoted to explaining how D-modules occur as the categories of branes in an A-twisted theory, by identifying categories of A-branes as acted on by the algebra of open strings ending on an object called the canonical coisotropic brane.  Kapustin \cite{Kapustinnote} gave an alternative analysis explaining the occurrence of twisted D-modules.  In this paper we follow a different perspective, using ideas from derived algebraic geometry: we give a classical description of all twists of 4d $\mc N=4$ theories using the language of derived algebraic geometry and then explain why twisted D-modules occur when we geometrically quantize these derived stacks.

\subsection{Quantum Geometric Langlands}
Let us begin by briefly reviewing the statement of the categorical geometric Langlands correspondence and its quantum version.\footnote{We should remark that the ordinary geometric Langlands correspondence, the quantum geometric Langlands correspondence and the ``classical limit'' of the geometric Langlands correspondence all arise from similar constructions from the point of view of quantum field theory. In this paper we follow the usual mathematicians's convention for naming these correspondences.}  The quantum geometric Langlands correspondence was originally proposed by Beilinson and Drinfeld, and first appeared in the literature in its modern form around 15 years ago \cite{FrenkelLectures, Stoyanovsky}.  We fix a compact Riemann surface $C$ and a reductive algebraic group $G$ with Langlands dual group $G^\vee$.  We also fix a non-degenerate invariant pairing $\kappa$ on the Lie algebra $\gg$ of $G$. Let us assume for simplicity of notation that $G$ is a connected and simply-connected simple group and $\kappa$ is a complex multiple of the Killing form on $\gg$.  In this case we can identify $\kappa$ with this complex number. 

Let $\LL$ denote the determinant line bundle on the stack $\bun_G(C)$ of principal $G$-bundles on $C$.  Recall that we can define a sheaf $\mathcal{D}_\kappa$ of twisted differential operators associated to an arbitrary complex power $\LL^{\otimes \kappa}$ of $\LL$ (the sheaf still makes sense even when the line bundle itself does not exist) \cite{BBJantzen}.  Write $\text{D}_\kappa(\bun_G(C))$ for the category of modules for the sheaf of twisted differential operators on $\LL^{\otimes \kappa  } \otimes K_{\bun_G(C)}^{1/2}$ where $K_{\bun_G(C)}^{1/2}$ is the unique square root of the canonical bundle $K_{\bun_G(C)}$ over $\bun_G(C)$ under the assumptions we have made on the group $G$. \footnote{To be more careful, we should note that the category of twisted D-modules on $\bun_G(C)$ is defined as the limit over a filtration of $\bun_G(C)$ by quasi-compact substacks, and there are two different ways of doing this depending on whether one uses $*$ or $!$ push-forwards between the strata.  This distinction will not be important for our purposes.} 

\begin{remark}
There is an alternative point of view on twisted D-modules presented in the work of Gaitsgory and Rozenblyum \cite{GRCrystals}.  Rather than defining the sheaf of twisted differential operators directly, one can identify twisted D-modules on $X$ in terms of the data of a trivialized $\bb G_m$ gerbe equipped with a flat connection, called a twisting datum. Since usual D-modules are realized as quasi-coherent sheaves on a de Rham stack $X_{\mr{dR}}$, twisted D-modules can be realized in terms of such a category twisted by a twisting datum  (see Section \ref{gerbe_section}).
\end{remark}

The \emph{quantum geometric Langlands correspondence} gives a description of this category in terms of the Langlands dual group.  We will describe two versions of the correspondence, one for non-zero $\kappa$, and one for $\kappa = 0$.

\begin{conjecture}[Quantum Geometric Langlands Correspondence]
There is an equivalence of DG categories
\[\text{D}_\kappa(\bun_G(C)) \iso \text{D}_{-\kappa^\vee}(\bun_{G^\vee}(C)),\]
where $\kappa^\vee$ is the invariant pairing on $\gg^\vee$ such that $\kappa^\vee - \frac 12 \kappa^\vee_{\mr{Kil}}$ is dual to $\kappa - \frac 12 \kappa_{\mr{Kil}}$.  Here $\kappa_{\mr{Kil}}$ and $\kappa_{\mr{Kil}}^\vee $ are the Killing forms on $\gg$ and $\gg^\vee$ respectively. 
\end{conjecture}
This conjecture has been verified in the case where $G$ is abelian in work of Polishchuk and Rothstein \cite{PolishchukRothstein}.

\begin{remark}
The ``critical shift'', where we subtract half the Killing form when defining the dual will, from our perspective, be related to the metaplectic correction in categorical geometric quantization, as in Remark \ref{metaplectic_remark}.
\end{remark}

The correct statement of the correspondence for $\kappa=0$ requires one to work out what the dual category should be at ``$\kappa^\vee = \infty$''.  Versions of this conjecture were developed by Beilinson and Drinfeld prior to the statement of the quantum conjecture described above, although it turns out that a correct statement involves some subtleties not present for generic $\kappa$.
\begin{conjecture}[Categorical Geometric Langlands Correspondence \cite{ArinkinGaitsgory}]
There is an equivalence of DG categories
\[\text{D}(\bun_G(C)) \iso \mr{IndCoh}_{\mc N}(\flat_{G^\vee}(C)),\]
where $\flat_{G^\vee}(C)$ is the stack of $G^\vee$-bundles on $C$ equipped with a flat connection, and $\mr{IndCoh}_{\mc N}$ indicates the category of ind-coherent sheaves with nilpotent singular support.
\end{conjecture}

The theory of ind-coherent sheaves with singular support for sufficiently general stacks was developed by Arinkin and Gaitsgory. For a discussion of the physical meaning of the nilpotent singular support condition, see \cite{EY2}.

\begin{remark}
Let us remark upon a few interesting subtleties and extensions of the conjectures above, which we won't discuss in this paper, but which certainly will admit interesting analogues in the world of supersymmetric gauge theory.  Instead of considering only the family of twisted D-modules depending on a single complex parameter (multiples of the Killing form), we could study the action of duality on the full family of twists, parameterized by $(\bb A^1)^n \times \sym^2(\mf z^*)$, where $n$ is the number of simple factors of $\gg$ and $\mf z$ is its center.  One can additionally include in the twisting datum an extension of the canonical sheaf $\omega_C$ by $\OO_C \otimes \mf z^*$.  The meaning of this full twisting datum in the geometric Langlands program has been studied by Zhao \cite{Zhao1, Zhao2}.  This twisting family is extended further in the metaplectic Langlands correspondence of Gaitsgory and Lysenko \cite{GaitsgoryLysenko}.  We will not consider these extensions in this paper, though it would be interesting to investigate the relationship between this and theories including a discrete $\theta$-angle, as in the work of Frenkel and Gaiotto \cite{FrenkelGaiotto}.
\end{remark}

\subsection{Review of Previous Work}
This work is part of a larger project \cite{EY1, EY2} using derived algebraic geometry and the BV formalism to connect the geometric Langlands program and the topological twists of $\mc N=4$ super Yang--Mills theory, after the work of Kapustin and Witten \cite{KapustinWitten}.  In this section we will recall the main results from our previous work, and then in the next section we will explain how this paper extends those earlier results.

We begin with the idea of \emph{twisting} supersymmetric field theories, as introduced by Witten \cite{WittenTQFT}.  The basic idea is that if a field theory admits an odd symmetry $Q$ such that $Q^2 = 0$, we can study the $Q$-cohomology of the algebra of observables.  This cohomology is itself realized as the algebra of observables in a theory with a modified action functional, where we disregard $Q$-exact terms.  The Batalin--Vilkovisky formalism lends itself to the mathematical analysis of this concept, being a description of field theory in the language of homological algebra \cite{ElliottSafronov}.  

In brief, in the perturbative BV formalism, one gives a description for a classical field theory using a cochain complex equipped with some additional structure -- a shifted symplectic pairing and a homotopy Lie structure -- encoding the information of the theory's action functional.  One can then twist a classical field theory equipped with an odd square-zero symmetry $Q$ by adding $Q$ to the differential.  In \cite{EY1} we explained how to extend this from the perturbative level (the formal neighborhood of a classical solution to the equations of motion) to the non-perturbative level (the full moduli stack of classical solutions).  We considered theories whose fields included an algebraic gauge field -- a principal $G$-bundle on $\bb{R}^4$ that can be equipped with a complex algebraic structure upon choice of a complex structure of $\bb{R}^4$ -- in addition to other fields which are formal, for instance because they have odd degree.  If the action of the symmetry $Q$ leaves the gauge field fixed then one can compute the twist by $Q$ as a stack by twisting the fibers over each point in the moduli stack of $G$-bundles. 

We applied these ideas using the twistor formalism for $\mc N=4$ super Yang--Mills theory.  We will review this approach in Section \ref{twistor_section}.  The main advantage of this perspective is that it allows us to consider a moduli space with a natural algebraic structure, although it also has a bunch of non-algebraic objects before twisting.  The simplest non-zero twist of $\mc N=4$ super Yang--Mills theory is referred to as the ``holomorphic twist''; it is obtained by twisting by a non-trivial supercharge $Q_{\mr{hol}}$ of smallest possible rank and it removes all the non-algebraic information to yield the following:
 
\begin{theorem}[{\cite[Theorem 4.2, Lemma 4.7]{EY1}}] \label{EY1_holo_thm} 
The holomorphic twist of 4d $\mc N=4$ super Yang--Mills theory is defined on a complex algebraic surface $X$, and has the classical moduli space of solutions
\[\mr{EOM}_{\mr{hol}}(X) \iso T[1]\higgs_G(X),\]
where $\higgs_G(X)$ is the stack of $G$-Higgs bundles on $X$; that is, of pairs $(P,\phi)$, where $P$ is an algebraic $G$-bundle and $\phi$ is a section of the sheaf $K_X \otimes \gg^*_P$ such that $[\phi \wedge \phi]=0$.\footnote{To state this a little more precisely we should instead consider the formal completion of $\bun_G(X)$ in $T[1]\higgs_G(X)$.}
\end{theorem}

We will give a new proof of this result, at the perturbative level, in Section \ref{holo_section}.  The holomorphic twist has a natural two-parameter family of further deformations. To describe these further deformations, we will introduce the notion of the \emph{$\mu$-de Rham stack} $\mc X_{\mu\text{-dR}}$ of a stack $\mc X$.  This is a 1-parameter deformation of $T_f[1]\mc X$ \footnote{Here and later $T_f[1] \mc X$ refers to the formal completion of the zero section in $T[1]\mc X$.} called the \emph{Hodge deformation} -- the tangent space at a point $x \in \mc X$ to $\mc X_{\mu\text{-dR}}$ takes the form $\bb T_{\mc X,x}[1] \overset \mu \to \bb T_{\mc X,x}$.
\begin{definition}
The moduli stack of \emph{flat $\lambda$-connections} $\flat^\lambda_G(X)$ is the mapping stack $\mr{Map}(X_{\lambda\text{-dR}}, BG)$.  If $\lambda = 0$, this is the moduli stack of Higgs bundles described above.  If $\lambda = 1$ this is the moduli stack of $G$-bundles with flat connection.
\end{definition}

There is a two-dimensional family of twists further deforming the holomorphic twist, generated by a pair $Q_A,Q_B$ of supercharges.  These are supercharges in the cohomology of the supersymmetry algebra with respect to the operator $[Q_{\mr{hol}}, -]$.  The further twists by a linear combination of $Q_A$ and $Q_B$ can be defined on any complex surface $X$, and we showed that they have the following description.

\begin{theorem}[{\cite[Theorem 4.9, Theorem 4.21]{EY1}}] \label{EY1_further_thm}
The two-parameter family of further twists of the holomorphic twist generated by $Q_A$ and $Q_B$ coincides with the two parameter family of stacks
\[\mr{EOM}_{(\lambda, \mu)}(X) \iso \flat^\lambda_G(X)_{\mu\text{-dR}}.\]
\end{theorem}

\begin{remark}
The statement of \cite[Theorem 4.9]{EY1} is not quite correct although the mathematical content of the theorem is intact. Indeed, while the deformation at the point $(\lambda,\mu) = (1,0)$ in the family above coincides with a twist in the two parameter family spanned by $Q_A$ and $Q_B$, it doesn't correspond to the point we called $Q_B$, as we will explain in the next section and spell out in detail in Section \ref{22_section}.
\end{remark}

\subsection{Goals of this Paper}
In this paper, we extend our previous work in two ways.

\begin{enumerate}
 \item First, we describe the precise relationship between the space of square-zero supertranslations in the $Q_{\mr{hol}}$-cohomology of the $\mc N=4$ supersymmetry algebra and the space of further deformations of the holomorphic twist of the $\mc N=4$ field theory.  As described in Theorem \ref{EY1_holo_thm} above, this holomorphically twisted theory can be identified, perturbatively on $\mathbb{C}^2$, as the complex
 \[\Omega^{\bullet,\bullet}(\mathbb{C}^2, \gg[\eps]) \iso \Omega^{0,\bullet}(\mathbb{C}^2, \gg[\eps, \eps_1, \eps_2]),\]
 where $\eps, \eps_1, \eps_2$ are degree 1 parameters.  There is a three-dimensional space of deformations, corresponding to the three vector fields $\dd_\eps, \dd_{\eps_1}, \dd_{\eps_2}$.  There is a two-dimensional subspace consisting of ``B-type'' deformations spanned by $\dd_{\eps_1}$ and $\dd_{\eps_2}$.  The complement of this plane consists of ``A-type'' deformations, where the complex becomes contractible.
 
 When we identify these deformations as arising from supersymmetric twists however, there are some subtleties.
 \begin{enumerate}
    \item As we discussed in \cite[Section 2.1.1]{EY1}, the two-dimensional family of twists spanned by the A- and B-deformations must differ from Kapustin--Witten's two-parameter family of twists.  Their family does \emph{not} entirely consist of further twists of any single holomorphic supertranslation.  We will identify the twisted theories associated to a larger family of supertranslations, and explain how our family from our previous work, as well as Kapustin and Witten's family, sit inside it.
    
    \item The space of square-zero supertranslations in the $Q_{\mr{hol}}$-cohomology of the $\mc N=4$ supersymmetry algebra and the space of deformations of the $Q_{\mr{hol}}$-twisted theory are both three-dimensional.  There is a map from the former to the latter, but -- somewhat unintuitively -- this map is not linear, but instead is a quadratic $L_\infty$ map.  We identify this map explicitly in Section \ref{22_section}.
 \end{enumerate}
 
 \item Next, in Section \ref{QGL_section}, which is meant to informally suggest a new proposal for the physics of quantum geometric Langlands correspondence, we explain how to go from the derived moduli spaces associated to the $\mc N=4$ topological twists on complex algebraic surfaces of the form $\CC \times C$, to the categories of twisted D-modules occurring in the quantum geometric Langlands correspondence.  This procedure involves the idea of ``categorified geometric quantization'' (considered by \cite{Wallbridge} and Safronov \cite{SafronovGQ}).  We emphasize that this ansatz is not fully developed, and therefore the construction we describe is still somewhat ad hoc, rather than an application of a systematic method.  The idea, in brief, is as follows.  Starting with a 1-shifted symplectic derived stack $(\mc X, \omega)$ we specify two additional pieces of data: 
 \begin{enumerate}
  \item A \emph{prequantization}: a $\bb G_m$-gerbe $\mc G_{\mc X}$ on $\mc X$ with connection, chosen so that its curvature 2-form coincides with $\omega$.
  \item A \emph{polarization}: A 1-shifted Lagrangian fibration $\pi \colon \mc X \to \mc Y$. We will assume that $\mc G_{\mc X} \iso \pi^*\mc G_{\mc Y}$ for some $\bb G_m$ gerbe $\mc G_{\mc Y}$ on $\mc Y$.
 \end{enumerate}
 The ansatz of categorified geometric quantization says, by analogy with the process of geometric quantization for ordinary symplectic phase spaces, that we should quantize $\mc X$ by considering the category $\QC_{\mc G_{\mc Y}}(\mc Y)$ of quasi-coherent sheaves on $\mc Y$ twisted by the gerbe $\mc G_{\mc Y}$. 
 
 We can apply this ansatz to the moduli stacks assigned to the complex surface $\bb C \times C$, where $C$ is a closed algebraic curve, by the various twists of $\mc N=4$ super Yang--Mills.  These moduli stacks come equipped with 1-shifted symplectic forms, and we can run the geometric quantization procedure as discussed above.  The main result is the following.
 \begin{theorem}[see Section \ref{subsub:quantization examples}]
 Let $\flat^\lambda_G(C)_{\mu\text{-dR}}$ be the moduli stack of solutions to the equations of motion in a twisted theory on the algebraic surface $\bb C \times C$. The natural categorified geometric quantization of $\flat^\lambda_G(C)_{\mu\text{-dR}}$ is of the form $\text{D}_\kappa(\bun_G(C))$ where $\kappa=\lambda/\mu$.
 \end{theorem}


\end{enumerate}

\begin{remark}
The current paper discusses a part of the proposal of Kapustin and Witten in terms of derived algebraic geometry. However, a coupling constant doesn't play a role in our analysis. This can be heuristically explained by the physics fact as follows: all our discussion involving categories factors through a holomorphic-topological twist of 4-dimensional $\mc N=4$ theory and such a theory is realized by compactifying a holomorphic-topological twist of the 6-dimensional $\mc N=(2,0)$ theory along a topological torus which explains that the coupling constant $\tau$ should also be topological.

Indeed, in a recent paper by the second author with Surya Raghavendran \cite{RaghavendranYoo}, a different proposal was made for the realization of S-duality in type IIB string theory only \emph{after} performing a twist in the sense of \cite{CostelloLiSUGRA}.  Because the twisted S-duality procedure is defined only after performing a twist (by a supercharge in the 10-dimensional supersymmetry algebra) in which the spacetime directions responsible for the coupling constant $\tau$ become topological, it is insensitive to the value of $\tau$. In Remark \ref{comparison with KW}, we explain how the results of the current paper coupled with the framework of twisted S-duality would recover several known conjectures of equivalences of categories.
\end{remark}

\subsection{Conventions}
If $E$ is a graded vector bundle on a manifold $M$, we write $\mc E$ for its sheaf of smooth sections, and $\mc E_c$ for its cosheaf of compactly supported smooth sections.  We will sometimes refer to the space $\Oloc(\mc E)$ of \emph{local functionals} on $\mc E$.  This is the topological vector space consisting of those functionals which are given by integrating a density that varies polynomially in the fields and their derivatives (see \cite[Section 4.5.1]{CostelloGwilliam2} for a precise definition).

A \emph{local $L_\infty$ algebra} on a manifold $M$ is a $\ZZ \times \ZZ/2\ZZ$-graded vector bundle $E$ where $\mc E$ has the structure of a sheaf of $L_\infty$ algebras in the cohomological $\ZZ$ degrees such that the $L_\infty$ operations are even for the $\ZZ/2\ZZ$-grading, and are given by polydifferential operators $\mc E^{\otimes k} \to \mc E$.

\subsection*{Acknowledgements}
We would like to thank Owen Gwilliam, Justin Hilburn, Pavel Safronov, and Brian Williams for useful conversations during the preparation of this paper.  This research was supported in part by Perimeter Institute for Theoretical Physics. Research at Perimeter Institute is supported by the Government of Canada through Industry Canada and by the Province of Ontario through the Ministry of Economic Development \& Innovation. This work was also supported by the New Faculty Startup Fund from Seoul National University, the National Research Foundation of Korea (NRF) grant funded by the Korea government(MSIT) (No. 2022R1F1A107114212), and the LAMP Program of the National Research Foundation of Korea (NRF) grant funded by the Ministry of Education (No. RS-2023-00301976).

\section{\texorpdfstring{$\mc N=4$}{N=4} Super Yang--Mills Theory}
In this section we will construct an action of the $\mc N=4$ super translation Lie algebra on (complexified) self-dual $\mc N=4$ super Yang--Mills theory on $\RR^4$.  The action of a supercharge commuting with a holomorphic supercharge $Q_\mr{hol}$ induces a vector field on the holomorphically twisted moduli space: we classify all such supercharges and, in the next section, the corresponding vector fields.  Given such a vector field we can calculate the further twist of the moduli space, which we do in all cases.  Included in this family is a subfamily studied in \cite{EY1} parameterized by $\bb{CP}^1$, which induces the moduli spaces occuring in the quantum geometric Langlands program.  This family is closely related to, but does not coincide with, the $\bb{CP}^1$ family of Kapustin--Witten \cite{KapustinWitten}.

\subsection{Twisting Classical Field Theory}
We begin with some general discussion on the notion of supersymmetric twisting for classical field theory.  Our model for classical field theory is based on the classical Batalin--Vilkovisky (BV) formalism \cite{BatalinVilkovisky}.  The starting point of the BV formalism can be viewed as coming from derived geometry: given a space $\mc F$ of fields and an action functional $S$, we would like to study the solutions to the equations of motion.  That is, the space of elements $\phi \in \mc F$ satisfying $\d S(\phi) = 0$.  We can view this as the intersection $\Gamma_{\d S} \cap_{T^* \mc F} \mc F$, where $\Gamma_{\d S}$ is the graph of the derivative $\d S$.  In the classical BV formalism we use homological algebra to model this intersection in a derived sense, i.e. we model the derived intersection $\Gamma_{\d S} \cap^h_{T^* \mc F} \mc F$. We refer to \cite{CostelloBook, CostelloGwilliam2} for more details.

In this paper we will mainly consider the perturbative version of this formalism, meaning that we will try to model not the entirety of this derived intersection, but only the formal neighborhood of a single point.  This formal neighborhood will be described as a formal moduli problem, or equivalently (using the nexus of ideas developed by Pridham \cite{PridhamFMP}, Lurie \cite{DAGX}, To\"en \cite{Toen} and others) as an $L_\infty$-algebra.

\begin{definition}
A \emph{perturbative classical BV theory} on a manifold $M$ is a $\ZZ \times \ZZ/2\ZZ$-graded vector bundle $E$ with sheaf of sections $\mc E$, equipped with a local $L_\infty$-algebra structure on $\mc E[-1]$, and a $(-1)$-shifted symplectic pairing $\omega \colon \mc E \otimes  \mc E \to \dens_M[-1]$ where $\dens_M$ is the density line bundle on $M$.
\end{definition}

We can equivalently realize a classical BV theory as a graded vector bundle $E$ equipped with a $(-1)$-shifted symplectic pairing $\omega$, and a local functional $S \in \Oloc(\mc E)$ satisfying the classical master equation:
\[\{S,S\} = 0,\]
where the bracket here is the \emph{BV antibracket} acting on local functionals in the field theory (see, for instance, \cite[Section 1.2]{ElliottSafronovWilliams} for a construction). The $k^\text{th}$ Taylor coefficient of the functional $S$ can be identified with the order $k-1$ bracket $\ell_{k-1}$ of the $L_\infty$-algebra according to the formula
\[S(\alpha) = \sum_{k \ge 2} \frac 1{k!} \omega(\alpha, \ell_{k-1}(\alpha^{\otimes k-1})).\]

We can define the action of a super Lie group on a classical BV theory in the following way.  We can think of the following definition as introducing global background fields valued in the Lie algebra of a group $G$ of symmetries, and extending the action functional to depend on these background fields.
\begin{definition}
Let $G$ be a super Lie group acting smoothly on a manifold $M$. An \emph{action of $G$ on a classical BV theory $\mc E$} on $M$ consists of the following data:
\begin{enumerate}
 \item An action of $G$ on the DG vector bundle $E$, respecting the symplectic pairing and the $L_\infty$ structure;
 \item An infinitesimal inner action of the Lie algebra $\gg$ on $\mc E$.  That is, a linear map $S_\gg \colon \sym^{\ge 1}(\gg[1]) \to \Oloc(\mc E)$ so that $S + S_\gg$ satisfies the classical master equation. One has the degree $k$ term $S_\gg^{(k)}\colon \sym^k(\mf g[1])\to \Oloc(\mc E)$ for each $k\geq 1$.
\end{enumerate}
These data should together satisfy the condition that the linear term $S^{(1)}_\gg$ in $S_\gg$ coincides with the derivative of the induced $G$-action on $\Oloc(\mc E)$.
\end{definition}

Now, let's begin to study supersymmetric field theories.  These are field theories equipped with an action of a \emph{supersymmetry group}, i.e. a supersymmetric extension of the Poincar\'e group of isometries of a pseudo-Riemannian vector space $\RR^{p,q}$, where $p+q=n$.  It will be enough for us, in order to study twists, to restrict attention to the action of translations and supertranslations.

\begin{definition}
Let $\Sigma$ be a real (resp. complex) spinorial representation of $\Spin(p,q)$, and fix a nondegenerate spin-equivariant $\Gamma \colon \Sigma \otimes \Sigma \to \RR^{n}$ (resp $\CC^n$), where $n = p+q$.  The \emph{real} (resp. \emph{complex}) \emph{supertranslation algebra} associated to $\Sigma$ and $\Gamma$ is the super Lie algebra $\RR^n \oplus \Pi \Sigma$ (resp. $\CC^n \oplus \Pi \Sigma$), with the only non-trivial bracket given by $\Gamma$.
\end{definition}

\begin{definition}
The \emph{4d $\mc N=4$ supertranslation algebra} $\mf T_{\mc N=4}$ is the complex supertranslation algebra in dimension $n=4$ with $\Sigma =  S_+\otimes W \oplus  S_-\otimes  W^*$, where $W$ is a 4 dimensional complex vector space, and $S_\pm$ are the positive and negative helicity Weyl spinor representations.  That is, under $\spin(4) \iso \SU(2) \times \SU(2)$, $S_\pm$ are the defining representations of the two factors. The pairing $\Gamma$ is defined by identifying the complex vector representation of $\spin(4)$ with the tensor product $S_+ \otimes S_-$, and taking the composite
\[(S_+ \otimes W) \otimes (S_- \otimes W^*) \to S_+ \otimes S_- \iso \CC^4.\]
\end{definition}

From now on we will restrict attention to this example, so to classical BV theories on $\RR^4$ equipped with an action of the supertranslation algebra $\mf T_{\mc N=4}$.  We'll be interested in twisting theories of this form.  The idea is that we modify the differential on the classical BV complex $\mc E$ using an odd element $Q$ of the supertranslation algebra, which squares to zero.  We refer to \cite{ElliottSafronovWilliams} for a more detailed discussion.

\begin{definition}
The \emph{twist} of a supersymmetric classical BV theory $\mc E$ by an odd element $Q \in  \mf  T_{\mc N=4}$ satisfying $\Gamma(Q,Q)=0$ is the $\ZZ/2\ZZ$-graded classical field theory $\mc E^Q$ defined by adding the term $S_\gg(Q):=\sum_{k\geq 1 } S_{\gg }^{(k)}(Q,\cdots, Q)$ to the action functional. 
\end{definition}

We can promote the $\ZZ/2\ZZ$-graded theory to a $\ZZ$-graded theory by choosing an additional piece of data.
\begin{definition} \label{twisting_datum_def}
A \emph{twisting datum} is a pair $(\alpha, Q)$, where $Q$ is a square-zero supertranslation as above, and $\alpha$ is an action of $\mr U(1)$ on $\mc E$ so that $Q$ has weight one.  Given a twisting datum, the corresponding $\ZZ$-graded twisted theory is defined by taking the theory $\mc E$ and adding the $\alpha$-weight to the cohomological grading. With respect to the new grading the term $S_\gg(Q)$ has cohomological degree zero, so the $\ZZ/2\ZZ$-graded twisted theory $\mc E^Q$ is promoted to a $\ZZ $-graded theory.
\end{definition}

\begin{remark}
One can also study twisting in the non-perturbative setting, on the level of the derived stack of solutions to the equations of motion.  In \cite{EY1} we discussed twists of $\mc N=4$ super Yang--Mills from this point of view, by considering theories where the moduli space was a derived thickening of a fixed base space, left unchanged by the action of the twisting data.  In this setting we were able to discuss the global structure of twists by twisting the fibers over this fixed base.  In this paper we'll show how the twists we studied there match up with twists associated to the action of specific families of supercharges in the $\mc N=4$ supertranslation algebra.
\end{remark}

\subsection{\texorpdfstring{$\mc N=4$}{N=4} Super Yang--Mills Theory and Twistors} \label{twistor_section}

The theories of our interest are supersymmetric twists of $\mc N=4$ super Yang--Mills theory on $\RR^4$. This theory can be obtained in two very different ways.  Either one can first define a maximally supersymmetric Yang--Mills theory on $\RR^{10}$, with $\mc N=(1,0)$ supersymmetry, and dimensionally reduce it along $\RR^6$, or one can define a supertranslation invariant field theory on a complex 3-fold -- twistor space -- and dimensionally reduce it along a $\bb{CP}^1$-fibration. In this paper, we will construct twists using the second approach after briefly reviewing the first approach.

For the first approach, $\mc N=4$ super Yang--Mills theory on $\RR^4$ is the theory obtained by dimensionally reducing $\mc N=(1,0)$ super Yang--Mills theory on $\RR^{10}$ along a projection $\RR^{10} \to \RR^4$.  This theory degenerates to the theory of \emph{self-dual} $\mc N=4$ super Yang--Mills theory in the following way.  One can describe 4d Yang--Mills theory with fields including a $\gg$-valued 1-form $A$ and a $\gg$-valued 2-form $B$, with action functional including terms of the form $\int \langle \d_+ A \wedge B \rangle + g\langle B \wedge B \rangle$, for $g$ a positive constant.  This is the \emph{first-order formalism} of Yang--Mills theory, as explained -- in the BV formalism -- in \cite{CostelloBook, ElliottWilliamsYoo}.  The degeneration to self-dual Yang--Mills theory entails subtracting the $\langle B \wedge B \rangle$ term from the action functional, or setting $g=0$.  

For the second approach, we will review the construction of self-dual $\mc N=4$ super Yang--Mills theory on $\RR^4$ via dimensional reduction from twistor space, due originally to Witten \cite{WittenTwistor} and Boels, Mason and Skinner \cite{BMS}.  For more details on our perspective we refer the reader to the discussion in \cite{CostelloSH} and \cite[Section 3.2]{EY1}.

\begin{definition}\label{supertwistor_def}
Let $W$ be the 4-dimensional auxiliary space in the $\mc N=4$ supersymmetry algebra.  The \emph{$\mc N=4$ super twistor space} $\bb{PT}^{\mc N=4}$ is the super complex variety defined to be the total space of the super vector bundle $\OO(1) \otimes \Pi W \to \bb{CP}^3$.  The \emph{twistor morphism} associated to this super twistor space is the smooth map
\[p' \colon \bb{PT}^{\mc N=4} \to \bb{HP}^1 \iso S^4,\]
given by composing the projection with the canonical map $\bb{CP}^3 \to \bb{HP}^1$.  When we restrict to the complement of a point $\infty \in S^4$, we obtain a smooth map
\[p \colon \bb{PT}^{\mc N=4} \bs \bb{CP}^1 \to \RR^4,\]
where $\bb{CP}^1$ is the (even part of) the fiber of $p'$ over $\infty$.
\end{definition}

\begin{remark} \label{twistor_CY_remark}
In other words, $\bb{PT}^{\mc N=4}$ is just the super projective space $\bb{CP}^{3|4}$.  This carries a super Calabi--Yau structure, since the Berezinian of $\bb{CP}^{n|s}$ is readily computed to be isomorphic to $\OO(s-n-1)$. 
\end{remark}

The following action is constructed in \cite[Lemma 14.8.1]{CostelloSH}.
\begin{lemma}
There is a natural action of the $\mc N=4$ supertranslation algebra $\mf T_{\mc N=4}$ on $\bb{PT}^{\mc N=4} \bs \bb{CP}^1$.
\end{lemma}
This action is constructed as follows.  
\begin{itemize}
 \item First, the ordinary translations act by translating the base of the fibration $p \colon \bb{PT}^{\mc N=4} \bs \bb{CP}^1 \to \RR^4$.
 \item There is also a canonical identification $\Pi S_+\otimes W \iso H^0(\bb P(S_+), \OO(1) \otimes \Pi W)$.  We can use this identification to define an action $\Pi S_+\otimes W $ on the total space of our vector bundle by translating the fiber by the section associated to an element of $\Pi S_+\otimes W $.
 \item As $\bb{PT}\setminus \bb{CP}^1= \bb{CP}^3 \setminus \bb{CP}^1$ is identified as the total space of $\OO(1)\otimes S_-\to \bb P(S_+)$, we can identify $\bb{PT}^{\mc N=4} \bs \bb{CP}^1$ with the total space $\OO(1) \otimes (S_- \oplus \Pi W) \to \bb P(S_+)$. There is an action of $\Pi S_-\otimes  W^* $ on this total space, by embedding $\Pi S_-\otimes  W^* \iso \hom(\Pi W, S_-)$ in $\eend(S_- \oplus \Pi W)$, and acting on the fibers.
\end{itemize}
One can then check that these three constituent parts of the action satisfy the appropriate commutation relations.

We will obtain $\mc N=4$ super Yang--Mills theory by dimensionally reducing an appropriate super translation invariant theory built using super twistor space.
\begin{definition}
\emph{Holomorphic Chern--Simons theory} on $\mc N=4$ super twistor space with Lie algebra $\gg$ is the classical BV theory with BV fields $\Omega^{0,\bullet}(\bb{PT}^{\mc N=4} \bs \bb{CP}^1) \otimes \gg[1]$, with bracket given by the combination of the super Calabi--Yau structure on super twistor space and the invariant pairing on $\gg$.
\end{definition}

The following theorem is proven on the level of cochain complexes in \cite[Proposition 3.17]{EY1}.  For the non-supersymmetric theories it was proven on the level of $L_\infty$ algebras by Movshev \cite{Movshev}.

\begin{theorem}
There is a supertranslation-equivariant equivalence of classical field theories between the following two theories: \vspace{-1em}
\begin{itemize}
 \item The dimensional reduction of holomorphic Chern--Simons theory along the projection $p$;
 \item Self-dual $\mc N=4$ super Yang--Mills theory on $\RR^4$.
\end{itemize}
\end{theorem}

\begin{remark}
We will not need to be too concerned with the difference between ordinary and self-dual $\mc N=4$ super Yang--Mills theory.  After taking any non-trivial twist, the difference between the two action functionals becomes $Q$-exact, and hence the two theories become equivalent.
\end{remark}

\begin{remark}
One feature of the super twistor space perspective, and one of the main topics of \cite{EY1}, is that it allows us to construct \emph{algebraic} structures on twists of super Yang--Mills theory.  Indeed, super twistor space, being a projective space, naturally has the structure of a complex algebraic variety.  We used this algebraic structure to realize the holomorphic twist of $\mc N=4$ super Yang--Mills and its further deformations as algebraic quantum field theories on $\CC^2$, or more generally on a smooth algebraic surface. This in particular means that the classical BV complex $\mc E$ can be described in a purely algebraic way, where the $L_\infty$ structure can also be given in terms of algebraic differential operators.
\end{remark}

\subsection{Description of the Supersymmetry Action} \label{susy_action_section}
Let's use the manifest supertranslation action on super twistor space to give an explicit description of the supertranslation action on the classical BV complex of self-dual $\mc N=4$ super Yang--Mills theory on $\RR^4$.  We begin with the 6d holomorphic Chern--Simons theory on super twistor space itself.  The supertranslation action here is inherited from the supertranslation action on super twistor space.  Thus, there are two different descriptions of the positive and negative sectors of the space of supertranslations.   We will use the notation $u_1, u_2$ and $\theta_1, \ldots, \theta_4$ for complex coordinates in the even and odd directions respectively of the fiber of the map $\pi \colon \bb{PT}^{\mc N=4} \setminus\bb{CP}^1 \to \bb P(S_+)$, namely, for $S_-$ and $\Pi W$, respectively.

\begin{itemize}
 \item First, given an element $Q_+ = \alpha \otimes w \in S_+ \otimes W$, it acts on super twistor space by a shear transformation.  In other words, $Q_+$ defines a section of $\OO(1) \otimes W$, and hence a vector field on the total space of the bundle In the coordinates above, this transformation can be written as the vector field $\alpha (w \cdot \dd_\theta)$, where $w \cdot \dd_\theta$ is a vector field on the fiber $\Pi W$. Hence, on holomorphic Chern--Simons theory, the supertranslation acts by adding this differential operator to the differential $\ol \dd$ on $\Omega^{0,\bullet}(\bb{PT}^{\mc N=4}\setminus\bb{CP}^1; \gg)$. 
 
 \item Now, given an element $Q_- = \beta  \otimes w^* \in S_- \otimes W^*$, it acts on super twistor space by a superspace rotation.  In the coordinates above, this transformation can be written as the vector field at the coordinate $(u,\theta)$ given by $\langle w^*, \theta \rangle (\beta  \cdot \dd_u)$, where $\langle -,- \rangle$ denotes the $\CC$-valued pairing between $W$ and $W^*$, and $\beta  \cdot \dd_u$ is the constant vector field on $S_-$ associated to the element $\beta $.  Therefore, on holomorphic Chern--Simons theory, the supertranslation acts by adding this differential operator to the differential $\ol \dd$ on $\Omega^{0,\bullet}(\bb{PT}^{\mc N=4}\setminus\bb{CP}^1; \gg)$. 
\end{itemize}

\subsection{Classification of Twists} \label{classification_section}
There are several possible twists of $\mc N=4$ super Yang--Mills theory.  We obtain a twisted theory for each square-zero supertranslation $Q$; within the space of square-zero supercharges there are a few different equivalence classes of twisted theory, corresponding to the orbits under the action of $\spin(4)$, and of the R-symmetry group $\SL(4;\CC)$.  This classification (in all dimensions) was analyzed in \cite{ElliottSafronov, ElliottSafronovWilliams, EagerSaberiWalcher}, and we refer there for more details.

An element $Q \in S_+ \otimes W \oplus S_- \otimes W^*$ is equivalent to a pair of homomorphisms $W^* \to S_+$ and $W \to S_-$.  We index such a pair of homomorphisms by their ranks, so by a pair of non-negative integers $(p,q)$, both at most 2.  The orbits in the space of square-zero supercharges admit the following classification:
\begin{itemize}
 \item Rank $(1,0)$ and $(0,1)$.  These twists are holomorphic.
 \item Rank $(2,0)$ and $(0,2)$.  These twists are topological.
 \item Rank $(1,1)$.  This twist is intermediate, with 3-invariant directions.
 \item Rank $(2,1)$ and $(1,2)$.  These twists are topological.
 \item Rank $(2,2)$.  Unlike the earlier examples, the space of square-zero supercharges of full rank decomposes further.  There is a $\CC^\times$-family of orbits indexed by a parameter $s$, which can be obtained as the ratio between the canonical volume form on the $\SL(4;\CC)$-representation $W$ and the volume form induced from the volume forms on $S_+$ and $S_-$ by the element $(Q_+, Q_-) \in S_+ \otimes W \oplus S_- \otimes W^*$, viewed as a short exact sequence
 \[0 \to S_+ \to W \to S_- \to 0.\]
\end{itemize} 

It is sometimes valuable to restrict attention to those supertranslations which are fixed by a twisting homomorphism from $\spin(4;\CC)$ to the R-symmetry group $\SL(4;\CC)$.  The idea is as follows.  The space of supertranslations is a module for $\spin(4;\CC) \times \SL(4;\CC)$.  Fix a subgroup $\iota \colon G \inj \spin(4;\CC)$ and a homomorphism $\phi \colon G \to \SL(4;\CC)$, called a \emph{twisting homomorphism}.

\begin{definition}
A supertranslation $Q$ is said to be \emph{compatible} with a twisting homomorphism $\phi$ if it is invariant for the action of the subgroup
\[ (\iota, \phi) \colon G \to \spin(4;\CC) \times \SL(4;\CC).\]
\end{definition}
If $Q$ is compatible with a twisting homomorphism $\phi$, then we can use the twisting homomorphism to define the $Q$-twisted theory not just on $\RR^4$, but on general $G$-structured manifolds.

Consider the \emph{Kapustin--Witten twisting homomorphism}:
\begin{align*}
\phi_{\mr{KW}} \colon \spin(4;\CC) \iso \SL(2;\CC)_+ \times \SL(2;\CC)_- &\to \SL(4;\CC) \\
(A,B) &\mapsto \mr{diag}(A,B).
\end{align*}
It's easy to check that there is a $\bb{CP}^1$-family of inequivalent supertranslations $Q$ compatible with the full Kapustin--Witten twisting homomorphism, consisting of the $\CC^\times$-family of rank $(2,2)$-twists, with the rank $(0,2)$ and $(2,0)$ twists at the poles.  

From now on we will fix a basis $\{\alpha_1, \alpha_2\}$ for $S_+$, a basis $\{\beta_1, \beta_2\}$ for $S_-$, and a basis $\{e_1,e_2,f_1,f_2\}$ for $W$.  Under the Kapustin--Witten twisting homomorphism, $\SL(2;\CC)_+$ will act on $e_1$ and $e_2$ by the defining representation, and likewise $\SL(2;\CC)_-$ will act on $f_1$ and $f_2$.  Linear combinations of the rank $(2,0)$ and $(0,2)$ supertranslations $\alpha_1 \otimes e_2 - \alpha_2 \otimes e_1$ and $\beta_1^ \otimes f_1^* - \beta_2 \otimes f_2^*$  will be left invariant by the twisted $\spin(4;\CC)$ action.

Now, more generally, we can consider the restriction of the Kapustin--Witten twisting homomorphism to the diagonal maximal torus $\CC^\times \times \CC^\times \subset \SL(2;\CC)_+\times \SL(2;\CC)_-$.

\begin{prop} \label{P1xP1_family_prop}
There is a $(\bb{CP}^1 \times \bb{CP}^1)/\CC^\times$-family of supertranslations compatible with this restricted twisting homomorphism, where the group $\CC^\times$ acts by simultaneously rescaling the two factors.
\end{prop}

More precisely, we are going to define a map of quotient stacks
\[(\bb{CP}^1 \times \bb{CP}^1)/\CC^\times \to (S_+ \otimes W \oplus S_- \otimes W^*)/(\Spin(4;\CC) \times \SL(4;\CC))\]
whose image factors through the substack of the target quotient stack on which the subgroup $(\CC^\times)^2 \sub \Spin(4;\CC) \times \SL(4;\CC)$ defined using the Kapustin--Witten twisting homomorphism acts trivially.

\begin{proof}
Under the $\CC^\times \times \CC^\times$-action induced from the Kapustin--Witten twisting homomorphism, the elements $e_1, e_2, f_1, f_2$ in $W$ have weights $(1,0), (-1,0), (0,1)$ and $(0,-1)$ respectively.  The dual elements $e_1^*, e_2^*, f_1^*, f_2^*$ in $W^*$ likewise have weights $(-1,0), (1,0), (0,-1)$ and $(0,1)$.  There is, therefore, a four-dimensional space of weight $(0,0)$ elements, spanned by the four elements
\[\alpha_1 \otimes e_2,\ \alpha_2 \otimes e_1,\ \beta_1 \otimes f_1^*, \text{ and } \beta_2\otimes f_2^*.\]

Let us denote this four-dimensional space by $U \sub (S_+ \otimes W \oplus S_- \otimes W^*)$.  We will now investigate the image of $U$ in the quotient stack $(S_+ \otimes W \oplus S_- \otimes W^*)/(\Spin(4;\CC) \times \SL(4;\CC))$.  The subgroup of $\Spin(4;\CC) \times \SL(4;\CC)$ that leaves $U$ fixed is just the normalizer of $(\CC^\times)^2 \sub \Spin(4;\CC) \times \SL(4;\CC)$, which is generated by two factors: the normalizer of $(\CC^\times)^2$ inside the Kapustin--Witten embedded copy of $\Spin(4;\CC)$, and the centralizer of $(\CC^\times)^2$ inside $\SL(4;\CC)$.

Let us first consider the second factor.  This centralizer is three dimensional.  It is generated by a copy of $\CC^\times$ acting on $U \iso \CC^4$ with weights $(0,1,1,0)$, a copy of $\CC^\times$ acting with weights $(1,0,0,1)$, and finally a copy of $\CC^\times$ acting with weights $(0,1,0,1)$.  Let us restrict to the locus where the first two factors act freely.  The quotient of this locus $(\CC^2 \bs  \{0\} )^2$ by the $(\CC^\times)^2$ action is exactly $\bb{CP}^1 \times \bb{CP}^1$ with homogeneous coordinates $(( \alpha_1 \otimes e_2 :  \beta_2 \otimes f_2^* ), (  \beta_1 \otimes f_1^*  : - \alpha_2 \otimes e_1))$, where we abuse the notation to denote coordinates using the corresponding basis. The coordinate is chosen in such a way that the third factor of $\CC^\times$ acts by simultaneous rescaling.
\end{proof}

\begin{remark}
Let us briefly discuss the remaining factor of the normalizer: the action of the normalizer of $(\CC^\times)^2$, which is an abelian subgroup of $\SL(2;\CC)^2$ containing  $(\CC^\times)^2$ as an index two subgroup.  This subgroup acts trivially on $U$ by definition.  The quotient (the Weyl group $(\ZZ/2\ZZ)^2$ of $\SL(2;\CC)^2$) does not preserve the locus considered above on which $(\CC^\times)^2$ acts freely.  A more refined picture of the moduli space of twisting supercharges would use this larger locus and keep track of this discrete symmetry, but we will not need this in what follows.
\end{remark}

In the current paper, we are going to study twists that can be defined on a spacetime manifold taking the form of a product of two Riemann surfaces.  Such theories arise when we consider twists by supercharges that are compatible with $\CC^\times \times \CC^\times \iso \spin(2;\CC) \times \spin(2;\CC) \sub \Spin(4;\CC)$, which is isomorphic to the diagonal maximal torus in $\SL(2;\CC) \times \SL(2;\CC)$.  For this reason, throughout the paper we are going to work with a fixed twisting datum $\alpha \colon \mr U(1)\to G_R= \SL(4;\CC)$ with respect to which $e_1,e_2$ have weight $1$ and $f_1,f_2$ have weight $-1$. Note that any linear combination of $\alpha_1 \otimes e_2$, $\alpha_2 \otimes e_1$, $\beta_1 \otimes f_1^*$, and $\beta_2 \otimes f_2^*$ has weight 1 with respect to $\alpha$. In what follows, a twist by a supercharge $Q$ will always be considered together with this choice of $\alpha$.

\begin{remark}
 Let us take a moment to comment on the spaces of twists that will occur in the actual computations that we will perform.  We will describe, particularly in Theorem \ref{22_theorem}, a family of twists of $\mc N=4$ super Yang--Mills theory parameterized by the affine space $\bb A^3$.  This family of twists will be $\CC^\times$ invariant, and will descend, after removing a line, to a family over $\bb A^1 \times \bb P^1 \sub \bb P^1 \times \bb P^1$.  Note that these computations depend on a choice of point in one of the two $\bb P^1$ factors (or equivalently, on an initial choice of holomorphic supercharge).
\end{remark}

\section{Twisting \texorpdfstring{$\mc N=4$}{N=4} Theory} \label{twist_section}
Having set up the necessary background, we can discuss the computation of the twist of $\mc N=4$ super Yang--Mills theory by the various twists described above.  Our approach will be to begin by computing the holomorphic twists by rank $(1,0)$ and $(0,1)$ supercharges and then their further twists in each case.  Twists of rank $(1,1)$ and higher can be obtained starting from either holomorphic point, and we will compare the two descriptions.

\subsection{Dimensional Reduction}
Let us first investigate the dimensional reduction of the action discussed in Section \ref{susy_action_section} to $\RR^4$.  We first identify the classical BV complex of our holomorphic Chern--Simons theory as follows.  

\begin{prop} \label{hCS_BV_complex_prop}
The pushforward of the classical BV complex of holomorphic Chern--Simons theory on super twistor space to $\RR^4$ is equivalent to the total complex of the triple complex
\[\bigoplus_{i,j,k} C^\infty_\CC(\RR^4) \otimes \Omega^{0,i}(\bb P(S_+); \wedge^j(\mc O(1)\otimes S_-) \otimes \sym^k(\Pi \mc O(-1)\otimes W^* ) \otimes \mf g )), \]
where the summands live in cohomological degree $i+j-1$ and fermionic degree $k$ (mod 2).  The differential is given by the sum of $\d_i = \ol \dd_{\bb P(S_+)}$ and $\d_j$, which is given by the operator $\sum_{a,b=1}^2 \frac{\dd }{\dd x_{ab}} \otimes (\alpha_a \otimes \beta_b)$, where $x_{11},x_{12},x_{21},x_{22}$ are (complexified) coordinates on $\RR^4$, and $\alpha_a \otimes \beta_b$ is considered as a section of $\OO(1) \otimes S_-$.  There is no additional differential in the $k$ direction.
\end{prop}

\begin{proof}
The BV complex for holomorphic Chern--Simons theory on $\bb{PT}^{\mc N=4}$ is given by \[\Omega^{0,\bullet}(\bb {PT}^{\mc N=4}\setminus\bb{CP}^1 ; \mf g ) \cong \Omega^{0,\bullet}(\bb{PT}\setminus \bb{CP}^1; \sym^\bullet( \Pi(\mc O(-1)\otimes W^* ) \otimes \mf g)\]
with the differential induced from the Dolbeault differential of $\bb {PT}^{\mc N=4}\setminus\bb{CP}^1$. Note that the Dolbeault forms are generated by functions and anti-holomorphic 1-forms. As $T^*_{\bb{PT}\setminus \bb{CP}^1} =\pi^*(\mc O(-2)\oplus \mc O(-1)\otimes S_-^* )$, where $\pi \colon \bb{PT}\setminus \bb{CP}^1 \to \bb P(S_+)$, one has
\[\Omega^{0,\bullet}(\bb{PT}\setminus \bb{CP}^1) \iso \Gamma(\bb{PT}\setminus \bb{CP}^1, \wedge^\bullet \ol{T^*_{\bb{PT}\setminus \bb{CP}^1}}) \iso  \Gamma( \bb{PT}\setminus \bb{CP}^1; \wedge^\bullet \pi^*( \ol{\mc O(-2)}\oplus \ol{ \mc O(-1)\otimes S_-^* } ) ). \]
Its pushforward to $\RR^4$ is 
\[C^\infty_\CC(\RR^4) \otimes\Gamma(\bb P(S_+) ; \wedge^\bullet(\ol{\mc O(-2)}\oplus \mc O(1)\otimes S_- ) )\iso C^\infty_{\CC}(\RR^4) \otimes \Omega^{0,\bullet}(\bb P(S_+) ; \wedge^\bullet(\mc O(1)\otimes S_-) ).\]
Hence the BV complex of interest is
\[\bigoplus_{i,j,k} C^\infty_\CC(\RR^4) \otimes \Omega^{0,i}(\bb P(S_+); \wedge^j(\mc O(1)\otimes S_-) \otimes \sym^k(\Pi \mc O(-1)\otimes W^* ) \otimes \mf g ) \]
where such a term is in cohomological degree $i+j-1$ and fermionic degree $k$ (mod 2). The Dolbeault differential has summands associated to coordinates on the base $\bb P(S_+)$ and global sections of the fibration $\pi \colon \bb{PT}\setminus \bb{CP}^1 \to \bb P(S_+)$.  The first summand $\d_i$ is just the Dolbeault differential on the base.  The second summand $\d_j$ is the sum of four terms associated to a basis for the space of holomorphic sections of the vector bundle $\OO(1) \otimes S_-$.
\end{proof}

Dimensional reduction to $\RR^4$ simply entails taking the $\d_i$-cohomology.  We can compute what happens to the differential $\d_j$ and the supertranslation action using the corresponding spectral sequence.

\begin{prop}[{\cite[Proposition 3.17]{EY1}}] \label{BV_SS_prop}
The pushforward of the classical BV complex of holomorphic Chern--Simons theory on super twistor space is quasi-isomorphic to the following $\ZZ \times \ZZ/2\ZZ$-graded complex by taking $\d_i$-cohomology and computing the further differentials acting on the $E_1$ page of the spectral sequence of the double complex:
\[\xymatrix{
\ul{-1} &\ul{0} &\ul{1} &\ul{2} \\
\Omega^0(\RR^4; \gg) \ar[r]^{\d} &\Omega^1(\RR^4; \gg) \ar[r]^{\d_+} &\Omega^2_+(\RR^4; \gg ) & \\
&\Omega^0(\RR^4; S_- \otimes \Pi W^* \otimes \gg) \ar[r]^{\sd {\d}} & \Omega^0(S_+ \otimes \Pi W^* \otimes \gg) & \\
&\Omega^0(\RR^4; \wedge^2 W^* \otimes \gg) \ar[r]^{\Delta} &\Omega^0(\RR^4; \wedge^2 W^* \otimes \gg) & \\
&\Omega^0(\RR^4; S_+ \otimes \Pi \wedge^3 W^* \otimes \gg) \ar[r]^{\sd{\d}} &\Omega^0(\RR^4; S_- \otimes \Pi \wedge^3 W^* \otimes \gg) & \\
&\Omega^2_+(\RR^4; \wedge^4 W^* \otimes\gg) \ar[r]^{\d} &\Omega^3(\RR^4; \wedge^4 W^* \otimes\gg) \ar[r]^{\d} &\Omega^4(\RR^4; \wedge^4 W^* \otimes\gg).
}\]
\end{prop}
For future reference, we note the relevant cohomology computation in the following table:
\begin{center}
\begin{tabular}{|c|c|c|c|c|}
\hline
& & $j=0$ & $j=1$ & $j=2$ \\
\hline
\multirow{4}{*}{$k=0$} & \multirow{2}{*}{$i=0$} & $H^0(\bb P(S_+),\mc O)$ & $H^0(\bb P(S_+),\mc O(1)\otimes S_-)$ & $H^0(\bb P(S_+),\mc O(2))$\\
 & & $=\textcolor{red}{\bb C} $  &  $ =\textcolor{red}{S_+\otimes S_-}$ & $=\textcolor{red}{\sym^2 S_+}$ \\
 & \multirow{2}{*}{$i=1$} & $H^1(\bb P(S_+),\mc O) $ & $H^1(\bb P(S_+),\mc O(1)\otimes S_-)$ & $H^1(\bb P(S_+),\mc O(2))$\\
  & & $=0$  & $=0$ & $=0$\\
\hline
\multirow{4}{*}{$k=1$} & \multirow{2}{*}{$i=0$} & $H^0(\bb P(S_+),\mc O(-1))$ & $H^0(\bb P(S_+),\mc O\otimes S_-)$ & $H^0(\bb P(S_+),\mc O(1))$\\
 & & $=0 $  &  $ = \textcolor{red}{S_-}$ & $=\textcolor{red}{S_+}$ \\
 & \multirow{2}{*}{$i=1$} & $H^1(\bb P(S_+),\mc O(-1)) $ & $H^1(\bb P(S_+),\mc O\otimes S_-)$ & $H^1(\bb P(S_+),\mc O(1))$\\
  & & $=0$  & $=0$ & $=0$\\
\hline
\multirow{4}{*}{$k=2$} & \multirow{2}{*}{$i=0$} & $H^0(\bb P(S_+),\mc O(-2))$ & $H^0(\bb P(S_+),\mc O(-1)\otimes S_-)$ & $H^0(\bb P(S_+),\mc O)$\\
 & & $=0 $  &  $ =0$ & $=\textcolor{red}{\bb C}$ \\
 & \multirow{2}{*}{$i=1$} & $H^1(\bb P(S_+),\mc O(-2)) $ & $H^1(\bb P(S_+),\mc O(-1)\otimes S_-)$ & $H^1(\bb P(S_+),\mc O)$\\
  & & $=\textcolor{red}{\bb C}$  & $=0$ & $=0$\\
\hline
\multirow{4}{*}{$k=3$} & \multirow{2}{*}{$i=0$} & $H^0(\bb P(S_+),\mc O(-3))$ & $H^0(\bb P(S_+),\mc O(-2)\otimes S_-)$ & $H^0(\bb P(S_+),\mc O(-1))$\\
  & & $=0$  & $=0$ & $=0$\\
 & \multirow{2}{*}{$i=1$} & $H^1(\bb P(S_+),\mc O(-3)) $ & $H^1(\bb P(S_+),\mc O(-2)\otimes S_-)$ & $H^1(\bb P(S_+),\mc O(-1))$\\
  & & $=\textcolor{red}{S_+}$  & $=\textcolor{red}{S_-}$ & $=0$\\
\hline
\multirow{4}{*}{$k=4$} & \multirow{2}{*}{$i=0$} & $H^0(\bb P(S_+),\mc O(-4))$ & $H^0(\bb P(S_+),\mc O(-3)\otimes S_-)$ & $H^0(\bb P(S_+),\mc O(-2))$\\
  & & $=0$  & $=0$ & $=0$\\
 & \multirow{2}{*}{$i=1$} & $H^1(\bb P(S_+),\mc O(-4)) $ & $H^1(\bb P(S_+),\mc O(-3)\otimes S_-)$ & $H^1(\bb P(S_+),\mc O(-2))$\\
  & & $=\textcolor{red}{\sym^2 S_+}$  & $=\textcolor{red}{S_+\otimes S_-}$ & $=\textcolor{red}{\bb C} $\\
\hline
\end{tabular}
\end{center}

Now let us consider the action of a supertranslation $Q_+ = \alpha \otimes w$ in $S_+ \otimes W$ on this complex.  Recall that the supertranslation acts on super twistor space by a shear transformation.  The corresponding action on $\Omega^{0,\bullet}$ is by contraction with this vector field. In terms of the trigrading $(i,j,k)$, the supertranslation $Q_+$ has $i$-degree 0, $j$-degree $0$ and $k$-degree $-1$.  The action of $Q_+$ on the $\d_i$-cohomology will split into two pieces: an operator acting on the $E_1$-page of the spectral sequence with degree $(0,0,-1)$, and a correction term on the $E_2$-page of the spectral sequence with degree $(-1,1,-1)$.

\begin{prop} \label{action_of_positive_spinor_prop}
The action of the supertranslation $Q_+ = \alpha \otimes w$ on the complex of Proposition \ref{BV_SS_prop} is by the operator
\[\xymatrix{
\ul{-1} &\ul{0} &\ul{1} &\ul{2} \\
\Omega^0(\RR^4; \gg) \ar[r] &\Omega^1(\RR^4; \gg) \ar[r] &\Omega^2_+(\RR^4; \gg ) & \\
&\Omega^0(\RR^4; S_- \otimes \Pi W^* \otimes \gg) \ar[u]^{t_\alpha \otimes i_w}\ar[r] & \Omega^0(\RR^4; S_+ \otimes \Pi W^* \otimes \gg) \ar[u]^{t_\alpha \otimes i_w} & \\
&\Omega^0(\RR^4; \wedge^2 W^* \otimes \gg) \ar[u]^{d_\alpha \otimes i_w}\ar[r] &\Omega^0(\RR^4; \wedge^2 W^* \otimes \gg) \ar[u]^{t_\alpha \otimes i_w} & \\
&\Omega^0(\RR^4; S_+ \otimes \Pi \wedge^3 W^* \otimes \gg) \ar[u]^{t_\alpha \otimes i_w} \ar[r] &\Omega^0(\RR^4; S_- \otimes \Pi \wedge^3 W^* \otimes \gg) \ar[u]^{d_\alpha \otimes i_w} & \\
&\Omega^2_+(\RR^4; \wedge^4 W^* \otimes \gg) \ar[r] \ar[u]^{t_\alpha \otimes i_w} &\Omega^3(\RR^4; \wedge^4 W^* \otimes\gg)  \ar[u]^{t_\alpha \otimes i_w}\ar[r] &\Omega^4(\RR^4; \wedge^4 W^* \otimes\gg),
}\]
where $t_\alpha \otimes i_w \colon \wedge^k W^* \to S_+ \otimes \wedge^{k-1} W^*$ acts by tensoring with $\alpha \in S_+ \iso S_+^*$ and contracting with $w \in W$, and $d_\alpha \otimes i_w \colon \Omega^0(\RR^4; \wedge^j S_- \otimes \wedge^k W^*) \to \Omega^0(\RR^4; \wedge^{j-1} S_- \wedge^{k-1} W^*)$ for $j=1,2$ still contracts with $w \in W$, but now we first tensor with $\alpha$ and then use the isomorphism $S_+ \otimes S_- \iso \CC^4$ to identify an element of $S_+ \otimes S_-$ with a first order differential operator acting on $\CC$-valued functions on $\RR^4$. Here we are using the convention of a double complex with commuting differentials.
\end{prop}

\begin{proof}
We will use a spectral sequence argument as described above: the operator $t_\alpha \otimes i_w$ will be the differential $\d_1$ on the $E_1$-page, and the operator $d_\alpha \otimes i_w$ will be the differential $\d_2$ on the $E_2$-page.  Specifically, we consider the spectral sequence of the double complex with gradings $(j-k,i)$, with differentials $\d_j$ and $S^{(1)}(Q_+)$, the action of $Q_+$.  There are no further differentials in the spectral sequence beyond the $E_3$-page for degree reasons.

We study the differential $\d_1$ first, this is just the residual action of $S^{(1)}(Q_+)$ on the $\d_i$-cohomology in addition to the differential identified in the above proposition coming from $\d_j$.  This supertranslation, as described in Section \ref{susy_action_section}, acts on supertwistor space by a shear transformation $[\alpha] (w \cdot \dd_\theta)$, where $[\alpha]$ is multiplication by a projective coordinate on $\bb P(S_+)$.  The differential $\d_1$ is by contraction with this vector field: the contraction with $w \cdot \dd_\theta$ operator is $i_w$, and multiplication by $[\alpha]$ corresponds to the operator $t_\alpha$. The total differential is of the form $\d_j + (-1)^{j-k} S^{(1)}(Q_\alpha)$.

We now compute the differential $\d_2$ on the $E_2$ page of the spectral sequence.  By inspecting the degrees in which the complex is supported, the only possible differentials on this page go from $(i,j,k)=(1,0,2)$ to $(0,1,1)$ and $(1,1,3)$ to $(0,2,2)$.  We compute these differentials as the composite $(\d_j \pm t_\alpha \otimes i_w) \circ \d_i^{-1} \circ  (\d_j\mp t_\alpha \otimes i_w)$; we must verify that this composite acts as $d_\alpha \otimes i_w$ as described above.  Here, the operator $\d_i$ induces an isomorphism between the terms in degree $(0,0,1)$ and $(1,0,1)$, the terms in degree $(0,1,2)$ and $(1,1,2)$, and the terms in degree $(0,2,3)$ and $(1,2,3)$, so it makes sense to discuss its inverse.  Notice also that the term quadratic in $i_w$ vanishes, so we must compute $\d_j \circ \d_i^{-1} \circ t_\alpha$ and $t_\alpha \circ \d_i^{-1} \circ \d_j$.  We will apply this computation starting in degree $(i,j,k)=(1,0,2)$: the analogous term from degree $(1,1,3)$ follows the same steps.

We begin with a function $f \in \Omega^0(\RR^4)$, representing the $\d_i$-cohomology class of $f \otimes \omega \in \Omega^0(\RR^4) \otimes \Omega^{0,1}(\bb P(S_+), \OO(-2))$. (we ignore the $\wedge^2 W^* \otimes\gg$ coefficients: the composite operator acts on this factor simply by the application of $i_w$).  In this degree, the difference between $\d_j \circ \d_i^{-1} \circ t_\alpha$ and $t_\alpha \circ \d_i^{-1} \circ \d_j$ sends $f \otimes \omega$ to
\begin{align*}
(\d_j \circ \d_i^{-1} \circ t_\alpha - t_\alpha \circ \d_i^{-1} \circ \d_j) (f \otimes \omega) &= \sum_{a,b=1}^2 \frac {\dd f}{\dd x_{ab}} \otimes (\alpha \d_i^{-1}(\alpha_a \omega) -	 \alpha_a \d_i^{-1}(\alpha \omega)) \beta_b \\
&= \sum_{a,b=1}^2 \frac {\dd f}{\dd x_{ab}} \otimes \langle \alpha_a, \alpha \rangle \beta_b.
\end{align*}
where the last equality can be proved, for example, by explicitly computing in coordinates. This coincides with the desired operator $d_\alpha$, so the result follows.
\end{proof}

Let us now consider the action of a supertranslation $Q_- = \beta \otimes w^* \in S_- \otimes W^*$.
\begin{prop} \label{action_of_negative_spinor_prop}
The action of the supertranslation $Q_- = \beta \otimes w^*$ on the complex of Proposition \ref{BV_SS_prop} is by the operator
\[\xymatrix{
\ul{-1} &\ul{0} &\ul{1} &\ul{2} \\
\Omega^0(\RR^4; \gg) \ar[r] &\Omega^1(\RR^4; \gg) \ar[r] \ar[d]_{d_{\beta} \otimes t_{w^*}} &\Omega^2_+(\RR^4; \gg ) \ar[d]_{d_{\beta} \otimes t_{w^*}} & \\
&\Omega^0(\RR^4; S_- \otimes \Pi W^* \otimes \gg) \ar[r] \ar[d]_{t_{\beta} \otimes t_{w^*}} & \Omega^0(S_+ \otimes \Pi W^* \otimes \gg) \ar[d]_{d_{\beta} \otimes t_{w^*}}& \\
&\Omega^0(\RR^4; \wedge^2 W^* \otimes \gg) \ar[r] \ar[d]_{d_{\beta} \otimes t_{w^*}} &\Omega^0(\RR^4; \wedge^2 W^* \otimes \gg) \ar[d]_{t_{\beta} \otimes t_{w^*}} & \\
&\Omega^0(\RR^4; S_+ \otimes \Pi \wedge^3 W^* \otimes \gg) \ar[d]_{d_{\beta} \otimes t_{w^*}} \ar[r] &\Omega^0(\RR^4; S_- \otimes \Pi \wedge^3 W^* \otimes \gg) \ar[d]_{d_{\beta} \otimes t_{w^*}} & \\
&\Omega^2_+(\RR^4; \wedge^4 W^* \otimes\gg) \ar[r] &\Omega^3(\RR^4; \wedge^4 W^* \otimes\gg) \ar[r] &\Omega^4(\RR^4; \wedge^4 W^* \otimes\gg),
}\]
where the operator $t_{\beta} \otimes t_{w^*}$ is by tensoring with $\beta \otimes w^*$, and the operator $d_{\beta} \otimes t_{w^*}$ is the first order differential operator defined as in Proposition \ref{action_of_positive_spinor_prop}.
\end{prop}

\begin{proof}
The calculation here follows the same procedure as that of Proposition \ref{action_of_positive_spinor_prop}, where the action is generated by the first two differentials in the spectral sequence of the double complex.  Recall that a negative spinor $Q_- = \beta \otimes w^*$ acts on super twistor space by a superspace rotation, which we described as the vector field $\langle w^*, \theta \rangle (\beta \cdot \dd_{z})$. The first factor $\langle w^*, \theta \rangle$ is just the tensoring by $w^*$ operator $t_{w^*}$.  The second factor is the differential operator we have denoted by $d_{\beta}$.  The differential $\d_2$ on the $E_2$ page of the spectral sequence acts as the composite $\d_i \circ \d_j^{-1} \circ (d_{\beta} \otimes t_{w^*})$, which is just $t_{\beta} \otimes t_{w^*}$ by a similar calculation to the previous proposition.
\end{proof}

The action of the supertranslation algebra obtained here by transferring along the quasi-isomorphism of Proposition \ref{BV_SS_prop} is not strict, but only an $L_\infty$ action.  There is an explicit quadratic $L_\infty$ correction term to the action described so far, associated to a pair of supertranslations $Q_+ \in S_+ \otimes W$ and $Q_- \in S_- \otimes W^*$.  Its form is somewhat complicated, so we won't derive it here, but such correction terms are calculated in \cite{ElliottSafronovWilliams} (for instance, refer to equation (34) of loc. cit, and dimensionally reduce from 10 to 4 dimensions).

\subsection{Holomorphic Twists} \label{holo_section}
Fix first a holomorphic supercharge of rank $(1,0)$, say $Q_{\mr{hol}} = \alpha_1 \otimes e_2$.  Such a holomorphic supercharge corresponds in particular to a choice of complex structure on $\RR^4$ as follows. Recall that we have fixed an isomorphism of $\spin(4)$-representations $S_+ \otimes S_- \to \CC^4$. In terms of the bases $\alpha_i$ for $S_+$ and $\beta_i$ for $S_-$, let us say that $\alpha_1\otimes \beta_i \mapsto \ol \dd_{z_i}$, and $\alpha_2\otimes \beta_i \mapsto \dd_{z_i}$.

Consider the homomorphism $\iota \colon \mr U(2) \to \spin(4) \iso \SU(2)_+ \times \SU(2)_-$ associated to this choice of complex structure, when restricted to $\SU(2) \sub \mr U(2)$ this is just the inclusion of $\{1\} \times \SU(2)_-$.  We will consider the classical BV complex of $\mc N=4$ super Yang--Mills theory as a $\mr U(2)$-representation, under the homomorphism
\[(\iota, \phi_{\mr{KW}} \circ \iota) \colon \mr U(2) \to \spin(4) \times \SL(4;\CC)_R\]
into the product of the spin group and the group of R-symmetries (acting on $W$).

\begin{prop} \label{holo_BV_complex_prop}
As a $\mr U(2)$-representation, the ($\ZZ\times \ZZ/2\ZZ$-graded) classical BV complex of Proposition \ref{BV_SS_prop} decomposes into summands with the following form:
\[\xymatrix@R-=10pt@C-=4pt{
\ul{-1} &\ul{0} &\ul{1} &\ul{2} \\
{\Omega^0(\CC^2;\gg)} \ar[r]^-{(\del,\ol\del )} &{\Omega^{1,0}(\CC^2;\gg)} \oplus \Omega^{0,1}(\CC^2;\gg) \ar[r]^-{\begin{psmallmatrix}
\del  & 0 \\
\pi \ol \del  &\pi \del  \\
0 & \ol\del
\end{psmallmatrix}
 }  &{\Omega^{2,0}(\CC^2;\gg)}  \oplus \Omega^{1,1}_\parallel(\CC^2;\gg)\oplus \Omega^{0,2}(\CC^2;\gg) & \\
&\Pi\Omega^{1,0}(\CC^2;\gg) \ar[r]^-{(\del, \pi \ol{\del} )}  &\Pi(\Omega^{2,0}(\CC^2;\gg) \oplus \Omega^{1,1}_{\parallel}(\CC^2;\gg)) && \\
&\Pi\Omega^{0,1}(\CC^2;\gg) \ar[r]^-{(\pi \del, \ol{\del}) }   &\Pi(\Omega^{1,1}_{\parallel}(\CC^2;\gg) \oplus \Omega^{0,2}(\CC^2;\gg) ) \\
&{\Pi\Omega^{1,1}(\CC^2;\gg)} \ar[r]^-{(\del,\ol{\del} )} &{\Pi(\Omega^{2,1}(\CC^2;\gg)} \oplus\Omega^{1,2}(\CC^2;\gg))& \\
&{\Omega^{0}(\CC^2;\gg)}  \ar[r]^-{\omega \wedge  \Delta(-)} & {\Omega^{1,1}_{\parallel}(\CC^2;\gg)} & \\
&{\Omega^{1,0}(\CC^2;\gg)} \ar[r]^-{\omega \wedge  \Delta(-)} &\Omega^{2,1}(\CC^2;\gg) && \\
&\Omega^{0,1}(\CC^2;\gg) \ar[r]^-{\omega \wedge  \Delta(-)} &{\Omega^{1,2}(\CC^2;\gg)} \\
&  \Omega^{1,1}_\parallel(\CC^2;\gg) \ar[r]^-{\omega \wedge  \Delta(-)} &  {\Omega^{2,2}(\CC^2;\gg)} & \\%
&\Pi( \Omega^{1,0}(\CC^2;\gg) \oplus \Omega^{0,1}(\CC^2;\gg)) \ar[r]^-{\begin{psmallmatrix}
\ol \del & \del
\end{psmallmatrix}}  &{\Pi\Omega^{1,1}(\CC^2;\gg)}& \\
&\Pi(\Omega^{2,0}(\CC^2;\gg)  \oplus \Omega^{1,1}_{\parallel}(\CC^2;\gg)) \ar[r]^-{ \begin{psmallmatrix}   \ol  \del & \del  \end{psmallmatrix} }  &\Pi{\Omega^{2,1}(\CC^2;\gg)}&& \\
&\Pi(\Omega^{1,1}_{\parallel}(\CC^2;\gg) \oplus \Omega^{0,2}(\CC^2;\gg)) \ar[r]^-{ \begin{psmallmatrix}   \ol  \del & \del   \end{psmallmatrix} } &\Pi\Omega^{1,2}(\CC^2;\gg) \\
&\Omega^{2,0}(\CC^2;\gg) \oplus \Omega^{1,1}_\parallel(\CC^2;\gg)  \oplus {\Omega^{0,2}(\CC^2;\gg)}
\ar[r]^-{\begin{psmallmatrix}
\ol \del  & \del & 0 \\
0  &  \ol \del  & \del
\end{psmallmatrix}}  &\Omega^{2,1}(\CC^2;\gg) \oplus {\Omega^{1,2}(\CC^2;\gg)} \ar[r]^-{\begin{psmallmatrix}
\ol \del & \del
\end{psmallmatrix}
} & {\Omega^{2,2}(\CC^2;\gg)}}
\]
where $\Omega^{1,1}(\CC^2;\gg)$ decomposes into $\Omega^{1,1}_\parallel(\CC^2;\gg) \oplus \Omega^{1,1}_\perp(\CC^2;\gg)$: the components proportional to and orthogonal to the K\"ahler form $\omega = \d z_1 \wedge \d \ol z_1 + \d z_2 \wedge \d \ol z_2$ on $\CC^2$. Moreover, with respect to the chosen twisting datum $\alpha$, we obtain a $\ZZ$-graded classical BV complex
\begin{center}
\resizebox{\textwidth}{!}{\xymatrix@R-3pt@C-45pt{
\ul{-2}  &\ul{-1} &\ul{0} &\ul{1} &\ul{2} & \ul{3}  \\
& {\Omega^0(\CC^2;\gg)} \ar[r] &{\Omega^{1,0}(\CC^2;\gg)} \oplus \Omega^{0,1}(\CC^2;\gg) \ar[r]  &{\Omega^{2,0}(\CC^2;\gg)}  \oplus \Omega^{1,1}_\parallel(\CC^2;\gg)\oplus \Omega^{0,2}(\CC^2;\gg) &\\
&\Omega^{1,0}(\CC^2;\gg) \ar[r]  &\Omega^{2,0}(\CC^2;\gg) \oplus \Omega^{1,1}_{\parallel}(\CC^2;\gg) &\\
&\Omega^{0,1}(\CC^2;\gg) \ar[r] &\Omega^{1,1}_{\parallel}(\CC^2;\gg) \oplus \Omega^{0,2}(\CC^2;\gg)  & {\Omega^{1,1}(\CC^2;\gg)} \ar[r] &{\Omega^{2,1}(\CC^2;\gg)} \oplus\Omega^{1,2	}(\CC^2;\gg)&  \\
{\Omega^{0}(\CC^2;\gg)} \ar[r] &\Omega^{1,1}_\parallel(\CC^2;\gg)  &{\Omega^{1,0}(\CC^2;\gg)} \ar[r] &\Omega^{2,1}(\CC^2;\gg) && \\
&&\Omega^{0,1}(\CC^2;\gg) \ar[r] &{\Omega^{1,2}(\CC^2;\gg)} & \Omega^{1,1}_\parallel(\CC^2;\gg) \ar[r] &  {\Omega^{2,2}(\CC^2;\gg)} & \\
&\Omega^{1,0}(\CC^2;\gg) \oplus \Omega^{0,1}(\CC^2;\gg) \ar[r] &{\Omega^{1,1}(\CC^2;\gg)} & \Omega^{2,0}(\CC^2;\gg)  \oplus \Omega^{1,1}_{\parallel}(\CC^2;\gg) \ar[r] &{\Omega^{2,1}(\CC^2;\gg)}&& \\
&&&\Omega^{1,1}_{\parallel}(\CC^2;\gg) \oplus \Omega^{0,2}(\CC^2;\gg) \ar[r] &\Omega^{1,2}(\CC^2;\gg) \\
&&\Omega^{2,0}(\CC^2;\gg)\oplus \Omega^{1,1}_\parallel(\CC^2;\gg)  \oplus {\Omega^{0,2}(\CC^2;\gg)}  \ar[r] &\Omega^{2,1}(\CC^2;\gg) \oplus {\Omega^{1,2}(\CC^2;\gg)} \ar[r] &{\Omega^{2,2}(\CC^2;\gg)}.
}}   \end{center}
\end{prop}

\begin{proof}
This follows from the following observations about the restriction of complex $\spin(4)$-representations to $\mr U(2)$. The first line comes from \[\xymatrix{\Omega^0(\RR^4; \gg) \ar[r]^{\d} &\Omega^1(\RR^4; \gg) \ar[r]^{\d_+} &\Omega^2_+(\RR^4; \gg ) }\]
using
\begin{align*}
V^1 & \cong S_+\otimes S_- \cong   \alpha_1  \otimes \langle \beta_1 , \beta_2 \rangle \oplus \alpha_2\otimes \langle \beta_1 , \beta_2\rangle \cong  V^{1,0}\oplus V^{0,1} \\
 (\wedge^ 2 V)_+ &\cong \sym^2 S_+ \cong \langle \alpha_1 \alpha_1 \rangle\oplus \langle \alpha_1\alpha_2+\alpha_2\alpha_1\rangle    \oplus  \langle \alpha_2\alpha_2\rangle   \cong V^{2,0}\oplus V^{1,1}_{\parallel }  \oplus V^{0,2}
\end{align*}
The next third lines come from
\[\xymatrix{\Omega^0(\RR^4; S_- \otimes \Pi W^* \otimes \gg) \ar[r]^-{\sd {\d}} & \Omega^0(S_+ \otimes \Pi W^* \otimes \gg)}\]
using
\begin{align*}
 S_-\otimes e_1^*& \cong V^{0,1}\\
 S_-\otimes e_2^*& \cong V^{1,0}\\
S_-\otimes \langle f_1^*,f_2^*\rangle & \cong   \sym^2S_-\oplus \wedge^2 S_- \cong V^{1,1}_\perp \oplus V^{1,1}_{\parallel}\\
 S_+\otimes e_1^*& \cong  \langle \alpha_1 e_1^*\rangle \oplus \langle \alpha_2 e_1^*\rangle   \cong V^0\oplus V^{0,2} \\
S_+\otimes e_2^*& \cong   \langle \alpha_1 e_2^*\rangle \oplus \langle \alpha_2 e_2^*\rangle   \cong  V^{2,0}\oplus V^0\\
 S_+\otimes \langle f_1^*,f_2^*\rangle  &\cong  \alpha_1 \otimes \langle f_1^*,f_2^*\rangle\oplus \alpha_2 \otimes \langle f_1^*,f_2^*\rangle \cong V^{1,0} \oplus V^{0,1}.
\end{align*}
The next three come from
\[\xymatrix{\Omega^0(\RR^4; \wedge^2 W^* \otimes \gg) \ar[r]^{\Delta} &\Omega^0(\RR^4; \wedge^2 W^* \otimes \gg) }\]
where
\begin{align*}
 \wedge^2 W^* & \cong \wedge^2 ( S_+^*\oplus S_-^*) \cong S_+^* \otimes S_-^* \oplus \wedge^2 S_+^* \oplus \wedge^2 S_-^* \\
& \cong  (e _1^* \otimes \langle f_1^*,f_2^*\rangle \oplus e _2^* \otimes \langle f_1^*,f_2^*\rangle  ) \oplus \langle e_1^*\wedge e_2^* \rangle \oplus \langle f_1^*\wedge f_2^*\rangle  \cong  ( V^{0,1}\oplus V^{1,0} ) \oplus V^0  \oplus V^0
\end{align*}
The rest also follows from a similar analysis. Note that we can freely use the isomorphisms $V^0\cong V^{1,1}_{\parallel }\cong V^{2,2}$, $V^{1,0}\cong V^{2,1}$, and $V^{0,1}\cong V^{1,2}$ of $\mr U(2)$-representations, which corresponds to using the Lefschetz operator $L \colon \Omega^{p,q}\to \Omega^{p+1,q+1}$ given by $L=\omega\wedge (-)$. Now the second part also immediately follows as the $\alpha$-weight for each term is clear from the above proof.
\end{proof}

Because we have already fixed a twisting datum $\alpha \colon \mr U(1)\to G_R$, a choice of supercharge would just determine what additional arrows one obtains for the $\ZZ$-graded BV complex we just obtained, as long as we consider the chosen complex structure.

Now, let us consider the action of the supertranslation $Q_{\mr{hol}}=\alpha_1\otimes e_2$.  This element up to rescaling is stabilized by $\mr U(2)$, and hence the part of the associated differential operator of degree 0 will be a $\mr U(2)$-equivariant map.

\begin{lemma} \label{holo_twist_lemma}
The twist of 4d $\mc N=4$ self-dual Yang--Mills theory by $Q_{\mr{hol}}$ has the following classical BV complex, considered as a $\mr U(2)$-module:
\begin{center}
\resizebox{\textwidth}{!}{\xymatrix@R-3pt@C-45pt{
\ul{-2}  &\ul{-1} &\ul{0} &\ul{1} &\ul{2} & \ul{3}  \\
& {\Omega^0(\CC^2;\gg)} \ar[r] &{\Omega^{1,0}(\CC^2;\gg)} \oplus \Omega^{0,1}(\CC^2;\gg) \ar[r]  &{\Omega^{2,0}(\CC^2;\gg)}  \oplus \Omega^{1,1}_\parallel(\CC^2;\gg)\oplus \Omega^{0,2}(\CC^2;\gg) &\\
&\Omega^{1,0}(\CC^2;\gg) \ar[r] \ar@[blue][ur]  &\Omega^{2,0}(\CC^2;\gg) \oplus \Omega^{1,1}_{\parallel}(\CC^2;\gg) \ar@[blue][ur]  &\\
&\Omega^{0,1}(\CC^2;\gg) \ar[r] &\Omega^{1,1}_{\parallel}(\CC^2;\gg) \oplus \Omega^{0,2}(\CC^2;\gg)  & {\Omega^{1,1}(\CC^2;\gg)} \ar[r] &{\Omega^{2,1}(\CC^2;\gg)} \oplus\Omega^{1,2	}(\CC^2;\gg)&  \\
{\Omega^{0}(\CC^2;\gg)} \ar@[green][ur]  \ar[r] &\Omega^{1,1}_\parallel(\CC^2;\gg) \ar@[blue][ur]  &{\Omega^{1,0}(\CC^2;\gg)} \ar[r]\ar@[green][ur]  &\Omega^{2,1}(\CC^2;\gg) \ar@[blue][ur] && \\
&&\Omega^{0,1}(\CC^2;\gg) \ar[r] &{\Omega^{1,2}(\CC^2;\gg)} & \Omega^{1,1}_\parallel(\CC^2;\gg) \ar[r]  &  {\Omega^{2,2}(\CC^2;\gg)} & \\
&\Omega^{1,0}(\CC^2;\gg) \oplus \Omega^{0,1}(\CC^2;\gg) \ar[r] \ar@[blue][ur] &{\Omega^{1,1}(\CC^2;\gg)} \ar@[green][ur]  & \Omega^{2,0}(\CC^2;\gg)  \oplus \Omega^{1,1}_{\parallel}(\CC^2;\gg) \ar[r]\ar@[blue][ur]   &{\Omega^{2,1}(\CC^2;\gg)} \ar@[green][ur] && \\
&&&\Omega^{1,1}_{\parallel}(\CC^2;\gg) \oplus \Omega^{0,2}(\CC^2;\gg) \ar[r] &\Omega^{1,2}(\CC^2;\gg) \\
&&\Omega^{2,0}(\CC^2;\gg)\oplus \Omega^{1,1}_\parallel(\CC^2;\gg)  \oplus {\Omega^{0,2}(\CC^2;\gg)}  \ar@[blue][ur] \ar[r] &\Omega^{2,1}(\CC^2;\gg) \oplus {\Omega^{1,2}(\CC^2;\gg)} \ar@[blue][ur] \ar[r] &{\Omega^{2,2}(\CC^2;\gg)},
}}\end{center}
where the blue arrows, coming from the differential $\d_1$ in the spectral sequence of the double complex, are given by the signed identity map between matching terms $(-1)^{k+1} \colon \Omega^{p,q}[k+1] \to \Omega^{p,q}[k]$. The green arrows, coming from the differential $\d_2$ in the spectral sequence, are given by the differential $\ol \dd \colon  \Omega^{p,q}[k+1] \to \Omega^{p,q+1}[k]$.
\end{lemma}

\begin{proof}
First of all, the complex only with the black differentials comes from the complex of Proposition \ref{holo_BV_complex_prop} by shifting the summands in $\bigwedge ^\bullet   W^*$ by the $\alpha_{\mr{hol}}$-weight.

Second, the differential $\d_1 = t_\alpha \otimes i_w$ from Proposition \ref{action_of_positive_spinor_prop}, indicated in blue, corresponds to $\mr U(2)$-invariant maps between terms in this complex. Because $Q_{\mr{hol}}$ is by construction a scalar under $\mr U(2)$, all maps are isomorphisms of $\mr U(2)$-representations and uniquely determined up to a scalar. All these scalars are fixed to be $\pm 1$ depending on degree by a simple calculation: in the basis $\{\alpha_1,\alpha_2\}$ for $S_+$ and $\{e_1,e_2,f_1,f_2\}$ for $W$, the operator $t_{\alpha_1} \otimes i_{e_2}$ sends tensor products of basis vectors to tensor products of basis vectors. For example, consider $\Omega^{2,2}\oplus \Omega^{2,0}\to \Omega^{2,0}\oplus \Omega^{1,1}$ which is induced from $\Omega^0(\RR^4 ; S_+\otimes W^*)\to \Omega^2_+(\RR^4)$. The relevant map is the one $ \alpha_1 \otimes e_2 \colon  S_+\otimes e_2^* \to \sym^2 S_+$ which is defined by $\alpha_1e_2^*\mapsto \alpha_1\alpha_1$ and $\alpha_2e_2^* \mapsto  \frac{1}{2} ( \alpha_1\alpha_2+ \alpha_2\alpha_1)$. This is the claimed morphism from the above proposition. A similar computation can be done for the other blue arrows.

Finally, the differential $\d_2 = d_{\alpha_1} \otimes i_{e_2}$ from Proposition \ref{action_of_positive_spinor_prop}, indicated in green, corresponds to a first-order anti-holomorphic differential operator between terms in the complex.  Each such map is again uniquely determined up to a scalar, and these scalars are fixed to be $\pm 1$. For example, consider $\Omega^{1,0}\to \Omega^{1,1}$ which is induced from $\Omega^0(\RR^4; \wedge^2 W^*)\to \Omega^0(\RR^4; S_-\otimes W^*)$. As the relevant map is from $e_2^* \otimes S_-^* \subset \wedge^2 W^* $ to $S_-\otimes S_-^*\subset S_-\otimes W^*$, the induced map is given by $f \otimes e_2^* \otimes f_k^*  \mapsto \sum _{i=1}^2 \del_{\ol{z}_i}f \otimes  \beta_i \otimes f_k^* $, or $\overline{\del}\colon \Omega^{1,0}(\CC^2)\to \Omega^{1,1}(\CC^2)$, because
\[\xymatrix{ \Omega^{0,1}(\PP(S_+) , \mc O(-2) ) \ar[r]^-{\alpha_1} & \Omega^{0,1}(\PP(S_+) , \mc O(-1) ) \ar[r]^-{\overline{\del}^{-1}}& \Omega^{0,0}(\PP(S_+) , \mc O(-1) ) \ar[r]^-{\alpha_1} & \Omega^{0,0}(\PP(S_+) , \mc O )  }\]
induces the identity map in cohomology.  A similar computation can be done for the other green arrows.
\end{proof}

This holomorphically twisted theory is equivalent to a theory with a much more familiar description via the complex of $(p,q)$ forms on $\CC^2$.  In fact, the equivalence can be realized as a deformation retraction.  Recall that a cochain complex $\LL_0$ is a \emph{deformation retract} of a cochain complex $\LL$ if one can construct cochain maps
\[
\xymatrix{
\LL_0 \ar@<3pt>[r]^i & \LL \ar@<3pt>[l]^p \ar@(dr,ur)[]_h
}
\]
where $i$ and $p$ have cohomological degree 0, $p \circ i = \id_{\LL_0}$, and $h$ is a cochain homotopy from $i \circ p$ to $\id_{\LL}$.

\begin{theorem} \label{holo_twist_BF_thm}
The twist of $\mc N=4$ self-dual Yang--Mills by the twisting datum $(Q_{\mr{hol}}, \alpha_{\mr{hol}})$ has a deformation retract given by the classical BV theory $(\Omega^{\bullet, \bullet}(\CC^2; \gg[2] \oplus \gg[1]), \ol \dd)$, with symplectic pairing given by the invariant pairing of the two copies of $\gg$.
\end{theorem}

\begin{remarks}
\begin{enumerate}
 \item The twisted theory discussed here is an example of \emph{holomorphic BF theory}, valued in the super Lie algebra $\gg[\eps_1, \eps_2]$ where $\eps_i$ are odd parameters.
 \item The same result appeared as \cite[Theorem 4.2]{EY1}, though using a different argument.  We present a new argument here that can be used when comparing the further deformations of the holomorphic twist with the further twists by supercharges commuting with $Q_{\mr{hol}}$.
 \item Using the methods of \cite{EY1} we can give a non-perturbative description of the holomorphically twisted theory, as we discussed there.  When we work globally we can identify the moduli space of solutions to the equations of motion of the holomorphic twist on a complex algebraic surface $X$ with the derived stack $T_f[1]\higgs_G(X)$: the 1-shifted tangent space to the moduli stack of $G$-Higgs bundles on $X$.
 \item One can also compare this calculation to the calculation of the same theory by dimensional reduction from 10 dimensions in \cite[Theorem 10.9]{ElliottSafronovWilliams}.
\end{enumerate}
\end{remarks}

\begin{proof}
We will split the twisted classical BV complex up into summands, and identify each summand explicitly as retracting onto $\Omega^{k, \bullet}(\CC^2)$ for some $k$.  We begin with the first two rows in our diagram:
\[\xymatrix@R-10pt{
{\Omega^0(\CC^2;\gg)} \ar[r] &{\Omega^{1,0}(\CC^2;\gg)} \oplus \Omega^{0,1}(\CC^2;\gg) \ar[r] &{\Omega^{2,0}(\CC^2;\gg)} \oplus \Omega^{1,1}_\parallel(\CC^2;\gg) \oplus \Omega^{0,2}(\CC^2;\gg)  & \\
\Omega^{1,0}(\CC^2;\gg) \ar[ur]\ar[r] &\Omega^{2,0}(\CC^2;\gg) \oplus \Omega^{1,1}_{\parallel}(\CC^2;\gg). \ar[ur] &&} \]
The canonical projection on the Dolbeault complex $\Omega^{0,\bullet}(\CC^2;\gg)$ is a quasi-isomorphism, so this defines our map $p$.  It has an explicit quasi-inverse $i$ from $\Omega^{0,\bullet}(\CC^2;\gg)$, given by the sum of the canonical inclusion (which is not a cochain map by itself) with the map sending an element $(\alpha_0, \alpha_{0,1}, \alpha_{0,2})$ to $(\dd \alpha_0, \dd \alpha_{0,1}) \in \Omega^{1,0}(\CC^2; \gg) \oplus \Omega^{1,1}_\parallel(\CC^2;\gg )$.  The cochain homotopy $h$ is just the degree $-1$ identity map between the two copies of the factors $\Omega^{1,0}(\CC^2; \gg)$, $\Omega^{2,0}(\CC^2; \gg)$, and $\Omega^{1,1}_\parallel(\CC^2;\gg)$.

Now we take the subcomplex
\[\xymatrix@R-10pt{
&&{\Omega^{0,1}(\CC^2;\gg)} \ar[r] &\Omega^{1,1}_{\parallel}(\CC^2;\gg) \oplus \Omega^{0,2}(\CC^2;\gg)  \\ \\
&{\Omega^{0}(\CC^2;\gg)} \ar[uur] \ar[r] & \Omega^{1,1}_\parallel(\CC^2;\gg) \ar[uur]
}\]
The map $p$ is defined by the projection onto $\Omega^{0,\bullet}(\CC^2;\gg) $. The map $i$ is defined to be the sum of the inclusion of $\Omega^{0,\bullet}(\CC^2;\gg) $ and the map sending $\alpha_{0,1}$ in $\Omega^{0,1}(\CC^2;\gg)$ to $\pi \dd \alpha_{0,1}$ in $\Omega^{1,1}_\parallel(\CC^2;\gg)$.  The cochain homotopy $h$ is the identity map of cohomological degree $-1$ on $\Omega^{1,1}_\parallel(\CC^2;\gg)$.

Next we take the component
\[\xymatrix@R-10pt{
&{\Omega^{1,1}(\CC^2;\gg)} \ar[r] &\Omega^{2,1}(\CC^2;\gg)\oplus \Omega^{1,2}(\CC^2;\gg)& \\ \\
{\Omega^{1,0}(\CC^2;\gg)} \ar[uur] \ar[r] &\Omega^{2,1}(\CC^2;\gg) \ar[uur] &&
}\]
Again, as the canonical projection on to the complex $ \Omega^{1,\bullet}(\CC^2;\gg)$ with differential $\ol \dd$ is a quasi-isomorphism, this defines our map $p$.  It has an explicit quasi-inverse given by the sum of the canonical inclusion and the map sending $(\alpha_{1,0}, \alpha_{1,1}, \alpha_{1,2})$ to $\dd \alpha_{1,1}$ in $\Omega^{2,1}(\CC^2;\gg)$.  The cochain homotopy $h$ is just the degree $-1$ identity map between the two copies of $\Omega^{2,1}(\CC^2;\gg)$.

Fourth, consider the subcomplex
\[\xymatrix{
&&\Omega^{0,1}(\CC^2;\gg) \ar[r] &{\Omega^{1,2}(\CC^2;\gg)} \\ \\
& \Omega^{1,0}(\CC^2;\gg) \oplus \Omega^{0,1}(\CC^2;\gg) \ar[uur] \ar[r] &{\Omega^{1,1}(\CC^2;\gg)}. \ar[uur]& \\
}\]
The canonical inclusion of the complex $ \Omega^{1,\bullet}(\CC^2;\gg)$ is again a quasi-isomorphism, this defines our map $i$. It has a quasi-inverse $p$ defined by the sum of the projection on to $ \Omega^{1,\bullet}(\CC^2;\gg)$ and the map sending $\alpha_{0,1} \in  \Omega^{0,1}(\CC^2;\gg)$ to $\dd \alpha_{0,1}$ in $ \Omega^{1,1}(\CC^2; \gg)$.  The cochain homotopy $h$ is given by the degree $-1$ identity map between the two copies of $\Omega^{0,1}(\CC^2;\gg)$.

For the component
\[\xymatrix@R-10pt{
& \Omega^{1,1}_\parallel(\CC^2;\gg) \ar[r] &{\Omega^{2,2}(\CC^2;\gg)} & \\ \\
\Omega^{2,0}(\CC^2;\gg)  \oplus \Omega^{1,1}_{\parallel}(\CC^2;\gg) \ar[uur] \ar[r] &\Omega^{2,1}(\CC^2;\gg). \ar[uur] &&
}\]
The map $p$ is defined by the sum of the projection on to $\Omega^{2,\bullet}(\CC^2;\gg)$ and the map sending $\alpha_{1,1}$ in $\Omega^{1,1}_\parallel(\CC^2;\gg)$ to $\dd \alpha_{1,1}$ in $\Omega^{2,1}(\CC^2;\gg)$.  The map $i$ is defined to be the inclusion of $\Omega^{2,\bullet}(\CC^2;\gg)$. The cochain homotopy $h$ is the identity map of cohomological degree $-1$ on $\Omega^{1,1}_\parallel(\CC^2;\gg)$.

Finally we consider the last two rows:
\[\xymatrix{&&\Omega^{1,1}_{\parallel}(\CC^2;\gg) \oplus \Omega^{0,2}(\CC^2;\gg) \ar[r] &\Omega^{1,2}(\CC^2;\gg) \\
&\Omega^{2,0}(\CC^2;\gg)\oplus \Omega^{1,1}_\parallel(\CC^2;\gg)  \oplus {\Omega^{0,2}(\CC^2;\gg)}  \ar[ur] \ar[r] &\Omega^{2,1}(\CC^2;\gg) \oplus {\Omega^{1,2}(\CC^2;\gg)} \ar[ur] \ar[r] &{\Omega^{2,2}(\CC^2;\gg)},
}\]
is quasi-isomorphic to $\Omega^{2,\bullet}(\CC^2;\gg)$ via the canonical inclusion $i$.It has a quasi-inverse $p$ defined by the sum of the projection onto $\Omega^{2,\bullet}(\CC^2;\gg)$, the map sending $\alpha_{1,1}$ in $\Omega^{1,1}_\parallel(\CC^2;\gg)$ to $\dd \alpha_{1,1} $ in $\Omega^{2,1}(\CC^2; \gg)$, and the map sending $\alpha_{1,2}$ in $\Omega^{1,2}(\CC^2;\gg)$ to $\dd \alpha_{1,2}$ in $\Omega^{2,2}(\CC^2;\gg)$.  The cochain homotopy $h$ is just the degree $-1$ identity map between the two copies of the factors $\Omega^{0,2}(\CC^2; \gg)$, $\Omega^{1,2}(\CC^2; \gg)$ as well as $\Omega^{1,1}_\parallel(\CC^2;\gg)$.

The fact that this equivalence is compatible with the symplectic pairing that pairs the two copies of $\gg$ follows by a direct calculation.
\end{proof}

\begin{remark}
This argument only demonstrates the equivalence between the two theories at the level of cochain complexes, but both the holomorphic twist and the complex $(\Omega^{\bullet, \bullet}(\CC^2; \gg[2] \oplus \gg[1]), \ol \dd)$ come equipped with shifted dg Lie structures (or equivalently with the datum of a BV interaction functional).  Because we have the data of a deformation retract, we could explicitly compute the homotopy transfer of the Lie bracket using the explicit formulas of \cite[Theorem 10.3.9]{LodayVallette}, and check that it coincides with the natural bracket induced from the bracket on $\gg$.  We will not do this calculation here, instead referring to the argument in \cite[Theorem 10.9]{ElliottSafronovWilliams}.
\end{remark}

Now, let us consider a holomorphic supertranslation $Q_{\mr{hol}}' = \beta_1 \otimes f_1^*$ lying in the negative helicity part of the supertranslation algebra.  We will show that the the twist by this supercharge is equivalent to the same holomorphic BF theory.  This supercharge is fixed by image of the embedding $\mr U(2) \inj \spin(4)$ associated to the \emph{opposite} complex structure.  This embedding is surjective onto $\SU(2)_+$ instead of $\SU(2)_-$.  When we decompose the classical BV fields into irreducible summands under this $\mr U(2)$ action, which we denote by $\mr U(2)_+$ to emphasize the difference from the previous $\mr U(2)$, we obtain the same component fields as in Proposition \ref{holo_BV_complex_prop}, where the roles of $S_+$ and $S_-$ are reversed.  The following proof proceeds similarly to the proof of Lemma \ref{holo_twist_lemma} and Theorem \ref{holo_twist_BF_thm} above.

\begin{lemma} \label{holo_twist_second_lemma}
The twist of 4d $\mc N=4$ self-dual Yang--Mills theory by the supercharge $Q'_{\mr{hol}},$ has the following classical BV complex, considered as a $\mr U(2)_+$-module:
\begin{center}
\resizebox{\textwidth}{!}{\xymatrix@R-3pt@C-45pt{
\ul{-2}  &\ul{-1} &\ul{0} &\ul{1} &\ul{2} & \ul{3}  \\
& {\Omega^0(\CC^2;\gg)} \ar[r] &{\Omega^{1,0}(\CC^2;\gg)} \oplus \Omega^{0,1}(\CC^2;\gg) \ar[r] \ar@[blue][dr]  &{\Omega^{2,0}(\CC^2;\gg)}  \oplus \Omega^{1,1}_\parallel(\CC^2;\gg)\oplus \Omega^{0,2}(\CC^2;\gg) \ar@[blue][dr]  &\\
&& & \Omega^{2,0}(\CC^2;\gg)  \oplus \Omega^{1,1}_{\parallel}(\CC^2;\gg) \ar[r] &{\Omega^{2,1}(\CC^2;\gg)}&& \\
&\Omega^{1,0}(\CC^2;\gg) \oplus \Omega^{0,1}(\CC^2;\gg) \ar[r] \ar@[green][dr]  &{\Omega^{1,1}(\CC^2;\gg)} \ar@[blue][dr]  &\Omega^{1,1}_{\parallel}(\CC^2;\gg) \oplus \Omega^{0,2}(\CC^2;\gg) \ar@[green][dr]  \ar[r] &\Omega^{1,2}(\CC^2;\gg) \ar@[blue][dr]  & \\
& &{\Omega^{1,0}(\CC^2;\gg)} \ar[r] &\Omega^{2,1}(\CC^2;\gg) &\Omega^{1,1}_\parallel(\CC^2;\gg) \ar[r] &  {\Omega^{2,2}(\CC^2;\gg)}   \\
{\Omega^{0}(\CC^2;\gg)} \ar[r] \ar@[blue][dr]  &\Omega^{1,1}_\parallel(\CC^2;\gg)  \ar@[green][dr]   &\Omega^{0,1}(\CC^2;\gg) \ar[r] \ar@[blue][dr]  &{\Omega^{1,2}(\CC^2;\gg)} \ar@[green][dr]   &  \\
&\Omega^{1,0}(\CC^2;\gg) \ar[r]  &\Omega^{2,0}(\CC^2;\gg) \oplus \Omega^{1,1}_{\parallel}(\CC^2;\gg) & {\Omega^{1,1}(\CC^2;\gg)} \ar[r] &{\Omega^{2,1}(\CC^2;\gg)} \oplus\Omega^{1,2	}(\CC^2;\gg) \\
&\Omega^{0,1}(\CC^2;\gg) \ar[r] \ar@[blue][dr]   &\Omega^{1,1}_{\parallel}(\CC^2;\gg) \oplus \Omega^{0,2}(\CC^2;\gg)  \ar@[blue][dr]  & &  \\
&&\Omega^{2,0}(\CC^2;\gg)\oplus \Omega^{1,1}_\parallel(\CC^2;\gg)  \oplus {\Omega^{0,2}(\CC^2;\gg)}  \ar[r] &\Omega^{2,1}(\CC^2;\gg) \oplus {\Omega^{1,2}(\CC^2;\gg)} \ar[r] &{\Omega^{2,2}(\CC^2;\gg)}.
}}   \end{center}
Here we have written in blue the differential $\d_1$ on the $E_1$ page of the spectral sequence.  This differential acts by the first-order differential operator $(-1)^{k} \dd \colon  \Omega^{p,q}[k+1] \to  \Omega^{p+1,q}[k]$. The green arrows, coming from the differential $\d_2$ in the spectral sequence, are given by the identity map.
\end{lemma}

\begin{theorem} \label{opposite_holo_twist_BF_thm}
The twist by $Q'_{\mr{hol}}$ has a deformation retraction to the classical BV theory $(\Omega^{\bullet, \bullet}(\CC^2; \gg[2] \oplus \gg[1]), \dd)$, with symplectic pairing given by the invariant pairing of the two copies of $\gg$.
\end{theorem}

\subsection{Further Twists}
Let us now discuss non-holomorphic twists of 4d $\mc N=4$ super Yang--Mills theory.  Such twists can be realized as \emph{further} twists of holomorphic twists.  That is, given a square-zero odd element $Q \ne 0$ in the $\mc N=4$ supertranslation algebra, we can decompose
\[Q = Q_{\mr{hol}} + Q',\]
where $Q_{\mr{hol}}$ is a holomorphic supercharge, and $[Q_{\mr{hol}}, Q'] = 0$.  Recall that the twist of a supersymmetric theory by a supertranslation $Q$ is defined by adding the term $S_\gg(Q)$ to the action functional.  In this setting, the twist by $Q$ is obtained from the twist by $Q_{\mr{hol}}$ by adding a further twisting term to the action functional: when $S_\gg^{(>2)} = 0$, i.e. when there is only a second-order $L_\infty$ correction to the on-shell supertranslation action, this further twisting term is just
\[S_\gg(Q') + 2S_\gg^{(2)}(Q_{\mr{hol}}, Q'),\]
where $S_\gg^{(2)}$ denotes the second order term in the supertranslation action.  We will treat the non-holomorphic twists  one class at a time, according to the classification that was described in Section \ref{classification_section}.

We will use the following tool in order to understand the deformation of the holomorphic twist by additional supercharges.

\begin{lemma}[Homological Perturbation Lemma]
Let \[
\xymatrix{
(\LL_0, \d_0) \ar@<3pt>[r]^i & (\LL, \d) \ar@<3pt>[l]^p \ar@(dr,ur)[]_h
}
\]
be a deformation retraction.  Choose a deformation of the differential on $\LL$ to $\d + \d'$.  Suppose that $(1-\d' \circ h)$ is invertible.  Then the operator $\d_0 + \d_0'$ where $\d_0' = p \circ (1-\d' \circ h)^{-1} \circ \d' \circ i$ is a deformation of the differential on $\LL_0$, and $(\LL_0, \d_0 + \d_0')$ is a deformation retract of $(\LL, \d + \d')$.
\end{lemma}

So, we will identify our further twisted theories using the following method.
\begin{enumerate}
 \item Let $(\mc L, \d)$ be the classical BV complex of the holomorphic twist by the supercharge $Q_{\mr{hol}}$, and let $(\mc L_0, \d_0) = (\Omega^{\bullet, \bullet}(\CC^2, \gg[2] \oplus \gg[1]), \ol \dd)$, using the deformation retraction of Theorem \ref{holo_twist_BF_thm}.  Let $\d'$ be given by the action of the supercharge $Q'$.

 \item Check that $\d' \circ h$ is a nilpotent operator, so that $1-\d' \circ h$ is automatically invertible.

 \item Identify the operator $\d_0'$ up to a scalar multiple $\lambda$.  We could do this by direct computation, but we will instead simply refer to the results of \cite{ElliottSafronovWilliams} where the deformations that can occur according to different classes of twisting supercharge are completely classified.

 \item Fix the value of $\lambda$ by directly computing the action of $\d_0'$ on a single summand of $\mc L_0$.
\end{enumerate}

\subsubsection{Rank \texorpdfstring{$(2,0)$}{(2,0)} and \texorpdfstring{$(0,2)$}{(0,2)} Twists} \label{20_section}
Let us begin by computing the further twist from a rank $(1,0)$ supercharge to the twist by a rank $(2,0)$ supercharge.  This is a topological twist.  We will also describe the (equivalent) further twist from a rank $(0,1)$ twist to rank $(0,2)$.

Let $Q_{\mr{hol}} = \alpha_1 \otimes e_2$, and let $Q' = \alpha_2 \otimes e_1$.  These supertranslations commute, and their sum $Q = Q_{\mr{hol}} + Q'$ is a square-zero topological supercharge of rank $(2,0)$. We will show that the twist of $\mc N=4$ super Yang--Mills theory by the supercharge $Q$ is \emph{perturbatively trivial}, meaning that the twisted classical BV complex is contractible.  We can, however, understand it as a contractible deformation of the holomorphically twisted theory.  In order to do so, we introduce the following terminology.

\begin{definition} \label{Hodge_family_deformation}
The \emph{Hodge family} associated to the local $L_\infty$ algebra $\mc L$ on $M$ is the $\CC[t]$-linear local $L_\infty$ algebra $\mc L_{\mr{Hod}} = \mc L[1] \oplus \mc L$, with differential given by the internal differential on $\mc L$, plus $t$ times the identity morphism $\mc L[1] \to \mc L$.  If $\mc L$ is equipped with a degree 0 symmetric non-degenerate pairing $\langle -, - \rangle$ then $\mc L_{\mr{Hod}}$ is equipped with a degree $-1$ symplectic pairing valued in $\CC[t] \otimes \mr{Dens}_M$.
\end{definition}

\begin{remark}
The notation we have used for such families of classical BV theories reflects the notion in algebraic geometry of the \emph{Hodge stack} of a prestack $X$, due to Simpson \cite{Simpson}.  One can define a deformation of the 1-shifted tangent bundle $T[1]X$ of $X$ to the de Rham stack $X_{\mr{dR}}$, which is a derived stack with contractible tangent complex.  In our earlier work \cite{EY1}, we identified a 1-parameter family of further twists of the holomorphic twist of 4d $\mc N=4$ super Yang--Mills theory on $X$ with the Hodge stack of $\mr{Higgs}_G(X)$.  In this section we will verify that this family matches up with the rank $(2,0)$ twist.
\end{remark}

\begin{theorem} \label{20_twist_thm}
The twist of self-dual $\mc N=4$ super Yang--Mills theory by the $\CC[b_1]$-family of square-zero supercharges $Q_{b_1}= Q_{\mr{hol}} + b_1Q'$, generically of rank $(2,0)$, is equivalent to the Hodge family $(\Omega^{\bullet, \bullet}(\CC^2; \gg), \ol \dd)_{\mr{Hod}}$ where the Hodge deformation is parameterized by $-b_1$.
\end{theorem}

\begin{remark}
The particular element $b_1 = -1$ is worth remarking on, yielding the complex $(\Omega^{\bullet, \bullet}(\CC^2; \gg), \ol \dd)_{\mr{dR}}$.  This is the twist by the supercharge $\alpha_1 \otimes e_2 - \alpha_2 \otimes e_1$ stabilized by the Kapustin--Witten twisted action of $\mr{Spin}(4,\CC)$.  As we will see in Section \ref{KW_section}, this corresponds to the twist Kapustin--Witten refer to as $I_A$.
\end{remark}

\begin{proof}
We will follow the outline indicated in the previous subsection.  The first step is to check that $\d' \circ h$ is nilpotent, where $\d'$ is the action of $Q'$ on the holomorphically twisted theory, and $h$ is the homotopy from the proof of Theorem \ref{holo_twist_BF_thm}.  Recall that the holomorphically twisted classical BV complex, as in Lemma \ref{holo_twist_lemma}, splits into six connected components, and the homotopy $h$ preserves this splitting.  To show that $\d' \circ h$ is nilpotent, it is enough to check that the operator $\d'$, when applied to each connected component, has image entirely contained in a different connected component.  To verify this property, we can use Proposition \ref{action_of_positive_spinor_prop}.  Specifically, one performs a calculation analogous to the one we performed in Lemma \ref{holo_twist_lemma} but using $Q'$.

Thus, we can now apply the homotopy perturbation lemma, which allows us to identify the twist by a rank $(2,0)$ supercharge with a deformation of the complex $(\Omega^{\bullet, \bullet}(\CC^2; \gg[\eps]), \ol \dd)$.  This twist is fixed, up to an overall scale, by \cite[Theorem 10.10]{ElliottSafronovWilliams}.  This theorem tells us that the twist by $Q_{b_1}$ is equivalent to the Hodge deformation
\[(\Omega^{\bullet, \bullet}(\CC^2; \gg[\eps]), \ol \dd + \alpha \frac{\dd}{\dd \eps}),\]
for some complex number $\alpha$.  To conclude the proof, we wish to show that $\alpha = -1$.  We can do this by computing the action of the deformed differential on any element of the holomorphically twisted BV complex.  Explicitly, we need to compute
\[\d_0' = p \circ (1-\d' \circ h)^{-1} \circ \d' \circ i.\]
Let us compute this operator when applied to the term $\Omega^{0,2}(\CC^2;\gg\cdot \eps)$.  Because $\d' \circ h$ is nilpotent, $(1-\d' \circ h)^{-1} = 1+ \d' \circ h$.  Since $h$ acts on this term by zero, we only need to compute the composite $p \circ \d' \circ i$.  On this factor of the BV complex, the operators $i$ and $p$ are the obvious projection and inclusion maps, so to compute this composite map, we only need to know the action of $\d'$ acts by sending $\Omega^{0,2}(\CC^2;\gg\cdot \eps^1)$ to $\Omega^{0,2}(\CC^2;\gg\cdot \eps^0)$, by minus the identity.  We can read this off from Proposition \ref{action_of_positive_spinor_prop}.
\end{proof}

\subsubsection{Rank \texorpdfstring{$(1,1)$}{(1,1)} Twists} \label{11_section}
Now, let us study the twist of $\mc N=4$ super Yang--Mills theory by a square-zero supercharge of rank $(1,1)$, say by the element
\[Q = Q_{\mr{hol}} + Q'_{\mr{hol}} = \alpha_1 \otimes e_2 + \beta_1 \otimes f_1^*.\]
We will realize this as a further twist of the rank $(1,0)$ twist; by an identical argument we could realize it as a further twist of the rank $(0,1)$ twist.

We will begin by describing the rank $(1,1)$ twist as a further twist of the rank $(1,0)$ $Q_{\mr{hol}}$-twisted $\mc N=4$ theory. Hence we will describe the 1-parameter family
\[Q_a = Q_{\mr{hol}} + a Q'_{\mr{hol}},\]
for $a \in \CC$.

\begin{theorem} \label{11_twist_thm}
The family of twists by the supercharges $Q_a$ admits a deformation retraction onto the $\CC[a]$-family of theories with classical BV complex $(\Omega^{\bullet, \bullet}(\CC^2;\gg[2]\oplus  \gg[1]), \ol \dd + a \dd_{z_1})$.
\end{theorem}

\begin{proof}
We will use an argument very similar to the proof of Theorem \ref{20_twist_thm}. Namely, we will apply the homotopy perturbation lemma for the deforming differential $\d' = Q'_{\mr{hol}}$, for which we must first check that $\d' \circ h$ is nilpotent.  This is actually easier to see than the analogous fact above; we need only observe that if we view the classical BV complex for holomorphic Chern--Simons theory as filtered by the index $k$ of Proposition \ref{hCS_BV_complex_prop}, both $\d'$ and $h$ raise $k$ degree, and hence so does their composite.

We can therefore apply the homotopy perturbation lemma once more.  We know by \cite[Theorem 10.11]{ElliottSafronovWilliams}
that the twist of 4d $\mc N=4$ super Yang--Mills theory by the supercharge $Q_1$ is equivalent to a complex of the form
\[(\Omega^{\bullet, \bullet}(\CC^2, \gg[\eps]), \ol \dd + \alpha \dd_{z_1}),\]
and we must show that $\alpha=1$.  To do this, we can compute the action of the operator
\[\d'_0 = p \circ \d' \circ i + p \circ \d' \circ h \circ \d' \circ i\]
on a single element in our BV complex, say the coordinate function $z_1$ in $\Omega^0(\CC^2, \gg[\eps])$.  We described the action of $Q'_{\mr{hol}}$ in Lemma \ref{holo_twist_second_lemma}.  On this term, $i(f) = (f, \del  f) \in \Omega^0(\CC^2, \gg) \oplus \Omega^{1,0}(\CC^2, \gg)$ in degree zero.

Now we must apply $\d'$.  This operator acts on $\Omega^{0,1}(\CC^2, \gg) \oplus \Omega^{1,0}(\CC^2, \gg)$ by projection onto the holomorphic factor, \emph{but} in the complex structure associated to the action of $\mr U(2)_+$. Recall that $S_+\otimes S_-\iso \CC^4$ is given by $\alpha_1\otimes \beta_i \mapsto \ol \dd_{z_i}$ and $\alpha_2\otimes \beta_i \mapsto \dd_{z_i}$. Hence, if we denote $\mr U(2)_-$ coordinates by $z_1 = x_1 + iy_1, z_2 = x_2 + iy_2$, then $\mr U(2)_+$ coordinates are given by $w_1 = z_1 = x_1 + iy_1, w_2=\ol{z}_1 = x_z -iy_1$. Then $\del f = \frac{\del f}{\del z_1} dz_1 + \frac{\del f}{\del z_2} dz_2 $ projects to $\frac{\del f}{\del z_1} dz_1= \del_{z_1}f$ in the $\mr U(2)_+$ coordinates.
\end{proof}

\subsubsection{Rank \texorpdfstring{$(2,1)$}{(2,1)}, \texorpdfstring{$(1,2)$}{(1,2)} and \texorpdfstring{$(2,2)$}{(2,2)} Twists} \label{22_section}

In this final part of the section, we will identify further twists of the rank $(1,1)$ twist, using the two isomorphic descriptions described above.  We will describe the twist by a rank $(2,1)$ supercharge using the description given in Theorem \ref{11_twist_thm}, and then we will study the family of maximal deformations by rank $(2,2)$ supercharges.

First, consider a family consisting generically of rank $(2,1)$ supercharges.  That is, let $Q_a$ be the 1-parameter family $Q_a = Q_{\mr{hol}} + aQ'_{\mr{hol}}$ from the previous subsection, let $Q' = \alpha_2 \otimes e_1$ as in Section \ref{20_section} and consider the $\CC[a,b_1]$-family of supercharges
\[Q_{a,b_1} = Q_{\mr{hol}} + aQ'_{\mr{hol}}  + b_1Q'.\]
When we set $a=0$ we obtain the family considered in Section \ref{20_section}, and when we set $b_1=0$ we obtain the family considered in Section \ref{11_section}.

\begin{theorem} \label{21_theorem}
The twist of self-dual $\mc N=4$ super Yang--Mills theory by the family $Q_{a,b_1}$ of supercharges is equivalent to the Hodge family $(\Omega^{\bullet, \bullet}(\CC^2; \gg), \ol \dd + a \dd_{z_1})_{\mr{Hod}}$,
where the Hodge deformation is parameterized by the coordinate $-b_1$.
\end{theorem}

\begin{remark}
In the specified $\bb A^2$-family of twists it is not possible to rescale away the $a$-dependence: the $\bb A^2$ we have chosen is affinely embedded in a four-dimensional vector space of square-zero supercharges, and this affine subspace is not preserved by the $\CC^\times$ action that rescales the positive helicity supercharges.
\end{remark}

\begin{proof}
This is a combination of the calculations from Theorem \ref{20_twist_thm} and Theorem \ref{11_twist_thm}.  Write $\d'$ and $\d'_{\mr{hol}}$ for the deformations of the holomorphically twisted BV complex coming from the actions of $Q'$ and $Q'_{\mr{hol}}$ respectively.  When we apply the homological perturbation lemma, we can see immediately that $((\d' + \d'_{\mr{hol}}) \circ h)^2 = 0$ (the maps $\d' \circ h$ and $\d'_{\mr{hol}} \circ h$ square to zero and commute), and thus by our two previous calculations, the $Q_{a,b_1}$ twisted BV complex is equivalent to the desired complex provided that 
\[p \circ \d' \circ h \circ \d_{\mr{hol}}' \circ i = p \circ \d'_{\mr{hol}} \circ h \circ \d' \circ i = 0.\]
To see this, first recall that in the proof of Theorem \ref{20_twist_thm} we saw that the map $p \circ \d' \circ h$ equals zero.  Very similarly we can see that the map $h \circ \d' \circ i$ is also zero.

Thus, by the homological perturbation lemma, the differential in the classical BV complex of the $Q_{a,b_1}$-twisted theory is
\begin{align*}
\ol \dd + p \circ (b_1\d' + a\d'_{\mr{hol}}) \circ i &= \ol \dd - b_1 \dd_\eps + a \dd_{z_1},
\end{align*}
where $\dd_\eps$ is the identity map between the two copies of $\Omega^{\bullet, \bullet}(\CC^2; \gg ) $ occurring in the Hodge family.
\end{proof}

We can do an analogous calculation using the same argument to compute the twist by the similar $\CC[b_1,b_2]$-family of supercharges generically of rank $(2,1)$:
\[Q_{b_1,b_2} = Q_{\mr{hol}} + b_1Q' + b_2Q''.\]
\begin{theorem} \label{21_opposite_theorem}
The twist of self-dual $\mc N=4$ super Yang--Mills theory by the family $Q_{b_1,b_2}$ of supercharges is equivalent to the Hodge family $(\Omega^{\bullet, \bullet}(\CC^2; \gg), \ol \dd - b_2\dd_{z_2})_{\mr{Hod}}$,
where the Hodge deformation is parameterized by the coordinate $-b_1$.
\end{theorem}

Alternatively, we can mirror this argument to compute the twist by a rank $(1,2)$ supercharge.  Let $Q_a'$ be the 1-parameter family $Q_a' = Q'_{\mr{hol}} + aQ_{\mr{hol}}$, and let $Q'' = \beta_2 \otimes f_2^*$.  Consider the $\CC[a,b_2]$-family of supercharges
\[Q'_{a,b_2} = Q'_{\mr{hol}} + aQ_{\mr{hol}} + b_2Q''.\]

\begin{theorem} \label{12_theorem}
The twist of self-dual $\mc N=4$ super Yang--Mills theory by the family $Q'_{a,b_2}$ of supercharges is equivalent to the Hodge family $(\Omega^{\bullet, \bullet}(\CC^2; \gg), \dd + a \ol \dd_{z_1})_{\mr{Hod}}$,
where the Hodge deformation is parameterized by the coordinate $-b_2$.
\end{theorem}

Finally, let's study the result when we twist by the most generic class of supercharge, those of rank $(2,2)$.  We can analyze these twists by combining our calculations of the rank $(2,1)$ and $(1,2)$ twists above.  We'll start with the supercharge $Q_{a,0}$, which we can equivalently describe -- for $a \ne 0$ -- as a rescaled element $aQ'_{a^{-1},0}$ of our rank $(1,2)$ deformation family.  We'll deform this supercharge to the following $\CC[a,b_1,b_2]$-family of supercharges:
\begin{align*}
Q_{a,b_1,b_2} &= Q_{\mr{hol}}+ aQ'_{\mr{hol}} + b_1Q' + b_2Q''.
\end{align*}

Before we compute the twist, let us outline our strategy.  Note that the abelian super Lie algebra $\CC^{0|3}$ generated by the supercharges $Q'_{\mr{hol}}$, $Q'$, and $Q''$ is the odd part of the cohomology of the supertranslation algebra $\mf T^{\mc N=4}$ with respect to the odd square-zero operator $[Q_{\mr{hol}}, -]$.   
The supertranslation action on self-dual $\mc N=4$ super Yang--Mills theory induces an action of this algebra on the $Q_{\mr{hol}}$-twisted theory $\mc T^{Q_{\mr{hol}}}$.  This twisted theory was identified in Theorem \ref{holo_twist_BF_thm} as having classical BV complex
\[(\Omega^{\bullet, \bullet}(\CC^2; \gg[\eps][1]), \ol \dd) \iso (\Omega^{0, \bullet}(\CC^2; \gg[\eps_1, \eps_2, \eps][1]), \ol \dd),\]
where the isomorphism sends $\d   z_i$ to $\eps_i$ by Kapustin--Witten twisting homomorphism.  This has a natural action of $\CC^{0|3}$, by acting on the three odd generators $\eps_1,\eps_2, \eps$.  These two actions are related by a non-linear endomorphism of the Lie algebra $\CC^{0|3}$:
\[\xymatrix{
\CC^{0|3}_{Q'_{\mr{hol}},Q',Q''} \ar[rr]^F \ar[dr]_{\alpha_{\mr{SUSY}}} &&\CC^{0|3}_{\eps_1,\eps_2,\eps}\ar[dl]^{\alpha_{\mr{can}}} \\ 
&\mr{Der}(\mc T^{Q_{\mr{hol}}}).
}\]
Concretely, the map $\alpha_{\mr{SUSY}}$ comes from the supertranslation action on the untwisted theory, and is a quadratic $L_\infty$ map, whereas the map $\alpha_{\mr{can}}$ is a linear embedding with image spanned by $\del_{\eps_1}=\del_{z_1}$, $\del_{\eps_2}=\del_{z_2}$, and $\del_{\eps}$.  The map $F$ can be defined by choosing a projection onto $\CC^{0|3}$ for which $\alpha_{\mr{can}}$ is a section (although, as we will see shortly, $F$ will be independent of this choice).  To determine the action of a generic element of the twisted supertranslation algebra it is enough to compute this $L_\infty$ endomorphism.

\begin{remark}
In the diagram above $\CC^{0|3}$ and $\mr{Der}(\mc T^{Q_{\mr{hol}}})$ are $\ZZ \times \ZZ/2\ZZ$-graded objects, but the map $\alpha_{\mr{SUSY}}$ (and hence also the map $F$) only preserves the diagonal $\ZZ/2\ZZ$-grading.  Recall that an $L_\infty$ morphism is the same as a CDGA map between the associated Chevalley--Eilenberg cochain complexes. In this case, $C^\bullet(\CC^{0|3}) \iso \CC[a,b_1,b_2]$, where $a,b_1,b_2$ are generators of bidegree $(1,1)$.  The map $F$, therefore, is the same as an endomorphism of $\CC[a,b_1,b_2]$ of even total degree.  To specify such a map is just to specify three polynomials in the variables $a$, $b_1$, and $b_2$.
\end{remark}

\begin{theorem} \label{22_theorem}
The twist of self-dual $\mc N=4$ super Yang--Mills theory by the family $Q_{a,b_1,b_2}$ of supercharges is equivalent to the Hodge family $(\Omega^{\bullet, \bullet}(\CC^2; \gg), \ol \dd + a \dd_{z_1} - b_2 \dd_{z_2})_{\mr{Hod}}$,
where the Hodge deformation is parameterized by $-b_1-ab_2$.
\end{theorem}

In particular we obtain a ``B-type'' twist -- meaning here that the Hodge deformation is trivial -- along the locus where $b_1 =- ab_2$.

\begin{remark}
Note that we need to take some care here.  We cannot apply the homotopy perturbation lemma anymore because the operator $(a\dd'_{\mr{hol}} + b_2\dd'') \circ h$ will not usually be nilpotent.
\end{remark}

\begin{proof}
We can completely determine the $L_\infty$ endomorphism of $\CC^{0|3}$ relating the desired supertranslation action to the natural action by observing that it is at most quadratic, along with the description of the endomorphism after restricting to the three possibly two-dimensional subalgebras from our calculations of less generic twists.

First, the endomorphism is quadratic, because the supertranslation action on $\mc N=4$ super Yang--Mills theory only included quadratic corrections.  That is, in the commutative triangle of $L_\infty$-morphisms, the map $\alpha_{\mr{SUSY}}$ -- the action coming from supertranslations -- is quadratic, and the map $\alpha_{\mr{can}}$ -- the canonical action coming from the vector fields $\dd_{z_1}, \dd_{z_2}$ and $\dd_{\eps}$ -- is linear, and therefore $F$ must also be quadratic.

Now, choose coordinates $(a,b_1,b_2)$ on the domain $\CC^{0|3}$ corresponding to the generators $Q'_{\mr{hol}}$, $Q'$, and $Q''$ of the $Q_{\mr{hol}}$-cohomology of the supertranslation algebra and coordinates on the codomain $\CC^{0|3}$ corresponding to the basis $\del_{z_1},\del_{z_2},\del_\eps$. According to Theorem \ref{21_theorem}, on the locus $b_2=0$ our endomorphism is given by the map $(a,b_1,0) \mapsto (a, 0, -b_1)$.  That is, $ aQ'_{\mr{hol}}+b_1Q' $ maps to the vector field $a\dd_{z_1} - b_1 \dd_\eps$.  Similarly, by Theorem \ref{21_opposite_theorem}, on the locus $a=0$ it is given by the map $(0,b_1,b_2)\mapsto (0, -b_2, -b_1)$.  

On the locus $b_1=0$ our endomorphism is given by $a\dd_{z_1} - b_2\dd_{z_2}$ -- as we can see by considering the limit as $a$ or $b_2 \to 0$ -- plus a quadratic correction term proportional to $ab_2$, say $ab_2(\alpha \dd_{z_1} + \beta \dd_{z_2} + \gamma \dd_\eps)$. We can fix $(\alpha,\beta,\gamma)$ using Theorem \ref{12_theorem}.  Indeed, suppose $a=b_2=1$.  The correction term is obtained from the equivalence
\[\xymatrix{
(\Omega^{\bullet,\bullet}(\CC^2; \gg[\eps][1]), \ol \dd + \dd_{z_1}) \ar@<1ex>[r]^{p'i} &(\Omega^{\bullet,\bullet}(\CC^2; \gg[\eps][1]), \dd + \ol \dd_{z_1}), \ar@<1ex>[l]^{pi'}
}\]
where $(p,i)$ and $(p',i')$ are the deformation retracts defining the rank $(1,1)$ twist as a deformation of a rank $(1,0)$ and $(0,1)$ twist respectively, from Section \ref{11_section}.  The maps $p'i$ and $pi'$ actually define inverse cochain isomorphisms, since at the level of the rank $(1,1)$ twist the composite $p' \circ \d \circ h \circ i' = 0$, and similarly where we swap the primed and unprimed maps.  The deformation of the complex $(\Omega^{\bullet,\bullet}(\CC^2; \gg[\eps][1]), \ol \dd + \dd_{z_1})$ to the rank $(1,2)$ twist is therefore given by the composite map $pi' \circ (-\dd_\eps) \circ p'i$.  We can compute this operator by evaluating it on the ghost $\Omega^0(\CC^2;\gg)$ in degree $-2$, where we find it acts by $-\dd_\eps - \dd_{z_2}$.  Indeed, the operator $p'i$ acts trivially on $f \in \Omega^0(\CC^2;\gg\eps)$, so $\dd_\eps \circ p'i(f \eps) = f$.  The operator $i'$ sends $f$ to $(f, \ol \dd f) \in \Omega^0(\CC^2, \gg) \oplus \Omega^{0,1}(\CC^2, \gg)$.  Finally, $p$ acts trivially on the first component, and it acts on the second component $\ol \dd f$ in coordinates as $f_{(1)} \d \ol z_1 + f_{(2)} \d \ol z_2 \mapsto f_{(2)} \d z_2 = \dd_{z_2} f \in \Omega^{1,0}(\CC^2, \gg)$.  Therefore $(\alpha,\beta,\gamma) = (0,0,-1)$.

Putting these three calculations together, for general elements, the quadratic map must send $(a,b_1,b_2)$ to $(a, -b_2, -b_1 - ab_2)$ as required.

\end{proof}

The discussion on twists is summarized in the diagram of Figure \ref{twist_explanation_fig}.

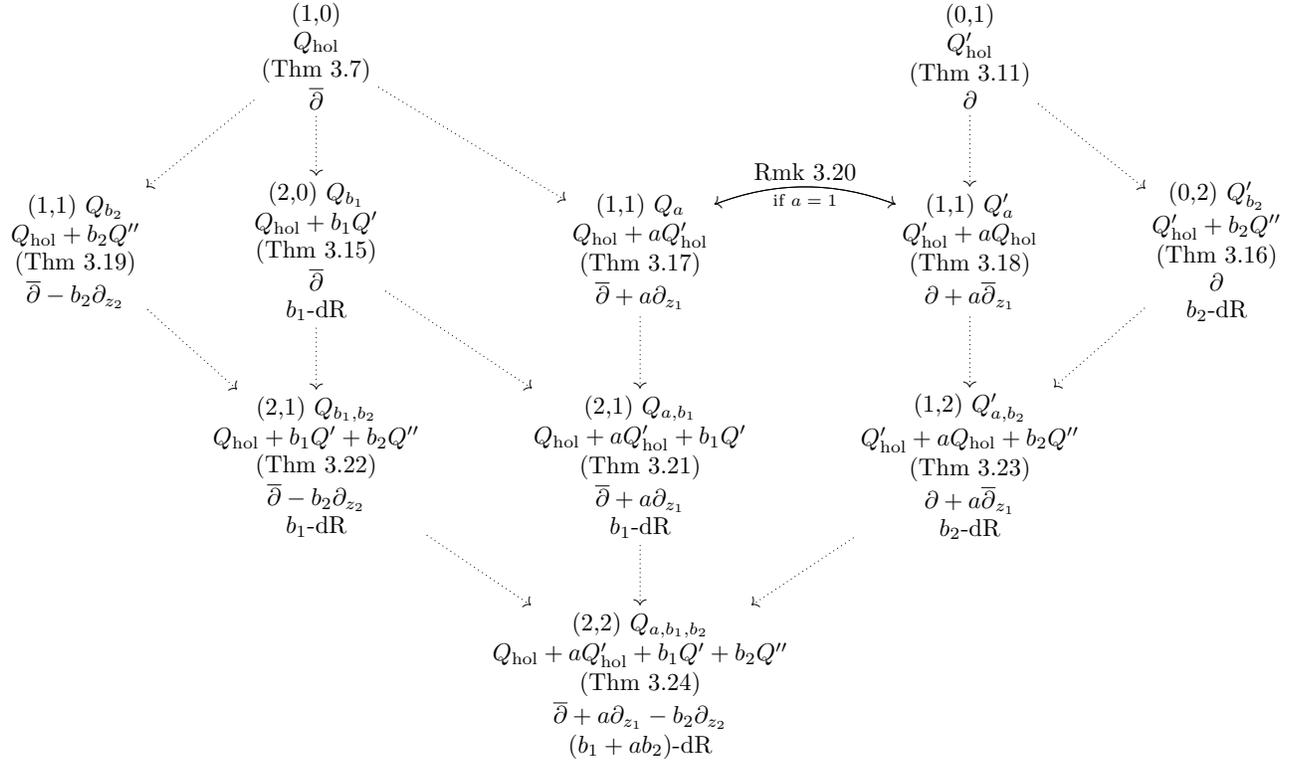
\begin{figure}[!ht]
\resizebox{\textwidth}{!}{
\xymatrix{
 & \txt{(1,0)\\ $Q_{\mr{hol}}$\\(Thm \ref{holo_twist_BF_thm})\\ $\ol\dd$}\ar@{.>}[dl]\ar@{.>}[d]\ar@{.>}[dr] & & \txt{(0,1)\\$Q'_{\mr{hol}}$\\ (Thm \ref{opposite_holo_twist_BF_thm}) \\$\dd$} \ar@{.>}[d]\ar@{.>}[dr]  \\
 \txt{(1,1) $Q_{b_2}$\\$ Q_{\mr{hol}}+ b_2 Q''$ \\ \\
$\ol \dd - b_2\dd_{z_2}$}  \ar@{.>}[dr] &  \txt{(2,0) $Q_{b_1}$\\$Q_{\mr{hol}}+b_1Q'$\\ (Thm \ref{20_twist_thm})\\ $\ol\dd$\\ $-b_1$-dR }\ar@{.>}[d]  \ar@{.>}[dr]    &  \txt{(1,1) $Q_a$ \\ $Q_{\mr{hol}}+ a Q_{\mr{hol}}'$\\ (Thm \ref{11_twist_thm}) \\
$\ol \dd + a \dd_{z_1}  $ } \ar@{.>}[d]  & \txt{(1,1) $Q'_a$\\$  Q_{\mr{hol}}'+a Q_{\mr{hol}}$ \\ \\
$\dd + a \ol \dd_{z_1}  $}  \ar@{.>}[d]  \ar@/_2pc/[l]   & \txt{(0,2) $Q'_{b_2}$\\ $Q'_{\mr{hol}}+b_2Q''$\\ \\  $\dd$\\ $-b_2$-dR}  \ar@{.>}[dl]   \\
& \txt{(2,1) $Q_{b_1,b_2}$\\ $ Q_{\mr{hol}}+ b_1 Q'+ b_2Q''$\\ (Thm \ref{21_opposite_theorem})\\ $\ol\dd - b_2 \dd_{z_2}$\\ $-b_1$-dR}  \ar@{.>}[dr]   & \txt{(2,1) $Q_{a,b_1}$\\ $ Q_{\mr{hol}}+ a Q'_{\mr{hol}}+ b_1 Q'$\\ (Thm \ref{21_theorem})\\ $\ol\dd + a \dd_{z_1}$\\ $-b_1$-dR}  \ar@{.>}[d]    & \txt{(1,2) $Q'_{a,b_2}$\\ $ Q_{\mr{hol}}'+ a Q_{\mr{hol}}+ b_2 Q''$\\ (Thm \ref{12_theorem})\\ $\dd + a \ol \dd_{z_1}$\\ $-b_2$-dR}  \ar@{.>}[dl]     \\
&& \txt{(2,2) $Q_{a,b_1,b_2}$ \\
$Q_{\mr{hol}} +aQ_{\mr{hol}}' +b_1 Q'+b_2Q''$\\
(Thm \ref{22_theorem})\\ $\ol\dd + a \dd_{z_1} - b_2\dd_{z_2}$\\ $(-b_1-ab_2)$-dR}
}}
\caption{A schematic illustrating the relationship between the twists we have discussed in this paper.  For each twist we give the rank of the relevant supercharge, the notation we have used to indicate the supercharge, the result in which a BV description for the twist is given, and then the differential on $\Omega^{\bullet,\bullet }(\CC^2; \gg[\eps][1])$ occurring in this description.  The operator $\lambda \dd_\eps$ is indicated by $\lambda$-dR.  Deformations, or ``further'' twists, are indicated with dotted lines.}
\label{twist_explanation_fig}
\end{figure} 

\subsection{Comparison with the Work of Kapustin--Witten} \label{KW_section}

In this section we will discuss the set-up considered by Kapustin and Witten in \cite{KapustinWitten}, and explain how the twists that they consider are realized in the analysis of the present paper.  Along the way we will discuss the occurrence of the (complexified) coupling constant and Kapustin--Witten's canonical parameter $\Psi$, and talk about how S-duality acts on the families of twisted theories. However, as discussed in Remark \ref{comparison with KW}, the way we obtain forms of categorical geometric Langlands correspondence will be different from the one of Kapustin and Witten.

For most of their paper, Kapustin and Witten consider the one parameter family of supercharges that are compatible with what we've referred to in Section \ref{classification_section} as the Kapustin--Witten twisting homomorphism $\phi_{\mr{KW}}$.  They denote this family by
\[u \eps_L + v \eps_R = u(\alpha_1 \otimes e_2 - \alpha_2 \otimes e_1) + v(\beta_1 \otimes f_1^* - \beta_2 \otimes f_2^*).\]
There is an action of $\CC^\times$ by R-symmetries that simultaneously rescales $u$ and $v$, and hence one can consider the $\bb{CP}^1$-family of twists with projective coordinate $t = v/u$.

Kapustin and Witten consider the full four parameter family of supercharges that we have been discussing in \cite[Section 5.1]{KapustinWitten}.  In terms of our family of supercharges from Proposition \ref{P1xP1_family_prop}, they use the notation
\[\hat \eps = u'(\alpha_1 \otimes e_2) + u''(\alpha_2 \otimes e_1) + v'(\beta_1 \otimes f_1^*) + v''(\beta_2 \otimes f_2^*).\]
Again, the four-dimensional family of twists is invariant under the $\CC^\times \times \CC^\times$ R-symmetry whose factors simultaneously rescale $u', v'$ and $u'', v''$ respectively, so we can instead consider the $\bb{CP}^1 \times \bb{CP}^1$-family with projective coordinates that Kapustin and Witten denote $(w_+,w_-) = (-v'/u', u''/v'')$.  Note that in our notation as for $u'=1$ fixed, we have $b_1=u''$, $a=v'$, and $b_2=v''$ and hence $w_+= -a$ and $w_-= b_1/b_2$. Note also that all points in this family correspond to supercharges of total rank at least two.

Within this family, two special loci can be extracted.  Recall that Kapustin and Witten argue that twists of $\mc N=4$ on a product $C \times \Sigma$ of two Riemann surfaces can be identified as twists of $\mc N=(4,4)$ supersymmetric sigma models from $\Sigma$ into the Hitchin system on $C$, with an appropriately chosen holomorphic or symplectic structure in its twistor family. In the $\bb{CP}^1 \times \bb{CP}^1$ family, the diagonal locus, $w_+ = w_-$, is argued to consist of \emph{B-models} into the Hitchin system.  Note that this coincides with our condition from Theorem \ref{22_theorem} that the Hodge deformation vanishes when $b_1 =- ab_2$. The locus $\ol{w}_+ \cdot w_-= -1$ is argued to consist of \emph{A-models} into the Hitchin system. It's also worth taking note of another special locus where $w_+ \cdot w_- = -1$: this is the locus of twists that are $\spin(4)$-invariant, and therefore can be defined on more general 4-manifolds. Note that as we mostly work in an algebraic framework that is defined for a complex algebraic surface, the latter two features are not visible. 

Let us recall the action of S-duality on the space of supertranslations.  The action of S-duality depends on the value of the complexified coupling constant $\tau$.  The ordinary coupling constant is a real number usually denoted $\frac 1{e^2}$, which rescales the action functional.  In the BV formalism, this can be viewed as a real rescaling of the BV symplectic form.  The \emph{complexified} coupling constant for super Yang--Mills theory is obtained by also including a topological term into the action functional of the form
\[S_\theta = -\frac \theta{8\pi^2} \int \langle F_A \wedge F_A \rangle.\]
The complexified coupling constant is then defined to be the complex linear combination
\[\tau = 4\pi\left(\frac \theta {8\pi^2} + i \frac{1}{e^2}\right).\]

Now, S-duality acts on the projective coordinates $(w_+, w_-)$ by $(w_+, w_-) \mapsto (\frac{\tau}{|\tau|}w_+, \frac{|\tau|}{\tau}w_-)$.  
In particular, if the $\theta$-angle vanishes, so that $\tau$ is purely imaginary, the S-duality transformation is given by $(w_+, w_-) \mapsto (iw_+, -iw_-)$.  Note that this is not yet an involution, but instead it has order 4.  We recover an involution when we quotient by the discrete $\ZZ/2\ZZ$-symmetry acting antipodally on the two copies of $\bb{CP}^1$. 

We can summarize the relationship between our notation and that of Kapustin and Witten;  In Table \ref{notation_table} we list the various twisting supercharges that we have analyzed so far, with the classical BV descriptions, Kapustin and Witten's notations, as well as their S-dual theory.
\begin{table}[h] 
\centering
\begin{tabular}{|c|c|c|c|c|c|}
\hline
Rank & $(w_+, w_-)$ & BV Differential & KW Name & Dual \\
\hline
\rowcolor{light-gray}$(1,0)$ & N/A & $\ol \dd$ & None & Self-dual\\
\rowcolor{light-gray}$(0,1)$ & N/A & $\dd$ & None & Self-dual\\
\rowcolor{light-gray}$(2,0)$ & $(0, \infty)$ & $\ol \dd + \dd_\eps$ & $I_A$ & Self-dual \\
\rowcolor{light-gray}$(0,2)$ & $(\infty, 0)$ & $ \dd + \dd_\eps$ & $(-I)_A$ & Self-dual \\
$(1,1)$ & $(0,0)$ & $\ol \dd + \dd_{z_1}$ & $I_B$ & Self-dual \\
$(1,1)$ & $(\infty, \infty)$ & $\dd + \ol \dd_{z_1}$ & $(-I)_B$ & Self-dual \\
$(2,1)$ & $(t, \infty)$ & $\ol \dd + t\dd_{z_1} + \dd_\eps$ & None & Self-dual \\
$(1,2)$ & $(\infty, t)$ & $\dd + t\ol \dd_{z_1} + \dd_\eps$ & None & Self-dual\\
$(2,2)$ & $(-i,i)$ & $\ol \dd + i\dd_{z_1} - i \dd_{z_2} - 2\dd_\eps$ & $J_A$ & $K_B$\\
\rowcolor{light-gray}$(2,2)$ & $(1,-1)$ & $\ol \dd - \dd -2 \dd_\eps$ & $K_A$ & $(-J)_B$ \\
\rowcolor{light-gray}$(2,2)$ & $(-i,-i)$ & $\ol \dd + i\dd$ & $J_B$ & $K_A$\\
$(2,2)$ & $(1,1)$ & $\ol \dd + \dd_{z_1} - \dd_{z_2}$ & $K_B$ & $(-J)_A$ \\
\hline
\end{tabular}
\caption{Comparison between our twists and those discussed by Kapustin and Witten. Here the BV differential refers to a differential on the graded vector space $\Omega^{\bullet,\bullet}(\CC^2;\gg[\eps][1])$.  Rows highlighted in gray indicate those twists that are $\spin(4)$-invariant.}
\label{notation_table}
\end{table}

We will conclude this section with a few words on the incorporation of the complexified coupling constant.  Kapustin and Witten consider a $\bb{CP}^1$ family of 4-dimensional topologically twisted theories parameterized not by $t$ but by the so-called canonical parameter 
\[\Psi = \frac{\tau +\ol\tau }{2 } + \frac{\tau - \ol\tau  }{2} \left( \frac{t-t^{-1}}{t+ t^{-1}} \right).\]
When the $\theta$-angle vanishes, i.e. when $\tau$ is purely imaginary, this can be written more simply as 
\[\Psi = \tau \frac{t^2-1}{t^2+1}.\]
In the next section we will discuss geometric quantization of the classical BV theories that we have described above.  This depends not only on the classical moduli space, but also on some additional data, including a choice of pre-quantization that will depend on the shifted symplectic form. We will argue that in our framework of discussing field theory in terms of derived algebraic geometry, we can discuss the occurrence of the categories of twisted differential operators appearing in the quantum geometric Langlands conjecture without specifying the coupling constant of the theory.


\section{Application to Quantum Geometric Langlands} \label{QGL_section}

In this subsection, we discuss a method by which the statement of the quantum geometric Langlands correspondence appears by quantizing the stacks of classical solutions in various twists of $\mc N=4$ super Yang--Mills theory.  We describe an ansatz for categorified geometric quantization by analogy with the ordinary geometric quantization story.  This paper will not, however, provide a systematic study of this categorified geometric quantization procedure beyond a definition applicable for our examples (although subsequent to the appearance of the first version of this paper, a more complete approach to categorified geometric quantization was developed by \cite{SafronovGQ}).

\subsection{Recollections on Geometric Quantization}

Let us briefly recall the traditional story of geometric quantization for smooth symplectic manifolds. The input data is a symplectic manifold $(M,\omega)$, called the \emph{phase space}, and the output will be a Hilbert space $\mc H$.

First, one needs to fix the following data:
\begin{itemize}
\item (prequantization) A prequantization of $(M,\omega)$ is a choice of a Hermitian line bundle $L_M$ equipped with a Hermitian connection $\nabla$ on $M$ such that its curvature 2-form $F(\nabla)$ coincides with $\omega$.
\item (polarization) A polarization of $(M,\omega)$ is a Lagrangian foliation, meaning a foliation whose leaves are Lagrangian submanifolds. In particular, one has an integrable distribution $\mc F\subset TM$ such that $\mc F_p \subset T_p M$ is a Lagrangian subspace for each $p\in M$.
\end{itemize}

Then the desired Hilbert space $\mc H$ consists of those sections that are flat along the foliation, that  is, \[\mc H= \{ s  \in \Gamma(M,L_M) \mid \nabla_{X} s =0 \text{ for }X\in\Gamma(M,\mc F)\} .\]
In practice, when we categorify this definition it will be by an analogy to the simplest situation: where one has a Lagrangian foliation realized by a map $\pi \colon M\to N$ where $N$ is the space of leaves, such that each leaf is contractible and the Hermitian line bundle $L_M$ is of the form $\pi^* L_N$ for a line bundle $L_N$ over $N$. In this case, the Hilbert space can be realized essentially as $\mc H \cong  \Gamma(N, L_N )$.  Slightly more precisely, one has to introduce a metaplectic correction.  This means one considers a half-density line bundle $\dens^{1/2}_N$ and forms the space $L^2(N, L_N\otimes \dens^{1/2}_N)$ of square-integrable sections.  These subtleties are necessary in order to define the inner product on Hilbert space, but they don't play an important role in our ansatz.

\begin{example}
Consider $M=T^*N$ with the standard symplectic form $\omega_M$. The trivial complex line bundle with the canonical 1-form is a prequantization datum. For a polarization, we consider the projection $\pi \colon T^*N\to N$ and take $\mc F=\ker (d\pi  \colon T(T^*N)\to TN)$. The corresponding Hilbert space is then $\mc H \cong \Gamma(N, N\times \bb C) =:C^\infty(N,\bb C)$.
\end{example}

\begin{example}\label{ex:twisted_cotangent}
Consider $T^*N$. For a closed 2-form $\alpha$ on $N$, one can consider $\omega= \omega_M + \pi^* \alpha$; we denote such a symplectic manifold by $T^*_\alpha N$ and call it a twisted cotangent bundle. The map $\pi\colon M\to N$ still gives a Lagrangian foliation. If the cohomology class $[\frac{\alpha}{2\pi }]\in \mr H^2(N,\bb R)$ lies in the lattice $\mr H^2(N, \ZZ)$ then one can find a line bundle $L_N$ over $N$ with $c_1(L_N)=[ \frac{\alpha}{2\pi }] \in H^2(N,\bb Z)$ so that the Hermitian line bundle is $\pi^*  L_N$. The resulting Hilbert space is $\mc H\iso \Gamma(N,L_N)$.
\end{example}

\begin{remark} \label{twist_cotangent_rmk}
Note that for a holomorphic line bundle $L$ over a complex manifold $X$, the corresponding cohomology class (of the curvature for any choice of a Hermitian metric with corresponding Chern connection) is given by $c_1(L) \in \mr H^{1,1}(X)  \iso \mr H^1(X,\Omega^1_X)$.

Thus, for geometric quantization of a complex manifold $M$ with $[\omega] \in \mr H^{1,1}(M)$, one may want to look for a holomorphic morphism $\pi\colon M\to N$ together with a holomorphic line bundle $L_N$ over $N$. For instance, if $\alpha\in \mr H^1(N,\Omega^1_{\mr{cl}})$,  one may use it to consider the twisted cotangent bundle $T^*_\alpha N$ where the symplectic form is $\omega_{T^*N} + \pi^* \alpha $ and in particular its cohomology class is $\pi^*[\alpha]$. Then geometric quantization of twisted cotangent bundle is given by the space $\Gamma(N, L_N)$ of holomorphic sections of the holomorphic line bundle $L_N$ over $N$. In this context of a complex manifold, a half-density bundle is a square-root $K_N^{1/2}$ of the canonical bundle $K_N$.

Combining with the above example, we obtained a recipe for geometric quantization of a twisted cotangent bundle $M=T^*_\alpha N$ of a complex manifold $N$. That is, one identifies the cohomology class $[\omega]$ of its symplectic form $\omega$ as $[ \frac{ \omega}{2\pi }] =  [\pi^* c_1( L)]$ for a holomorphic line bundle $L$ over $N$ and then the result of geometric quantization is the space $\Gamma(N, L)$ of holomorphic sections of a holomorphic line bundle $L$ over $N$.
\end{remark}

\subsection{Gerbes and Twisted D-modules} \label{gerbe_section}

We will define categorified geometric quantization for 1-shifted symplectic algebraic stacks by analogy with the ordinary geometric quantization story.  Instead of considering functions twisted by a prequantum line bundle, we will consider sheaves twisted by a prequantum gerbe.  In this section we will explain what this means in a brief and informal way, and discuss how one obtains the category of twisted D-modules as an example of this idea. For more details on this topic, we refer to the paper of Gaitsgory and Rozenblyum \cite{GRCrystals}.

\begin{definition}
A \emph{$\bb G_m$-gerbe} on a scheme $S$ is a category $\mc G$ on which $\mr{Pic}(S)$ acts simply transitively. If we have a map of schemes $ f\colon S_1\to S_2$ and a gerbe $\mc G_2$ on $S_2$, we can construct a gerbe on $S_1$ by setting $\mc {G}_1 := \mr{Pic}(S_1) \times_{ \mr{Pic} (S_2)  }\mc {G}_2$. Then a \emph{$\bb G_m$-gerbe} on a prestack $\mc Y$ is defined to be a family of $\bb G_m$-gerbes $\mc G_{S,y}$ on $S$ for every $y\colon S\to \mc Y$ in a way compatible with pullbacks.
\end{definition}
The trivial gerbe $\mc G_{\mr{triv}}$ on $\mc Y$ is defined to be $(\mc G _{\mr{triv}} )_{S,y}=\mr{Pic}(S)$. 

Given a scheme $S$ and a $\bb G_m$-gerbe $\mc G$ on $S$, let us describe the category $\QC_{\mc G}(S)$ of $\mc G$-twisted quasi-coherent sheaves. The idea is that for $g_1,g_2\in \mc G(S)$, there exists a unique associated line bundle $\mc L_{g_1,g_2} \in \mr{Pic}(S)$ and that $\QC_{\mc G}(S)$ is locally equivalent to $\QC(S)$ by any choice of $g\in\mc G(S)$ (in the same way a line bundle is locally equivalent to the trivial bundle by any choice of a non-vanishing section). Using this, an object $\mc F\in \QC_{\mc G}(S)$ is a collection $\{(\mc F_{g_1},\mc F_{g_2}, \mc L_{g_1,g_2})\}_{g_1,g_2\in\mc G}$, where $\mc F_{g_i}\in\QC(S)$ satisfies $\mc F_{g_2} = \mc L_{g_1,g_2} \otimes  \mc F_{g_1}$; one should think of the following commutative diagram: 
\[\xymatrix{
 & \QC(S) \ar[dd]_\simeq^{\mc L_{g_1,g_2}\otimes (-)} \\
\QC_{\mc G}(S) \ar[ur]_\simeq^{g_1}  \ar[dr]^\simeq_{g_2} &\\
& \QC(S).
}\]
In particular, if $\mc G=\mc G_{\mr{triv}}$, then $g_i=\mc L_i \in \mr{Pic}(S)$ and $\mc L_{g_1,g_2} = \mc L_2\otimes \mc L_1^{-1} $, so $\mc F\in \QC(S)$ determines $\mc F\in \QC_{\mc G}(S)$ and hence $\QC_{\mc G_{\mr{triv}}}(S)= \QC(S)$. For a prestack $\mc Y$,  we define $\QC_{\mc G}(\mc Y)=\displaystyle\lim_{y\colon  S\to\mc Y} \QC_{\mc G_{S,y}}(S)$.

Gaitsgory and Rozenblyum demonstrated that the theory of twisted D-modules can be understood in terms of twisted sheaves.
\begin{definition}
A \emph{twisting} on $\mc Y$ is a $\bb G_m$-gerbe $\mc G$ on $\mc Y_{\mr{dR}}$ together with a choice of a trivialization $\alpha \colon p_{\mc Y}^*\mc G \iso \mc G_{\mr{triv}}$ under pullback along the canonical map $p_{\mc Y}\colon  \mc Y\to \mc Y_{\mr{dR}}$.  We can associate to a twisting $\mc G$ an element of $\mr H^2_{\mr{dR}}(\mc Y)$ called the \emph{curvature 2-form} \cite[Section 6.5]{GRCrystals}.
\end{definition}

Given a twisting datum, we define the category of twisted D-modules on $\mc Y$ as $\QC_{\mc G}(\mc Y_{\mr{dR}})$ together with a fixed functor 
\[\xymatrix{ \QC_{\mc G}(\mc Y_{\mr{dR}})\ar[r]^-{p_{\mc Y}^*} &  \QC_{ \mc G|_{\mc Y} }(\mc Y)\\
 \text{D}_{\mc G}(\mc Y)  \ar@{=}[u] \ar[r] &  \QC(\mc Y)  \ar@{=}[u]_{\alpha}}
\]
that describes the underlying quasi-coherent sheaf of the twisted D-module. We may abuse the notation to write the category as $\text{D}_{(\mc G,\alpha)}(\mc Y)$.

If $S$ is a smooth classical scheme, then a twisting in this sense is equivalent to the notion of a twisted differential operator \cite[Section 6.5]{GRCrystals}.

\begin{remark}
The reuse of the term ``twisting'' in this section is potentially misleading.  There is a priori no relationship between the theory of twisting for sheaves and D-modules and the theory of twisting for supersymmetric field theory (although, of course, there will turn out to be a connection for the specific example of twisted $\mc N=4$ super Yang--Mills theory, as we will see shortly).
\end{remark}

\begin{example}
Let $S$ be a smooth classical scheme and $\mc L$ be a line bundle on $S$. Consider the sheaf $\mc D_{\mc L}$ of differential operators twisted by $\mc L$. The category of $\mc D_{\mc L}$-modules can be realized in terms of a corresponding twisting on $S$. We consider the trivial gerbe $\mc {G}=\mc {G}_{\mr{triv}}$ on $ S_{\mr{dR}}$, but under the pullback to $S$ its trivialization corresponds to an automorphism of $\mc {G}_{\mr{triv}}$ given by tensoring by the line bundle $\mc L$ on $S$. In other words, $\text{D}_{(\mc G_{\mr{triv}},\mc L)}(\mc Y)=\text{D}_{\mc L}(\mc Y) $ is the category of twisted D-modules on $\mc L$.
\end{example} 

\begin{example}
There is a generalization of this procedure developed by Beilinson and Bernstein \cite{BBJantzen} where we twist by a power $\mc L^\kappa$ of a line bundle $\mc L$, where $\kappa$ is a complex number.  This power does not exist as a line bundle for generic values of $\kappa$, but the twisting still exists (see \cite[pp 82--83]{FrenkelLectures} for a clear explanation of the construction for the case of our interest, namely, twisted D-modules on the stack $\bun_G(C)$).
\end{example}

\subsection{Categorified Geometric Quantization}

Now let us discuss geometric quantization in the algebraic and categorified setting. For us, the input data is a 1-shifted symplectic space $(\mc  X,\omega)$ and the output is a DG category. Note that as $\omega$ is an algebraic symplectic form, it has a cohomology class $[\omega]\in  F^2\mr H^3_{\mr{dR}}(\mc X) = \mr H^{2,1}(\mc X)$. 

By analogy to the ordinary story of geometric quantization, we will fix two pieces of data:
\begin{itemize}
\item (prequantization) A prequantization of the 1-shifted symplectic space $(\mc X,\omega)$ is a choice of a $\bb G_m$-gerbe $\mc G_{\mc X}$ with connection on $\mc X$ whose curvature 2-form coincides with $\omega$.
\item (polarization) A polarization of $(\mc X,\omega)$ is a 1-shifted Lagrangian fibration $\pi \colon \mc  X\to \mc Y$.  More generally we could consider a 1-shifted Lagrangian foliation, using the notion of a foliation in derived geometry recently developed by Borisov, Sheshmani, and Yau \cite{BSY}, To\"en and Vezzosi \cite{ToenVezzosiFoliation} and Pantev, but we will not need this level of generality for the examples in this paper. 
\end{itemize}

Let us assume that the prequantum $\bb G_m$-gerbe $\mc G_{\mc X}$ is realized as the pullback $\pi^* \mc G$ for a $\bb G_m$-gerbe $\mc G$ on $\mc Y$.  Given this data, we define the categorified geometric quantization of $(\mc X, \omega)$ as follows.

\begin{definition}\label{cat_geom_quant_def}
The \emph{categorified geometric quantization} of $(\mc X, \omega)$ associated to the polarization $p$ and prequantum gerbe $\mc G_{\mc X} = p^*\mc G$ is the category $\QC_{\mc G }(\mc Y)$ of quasi-coherent sheaves on $\mc Y$ twisted by the gerbe $\mc G$.
\end{definition}

\begin{example} \label{1_cotangent_example}
Consider $\mc X=T^*[1]\mc Y$ with the standard $1$-shifted symplectic form. The trivial gerbe with the canonical 1-form is a prequantization data. A polarization is given by $\pi\colon T^*[1]\mc Y\to \mc Y$. The corresponding category is $\QC(Y)$.
\end{example}

Now let us try to find geometric quantization of the de Rham stack $\mc X_{\mr{dR}}$ of a derived stack $\mc X$. This may look troublesome because $\mc X_{\mr{dR}}$ is canonically $k$-shifted symplectic for any $k\in\ZZ$, just because its tangent complex is trivial. Accordingly, one may consider any polarization together with the trivial $k$-gerbe as prequantum data, whose geometric quantization won't of course be well-defined.

The observation that saves us is that a de Rham stack is to be understood as a part of a Hodge stack. That is, we consider the Hodge stack $\mc X_{\mr{Hod}}$ such that its special fiber is the Dolbeault stack  $\mc X_{\mr{Dol}} = T_f[1]\mc X$ and its generic fiber is the de Rham stack $\mc X_{\mr{dR}}$. Note that if $\mc X$ were $n$-shifted, then $T_f[1]\mc X \cong T^*_f[n+1]\mc X$ canonically equips the Dolbeault stack $\mc X_{\mr{Dol}}$ with the $(n+1)$-shifted symplectic structure, which is to suggest that one has to consider an $n$-shifted symplectic structure on $\mc X$ and compatibility with the natural map $\Phi\colon \mc X \to \mc X_{\mr{dR}}$ to perform geometric quantization on $\mc X_{\mr{Dol}}$ and also on $\mc X_{\mr{dR}}$ as $(n+1)$-shifted symplectic space.

The case of interest for us is $n=1$. Then we have to use the geometric quantization data of 0-shifted symplectic space $(\mc X,\omega)$. The categorified geometric quantization of $(\mc X_{\mr{dR}},\mc X, \omega)$, where $(\mc X,\omega)$ is a 0-shifted symplectic space, should be using the geometric quantization data of $(\mc X,\omega)$ as a 0-shifted symplectic stack in a way compatible with the one of $\mc X_{\mr{dR}}$, that is, one should consider a polarization coming from a map $\pi \colon \mc X \to \mc Y$ (and hence a compatible polarization $\Pi \colon \mc X_{\mr{dR}}\to \mc Y_{\mr{dR}}$ as well in a trivial manner) making the commutative diagram
\[\xymatrix{
\mc X\ar[r]^-\Phi \ar[d]_{\pi} &\mc X_{\mr{dR}} \ar[d]^{\Pi} \\
\mc Y \ar[r]_-\phi &\mc Y_{\mr{dR}},
}\] together with a prequantum line bundle $L$ of $\mc X$, which can be regarded as a trivialization of the pullback $\Phi^*\mc G_{\mc X,\mr{triv}}$ of the trivial gerbe $\mc G_{\mc X,\mr{triv}}$ over $\mc X_{\mr{dR}}$. When the prequantum line bundle $L$ is of the form $L=\pi^* L'$ for a line bundle $L'$ over $\mc Y$, it can be understood as the trivialization of $\phi^*\mc G_{\mc Y,\mr{triv}}$, which defines a twisting on $\mc Y$. This motivates the following definition.

\begin{definition}\label{de_rham_example}
The \emph{categorified geometric quantization} of $(\mc X_{\mr{dR}},\mc X, \omega)$, associated to the geometric quantization data of 0-shifted symplectic space $(\mc X,\omega)$ given by the polarization $\pi\colon \mc X\to \mc Y$ and a prequantum line bundle $L$ of the form $L = \pi^*L'$, is the category $\text{D}_{ (\mc G_{\mc Y,\mr{triv} } , L' )  } (\mc Y)=\text{D}_{L'}(\mc Y)$.
\end{definition}

\begin{remark}[Metaplectic correction] \label{metaplectic_remark}
From now on we will modify our definition slightly for examples of this form, coming from the quantization of a 1-shifted symplectic de Rham stack.  Instead of using the line bundle $L$ prequantizing the 0-shifted symplectic structure on $\mc X$, we will impose a \emph{metaplectic correction}: we will choose a square-root $K^{1/2}$ of the canonical bundle on $\mc Y$, and denote a \emph{corrected} geometric quantization by $\text D_{L }(\mc Y)$ (which would have been written as $\text D_{L \otimes K^{1/2}}(\mc Y)$ in the convention until now).
\end{remark}

\subsection{Geometric Quantization of Twists of 4d \texorpdfstring{$\mc N=4$}{N=4} Gauge Theory}\label{subsub:quantization examples}

Now, having introduced the general notion of categorified geometric quantization, let us apply it to the examples at hand, namely to the Kapustin--Witten twists of $\mc N=4$ super Yang--Mills theory.  That is, we will consider the family of 1-shifted symplectic stacks over $\CC[\lambda,\mu]$ given by
\[\mr{Map}(\bb C_{\mr{dR}} \times C_{\mr{Hod}}, BG)_{\mr{Hod}} \iso \mr{Map}(C_{\mr{Hod}}, BG)_{\mr{Hod}},\]
where $\lambda$ and $\mu$ parameterize the Hodge families inside and outside the mapping stack respectively.  We can understand the 1-shifted symplectic structures following Definition \ref{de_rham_example}, using the 0-shifted symplectic form on the mapping stack $\mr{Map}(C_{\mr{Hod}}, BG)$ defined by transgression (that is, the AKSZ symplectic structure studied in \cite{PTVV}).  

\begin{remark}
Note that by requiring an $\bb C_{\mr{dR}}$-factor in the source, we are excluding rank $(1,0)$ and $(0,1)$ twists, along with their deformations to rank $(2,0)$ and $(0,2)$ twists.  The 1-shifted symplectic structure is more difficult to describe without these de Rham directions.  Such twists are self-dual, and do not appear in the geometric Langlands correspondence as described in \cite{KapustinWitten}, although they are considered in Kapustin's approach in \cite{Kapustinnote}. 
\end{remark}

\begin{remark}
Concretely, when we fix $\lambda$, the mapping stack $\mr{Map}(C_{\mr{Hod}}, BG)$ becomes the stack $\flat_G^\lambda(C)$ of flat $\lambda$-connections, that is, the Hitchin moduli space on $C$ with complex structure $I_\lambda$. The AKSZ symplectic structure we are considering is complex algebraic. This is what Kapustin and Witten would denote by $\Omega_{I_\lambda}$, which is a $(2,0)$-form in complex structure $I_\lambda$. We should distinguish this from the K\"ahler structure which they would denote by $\omega_{I_\lambda}$, which is a real $(1,1)$-form in the complex structure $I_\lambda$.
\end{remark}

Let us consider first the case where $\mu = 0$, so we are quantizing the 1-shifted cotangent space $T^*[1]\flat_G^\lambda(C)$.  We use the trivial prequantum gerbe, and the polarization given by the projection $\pi \colon T^*[1]\flat_G^\lambda(C) \to \flat_G^\lambda(C)$.  The resulting category is, therefore, $\QC(\flat_G^\lambda(C))$, as in Example \ref{1_cotangent_example}.  The case $\mu = 0$ occurs in our family of supersymmetric twists in the following two examples:

\begin{example}[Rank $(1,1)$] \label{rank (1,1)}
 The moduli space associated to the pair $(\lambda,\mu)=(0,0)$ occurs when we consider the twist by a rank $(1,1)$ supercharge.  This is also known as the \emph{Kapustin twist} \cite{KapustinHolo}.  The category the Kapustin twist of $\mc N=4$ super Yang--Mills theory assigns to $\RR^2 \times C$ is $\QC(\higgs_G(C))$.
\end{example}

\begin{example}[Special Rank $(2,2)$] \label{rank special (2,2)}
	 If $\lambda \ne 0$, the moduli space associated to the pair $(\lambda, 0)$ occurs when we consider the twist by special rank $(2,2)$ supercharges of B-type.  Recall from Section \ref{22_section} that we obtain a B-type twist, i.e., $\mu = 0$, when we consider a supercharge 
 \[Q_{\mr{hol}}+ aQ'_{\mr{hol}} + b_1Q' + b_2Q'' \]
 with $b_1 + ab_2 = 0$.  Assuming that $a \ne 0$, in these terms, the parameter $\lambda = -b_2$. Since $\flat^\lambda_G(C)\cong \flat_G(C)$ as derived stacks for any $\lambda \ne 0$, the category we obtain when we quantize is then $\QC(\flat_G(C))$. 
\end{example}

\begin{remark}
We should note that this is not the category occurring in the geometric Langlands correspondence: instead of quasi-coherent sheaves, in order to obtain a category which could plausibly be Langlands dual to the category of D-modules we would need to consider ind-coherent sheaves with a nilpotent singular support condition \cite{ArinkinGaitsgory}.  We studied the meaning of such singular support conditions in supersymmetric gauge theory in \cite{EY2}, and a complete analysis of the connection between twisted supersymmetric gauge theory and the geometric Langlands program would have to combine these approaches.	
\end{remark}

Now, suppose $\mu \ne 0$ and for the moment let us set $\mu = 1$. As discussed above, we consider the commutative square
\[\xymatrix{
\flat^\lambda_G(C)\ar[r]^-\Phi \ar[d]_{\pi_\lambda} &\flat^\lambda_G(C)_{\mr{dR}} \ar[d]^{\Pi_\lambda} \\
\bun_G(C) \ar[r]_-\phi &\bun_G(C)_{\mr{dR}}.
}\]
The claim is that we find a prequantum twisting $L_{\lambda}$ for the 0-shifted symplectic structure $\Omega_{\lambda} =\Omega_{I_\lambda }$ on $\flat^\lambda_G(C)$ and moreover we will see in Theorem \ref{geom_quant_line_bundles_thm} below that there is a twisting $L'_{\lambda}$ on $\bun_G(C)$ so that $L_{\lambda} \iso \pi_\lambda^*L'_{\lambda}$. 

There is a natural line bundle on $\bun_G(C)$ called the \emph{determinant line bundle}.  This line bundle has the property that the stack $\flat_G(C)$ can be identified as the twisted cotangent space to $\bun_G(C)$ by this line bundle, as in Remark \ref{twist_cotangent_rmk} (see \cite[Proposition 4.1.4]{BZFrenkelSS} and the references given there).  When we restrict to coarse moduli spaces, the equivalence between $T^*_{ L_{\mr{det}}} \bun_G(C)$ and $\flat_G(C)$ is, furthermore, a symplectomorphism (this is essentially the prequantization of complex Chern--Simons theory originally considered by Witten \cite{WittencomplexCS}).  This symplectic structure can be described very concretely. Recall that the tangent complex to the stack $\flat_G(C)$ at a point $(P,\mc A)$ can be described as the shifted de Rham complex $(\Omega^\bullet(C; \gg_P)[1], \d_{\mc A})$.  The pairing on the tangent complex coming from the holomorphic symplectic structure $\Omega$ on $\flat_G(C)$ is inherited from the wedge-and-integrate pairing on differential forms, and the invariant pairing on $\gg$: concretely it is given by 
\[ \langle \delta_1 \mc A, \delta_2 \mc A \rangle_{\Omega} = -\frac{1}{4\pi} \int_C \langle \delta_1 \mc A \wedge \delta_2 \mc A \rangle.\] 
On the other hand, for our purpose, we would like to understand $\Omega$ as part of the family $\Omega_{I_w}$ of holomorphic symplectic structures.

We recall the following from \cite{KapustinWitten}. The Hitchin moduli space is a hyper-K\"ahler manifold. It has complex structures $I$, $J$, and $K$, where linear holomorphic functions in those complex structures are evaluations at points of $C$ of $(A_{\ol z} ,\phi_z)$, $\mc A=A+i\phi $, and $(A_{\ol z } - \phi_{\ol z} , A_z +\phi_z)$, respectively. Indeed, the $\bb{CP}^1_h$-family of complex structures can be parametrized as
\[I_\lambda = \frac{1-\ol \lambda \lambda }{1+ \ol \lambda \lambda}I + \frac{i(\lambda-\ol \lambda)}{1+\ol \lambda \lambda}J + \frac{\lambda + \ol \lambda}{1 + \ol \lambda \lambda}K \]
with respect to which $(A_{\ol z } - \lambda\phi_{\ol z} , A_z + \lambda^{-1} \phi_z)$ are holomorphic. Note that the cases of $\lambda=0$, $\lambda=-i$, $\lambda=1$ correspond to the triple of complex structures $I,J,K$, respectively.

Now for symplectic structures, writing $|\d^2z| = idz \wedge  \d \ol  z$, we have the following explicit expressions for K\"ahler and holomorphic symplectic forms.
\begin{align*}
\omega_I & = -\frac{i}{2\pi  } \int_C  |\d^2z| \tr  (\delta A_{\ol z}\wedge \delta A_{z} - \delta  \phi_{\ol z }\wedge \delta \phi_z		 ) = -\frac{1}{4\pi }  \int_C  \tr  (\delta A  \wedge  \delta A -  \delta \phi \wedge \delta \phi )\\
\omega_J  & = \frac{1}{2\pi } \int_C |\d^2z|  \tr( \delta \phi_{\ol z}\wedge \delta A_z  + \delta \phi_z \wedge \delta A_{\ol z} )\\
\omega_K & = \frac{i}{2\pi }\int_C  |\d^2z| \tr( \delta \phi_{\ol z}\wedge \delta A_z  - \delta \phi_z \wedge \delta A_{\ol z}) = \frac{1}{2\pi }\int_C \tr \delta \phi \wedge \delta A\\
\Omega_I & = \omega_J+i\omega_K = \frac{1}{\pi} \int_C |\d^2z|\tr \delta \phi_z \wedge  \delta A_{\ol z}\\
\Omega_J & =\omega_K + i\omega_I = -\frac{i}{4\pi} \int_C \tr (\delta \mc A \wedge \delta \mc A) = \frac{i}{2\pi } \int_C |\d^2z| \tr (  \delta \phi_{\ol z}\wedge \delta A_z  - \delta \phi_z \wedge \delta A_{\ol z}  - i\delta A_{\ol z}\wedge \delta A_{z} +i \delta  \phi_{\ol z }\wedge \delta \phi_z	  )  \\
\Omega_K &  =\omega_I + i\omega_J = -\frac{i}{2\pi }\int_C  |\d^2z|\tr ( \delta A_{\ol z}\wedge \delta A_{z} - \delta  \phi_{\ol z }\wedge \delta \phi_z	 - \delta \phi_{\ol z}\wedge \delta A_z  - \delta \phi_z \wedge \delta A_{\ol z}  )
\end{align*}
Extrapolating to the full $\bb{CP}^1$ family, one obtains the expression
\begin{align*}
 \Omega_{I_\lambda}& = -\frac{i}{2\pi} \int_C |\d^2z|  \lambda \tr \left(   \delta(A_{\ol z } - \lambda\phi_{\ol z}) \wedge \delta( A_z + \lambda^{-1} \phi_z )   \right)\\
 & =  -\frac{i}{2\pi} \int_C |\d^2z|  \tr ( \lambda\delta A_{\ol z}\wedge \delta A_{z} - \lambda\delta  \phi_{\ol z }\wedge \delta \phi_z	 - \lambda^2 \delta \phi_{\ol z}\wedge \delta A_z  - \delta \phi_z \wedge \delta A_{\ol z} )
\end{align*}
where we have $\Omega_{I_{\lambda=0}} = \frac{i}{2}\Omega_I$, $\Omega_{I_{\lambda=-i}} = -\Omega_J = -i\Omega $, and $\Omega_{I_{\lambda=1}} =\Omega_K$.  Here $\omega_J,\omega_K$ are exact, while $\omega_I$ is not. One has $\omega_I= \omega_I' + \delta \lambda_I$ where
\begin{align*}
 \omega_I' & = -\frac{i}{2\pi  } \int_C  |d^2z| \tr  \delta A_{\ol z}\wedge \delta A_{z} 	  = -\frac{1}{4\pi }  \int_C  \tr   \delta A  \wedge  \delta A\\
 	\lambda_I & = \frac{1}{4\pi } \int_C \tr \phi \wedge \delta \phi.
\end{align*}
In this terminology, one can identify a generator of $\text{Pic}(\bun_G(C))\cong \bb Z $ as $L_{\mr{det}}$, for which one has $ [p^* c_1(L_{\mr{det}})] = [ \frac{\omega_I'}{2\pi }] \in \mr H^2(\mc M_H(G,C), \bb Z )$ where $p \colon \mc M_H(G,C) \to  \bun_G(C)$ is the map that forgets the Higgs field. Then the fact that $\Flat_G(C)$ is a twisted cotangent bundle over $\bun_G(C)$ boils down to the observation that \[[\Omega] = [ - i \Omega_J   ] = [\omega_I] = [\omega_I']\] in view of Remark \ref{twist_cotangent_rmk}.

Now, the Hitchin moduli space $ \mc M_H(G,C)$ with complex structure $I_\lambda$ is the moduli space of $G$-bundles with flat $\lambda$-connection on $C$. An important thing to note is
\[[\Omega_{I_\lambda } ] = [\lambda  \omega_I] =  \lambda [\omega_I'] \]
and hence the moduli space $\Flat^\lambda_G(C)$ of $\lambda$-connections on $C$ is the twisted cotangent bundle of $\bun_G(C)$ using the line bundle $L_{\mr{det}}^{\lambda}$.

In our argument below, we will make the assumption that this statement extends to the derived level.  In other words, we will assume the following.

\begin{assumption} \label{expectation}
We will assume that the equivalence between $T^*_{ L_{\mr{det}}} \bun_G(C)$ and $\flat_G(C)$ is an equivalence of 0-shifted symplectic stacks.  This is the subject of work in progress by Calaque and Safronov, who prove the Dolbeault analogue of this statement: that the equivalence $T^*\bun_G(C) \iso \mr{Map}(C_{\mr{Dol}}, BG)$ is a symplectomorphism -- the desired statement here is a Hodge deformation of their result.\footnote{We are grateful to P. Safronov for discussing his work in progress with us.}
\end{assumption}

Now we have everything we need to quantize: we have a trivial gerbe $\mc G'_{\lambda}$ on $\bun_G(C)_{\mr{dR}}$, and a trivialization of the pullback $\phi^*\mc G'_{\lambda}$ given by a twisting $L'_{\lambda}$ on $\bun_G(C)$.  We can understand this quantization concretely using the following result.

\begin{theorem} \label{geom_quant_line_bundles_thm}
In the commutative square
\[\xymatrix{
\flat^\lambda_G(C)\ar[r]^-\Phi \ar[d]_{\pi_\lambda} &\flat^\lambda_G(C)_{\mu\text{-dR}} \ar[d]^{\Pi_\lambda} \\
\bun_G(C) \ar[r]_-\phi &\bun_G(C)_{\mu\text{-dR}}.
}\]
the resulting twisting line bundle $L'$ on $\bun_G(C)$ is isomorphic to a power $L_{\mr{det}}^\kappa$ of the determinant line, where $\kappa =  \frac{\lambda}{\mu}$.
\end{theorem}

If we apply this result to the definition of the categorical geometric quantization as given in Definition \ref{de_rham_example}, the following corollary is immediate.

\begin{corollary} \label{twisted_Dmod_cor}
The category $\text{D}_{\kappa}(\bun_G(C))$ can be obtained by geometric quantization from the 1-shifted symplectic stack $\flat_G^\lambda(C)_{\mu\text{-dR}}$ induced from the holomorphic symplectic structure $\Omega_\lambda$ of $\flat_G^\lambda(C)$ where $\kappa = \frac{\lambda}{\mu}$.
\end{corollary}

\begin{proof}[Proof of Theorem \ref{geom_quant_line_bundles_thm}]
We begin with the case $\mu = 1$, which follows from Assumption \ref{expectation}.  Indeed, if $\flat_G^\lambda (C)$ is identified with the twisted cotangent bundle to $\bun_G(C)$ by $L_{\mr{det}}^\lambda$ as a 0-shifted symplectic stack, then in particular it admits a prequantization by $L_{\mr{det}}^\lambda$.

Now the tangent complex to $\flat_G(C)_{\mu\text{-dR}}$ can be identified relative to $\CC[\mu]$ with the complex $(\Omega^\bullet(C; \gg_P)[\eps],  \d_{\mc A} + \mu \dd_\eps)[1]$, where $\eps$ is a degree 1 parameter.  If we write $F_\mu \colon \flat_G(C)_{\mr{dR}} \overset \iso \to \flat_G(C)_{\mu\text{-dR}}$, we can identify the pullback $F_\mu^*\wt \Omega_\mu $ with the rescaling $\mu^{-1}\wt \Omega_{\mu = 1}$, where $\wt \Omega_\mu $ is the 1-shifted symplectic structure of $\Flat_G(C)_{\mu\text{-dR}}$. Although the symplectic structures are trivial for the de Rham stack, this is a meaningful statement in terms of the Hodge family. This is induced from the corresponding rescaling of 0-shifted symplectic structure on $\Flat_G(C)$. Thus, when we rescale $\mu$, the prequantization is given by the twisting $L_{\mr{det}}^{\lambda/\mu}$.
\end{proof}

\begin{example}[Rank $(2,1)$]\label{rank (2,1)}
Let us consider the case where $\lambda = 0$, but $\mu \ne 0$.  This occurs when we consider the twist by a rank $(2,1)$ supercharge of the form $Q_{\mr{hol}}+ aQ'_{\mr{hol}} + b_1Q' $.  As we learned in Section \ref{22_section}, the moduli space of solutions on $\bb C \times C$ in this twisted theory takes the form $\higgs_G(C)_{-b_1\text{-dR}}$. According to Corollary \ref{twisted_Dmod_cor}, we obtain the category $\text{D}(\bun_G(C))$ regardless of $\mu$.
\end{example}

\begin{example}[Generic rank $(2,2)$]\label{rank generic (2,2)}
Now, let us consider the case where both $\lambda$ and $\mu$ are non-zero.  This occurs for generic supercharges $Q_{\mr{hol}}+ aQ'_{\mr{hol}} + b_1Q' + b_2Q'' $ of rank $(2,2)$.  As we learned in Section \ref{22_section}, the moduli space of solutions on $\bb C \times C$ in this twisted theory takes the form $\flat^\lambda_G(C)_{\mu\text{-dR}}$ where $\lambda = -b_2$ and $\mu = b_1 + ab_2$.  According to Corollary \ref{twisted_Dmod_cor}, when we quantize this moduli stack with coupling constant $\tau$, we obtain the category $\text{D}_{\kappa }(\bun_G(C))$ where $\kappa = \lambda/\mu$.
\end{example}

\begin{remark}\label{comparison with KW}
Twisted S-duality of \cite{RaghavendranYoo} gives a concrete map on the space of all possible deformations of the Kapustin twist of 4d $\mc N=4$ supersymmetric gauge theory for $G=\GL_n$, which is interpreted as the space of closed string states. In particular, it sends $[\lambda,\mu] \mapsto [\mu,-\lambda ]$. Combined with the discussion of the current paper, it recovers most of the known forms of the global categorical geometric Langlands conjectures:
\begin{itemize}
	\item Consider Example \ref{rank (1,1)} where we had $(\lambda,\mu)=(0,0)$. This twist is physically known to be self-dual, and indeed so in the context of twisted S-duality. The associated conjectural equivalence of categories we obtain as the result of geometric quantization is \[\QC(\higgs_G(C)) \iso \QC(\higgs_{G^\vee}(C)).\] This is the so-called classical limit of the geometric Langlands correspondence or Dolbeault geometric Langlands correspondence, as conjectured by Donagi and Pantev \cite{DonagiPantev}.	
  	\item Consider Examples \ref{rank special (2,2)} and \ref{rank (2,1)}. Twisted S-duality exchanges $(a,0) \leftrightarrow (0,-a)$. This predicts the de Rham geometric Langlands equivalence
\[\text{D}(\bun_G(C)) \iso \QC(\flat_{G^\vee}(C)),\]
modulo issues of singular support as discussed in \cite{EY2}.
  	\item Consider Example \ref{rank generic (2,2)}. Twisted S-duality exchanges $\kappa \leftrightarrow - \frac{1}{\kappa}$. Upon categorified geometric quantization, this leads to a conjecture \[\text{D}_\kappa(\bun_G(C))\simeq \text{D}_{ -\frac{1}{\kappa}}(\bun_{\check{G}}(C))\] which is the statement of the quantum geometric Langlands correspondence.
\end{itemize}
We should comment on the difference between the parameter $\lambda/\mu$ and Kapustin--Witten's $\Psi$ parameter. In our analysis, we didn't incorporate the additional structure coming from the coupling constant of the theory.  A more elaborate analysis would include this additional structure as a deformation of the AKSZ shifted symplectic structure. We will not attempt to analyze this further in the present work.
\end{remark}

\pagestyle{bib}
\printbibliography

@incollection{CostelloSH,
	Author = {Costello, Kevin},
	Booktitle = {Special Issue: In Honor of Dennis Sullivan},
	Number = {1},
	Series = {Pure and Applied Mathematics Quarterly},
	Title = {Notes on supersymmetric and holomorphic field theories in dimensions 2 and 4},
	Volume = {9},
	Year = {2013}}

@article{KapustinWitten,
	Author = {Kapustin, Anton and Witten, Edward},
	Journal = {Communications in Number Theory and Physics},
	Pages = {1--236},
	Title = {Electric-magnetic duality and the geometric {Langlands} program},
	Volume = {1},
	Year = {2007}}

@article{ArinkinGaitsgory,
	Author = {Arinkin, Dima and Gaitsgory, Dennis},
	Journal = {Selecta Mathematica},
	Number = {1},
	Pages = {1--199},
	Publisher = {Springer},
	Title = {Singular support of coherent sheaves and the geometric {Langlands} conjecture},
	Volume = {21},
	Year = {2015}}

@book{CostelloBook,
	Author = {Kevin Costello},
	Optseries = {Mathematical Surveys and Monographs},
	Publisher = {AMS},
	Title = {Renormalization and Effective Field Theory},
	Volume = {170},
	Year = {2011}}

@book{CostelloGwilliam2,
	Author = {Costello, Kevin and Gwilliam, Owen},
	Title = {Factorization algebras in quantum field theory. Vol. 2},
	Url = {https://people.math.umass.edu/~gwilliam/vol2may8.pdf},
	Year = {2018}}

@article{WittenTQFT,
	Author = {Witten, Edward},
	Journal = {Communications in Mathematical Physics},
	Number = {3},
	Pages = {353--386},
	Publisher = {Springer},
	Title = {Topological quantum field theory},
	Volume = {117},
	Year = {1988}}

@article{BMS,
	Author = {Boels, Rutger and Mason, Lionel and Skinner, David},
	Journal = {Journal of High Energy Physics},
	Number = {02},
	Pages = {014},
	Publisher = {IOP Publishing},
	Title = {Supersymmetric gauge theories in twistor space},
	Volume = {2007},
	Year = {2007}}

@article{DAGX,
	Author = {Lurie, Jacob},
	Journal = {available at author's website: {http://www.math.harvard.edu/lurie}},
	Title = {Derived Algebraic Geometry {X}: Formal Moduli Problems},
	Year = {2011}}

@article{KapustinHolo,
	Author = {Kapustin, Anton},
	archiveprefix = {arXiv},
	eprint = {0612119},
	primaryclass = {hep-th},
	Title = {Holomorphic reduction of {$N=2$} gauge theories, {Wilson}-'t {Hooft} operators, and {S}-duality},
	Year = {2006}}

@article{PTVV,
	Author = {Pantev, Tony and To{\"e}n, Bertrand and Vaqui{\'e}, Michel and Vezzosi, Gabriele},
	Journal = {Publications math{\'e}matiques de l'IH{\'E}S},
	Number = {1},
	Pages = {271--328},
	Publisher = {Springer},
	Title = {Shifted symplectic structures},
	Volume = {117},
	Year = {2013}}

@article{Simpson,
	Author = {Simpson, Carlos},
	Journal = {Lecture Notes in Pure and Applied Mathematics},
	Pages = {229--264},
	Publisher = {MARCEL DEKKER AG},
	Title = {Homotopy over the complex numbers and generalized de {Rham} cohomology},
	Year = {1996}}

@article{PolishchukRothstein,
	Author = {Polishchuk, Alexander and Rothstein, Mitchell},
	Journal = {Duke Mathematical Journal},
	Number = {1},
	Pages = {123--146},
	Publisher = {Duke University Press},
	Title = {Fourier transform for {D}-algebras, {I}},
	Volume = {109},
	Year = {2001}}

@article{Kapustinnote,
	Author = {Kapustin, Anton},
	archiveprefix = {arXiv},
	eprint = {0811.3264},
	Title = {A note on quantum geometric {Langlands} duality, gauge theory, and quantization of the moduli space of flat connections},
	Year = {2008}}

@article{GRCrystals,
	Author = {Gaitsgory, Dennis and Rozenblyum, Nick},
	archiveprefix = {arXiv},
	eprint = {1111.2087},
	Title = {Crystals and {D}-modules},
	Year = {2011}}

@article{Movshev,
	Author = {Movshev, Mikhail},
	archiveprefix = {arXiv},
	eprint = {0812.0224},
	Title = {A Note on Self-Dual {Yang}-{Mills} Theory},
	Year = {2008}}

@article{EY1,
	Author = {Elliott, Chris and Yoo, Philsang},
	Journal = {Advances in Theoretical and Mathematical Physics},
	Title = {Geometric {Langlands} Twists of {$N=4$} Gauge Theory from Derived Algebraic Geometry},
	Volume={22},
	Number={3},
	Pages={615--708},
	Year = {2018}}

@incollection{FrenkelLectures,
	Author = {Frenkel, Edward},
	Booktitle = {Frontiers in number theory, physics, and geometry II},
	Pages = {387--533},
	Publisher = {Springer},
	Title = {Lectures on the {Langlands program} and conformal field theory},
	Year = {2007}}

@article{Wallbridge,
	Author = {Wallbridge, James},
	archiveprefix = {arXiv},
	eprint = {1610.00441},
	Title = {Derived smooth stacks and prequantum categories},
	Year = {2016}}

@article{CostelloLiSUGRA,
	Author = {Costello, Kevin and Li, Si},
	Journal = {arXiv preprint arXiv:1606.00365},
	Title = {Twisted supergravity and its quantization},
	Year = {2016}}

@article{EY2,
	Author = {Elliott, Chris and Yoo, Philsang},
	Journal = {Communications in Mathematical Physics},
	Title = {A Physical Origin for Singular Support Conditions in Geometric {Langlands}},
	Year = {2019},
	volume= {368}, Issue = {3}, Pages = {985--1050}}

@article{ElliottSafronov,
	Archiveprefix = {arXiv},
	Author = {Elliott, Chris and Safronov, Pavel},
	Doi = {10.1007/s00220-019-03393-9},
	Eprint = {1805.10806},
	Fjournal = {Comm. Math. Phys.},
	Issn = {0010-3616},
	Journal = {Comm. Math. Phys.},
	Mrclass = {81T60 (17B37 18D50 18F20)},
	Mrnumber = {4019918},
	Number = {2},
	Pages = {727--786},
	Primaryclass = {math-ph},
	Title = {Topological twists of supersymmetric algebras of observables},
	Volume = {371},
	Year = {2019}}

@article {WittenTwistor,
    AUTHOR = {Witten, Edward},
     TITLE = {Perturbative gauge theory as a string theory in twistor space},
   JOURNAL = {Comm. Math. Phys.},
  FJOURNAL = {Communications in Mathematical Physics},
    VOLUME = {252},
      YEAR = {2004},
    NUMBER = {1-3},
     PAGES = {189--258},
      ISSN = {0010-3616},
   MRCLASS = {81T60 (32L25 32L81 81R25 81T13 81T18 81T30 81T45)},
  MRNUMBER = {2104879},
MRREVIEWER = {Andrew Neitzke},
       DOI = {10.1007/s00220-004-1187-3},
       URL = {https://doi.org/10.1007/s00220-004-1187-3},
}

@article {ElliottSafronovWilliams,
    AUTHOR = {Elliott, Chris and Safronov, Pavel and Williams, Brian R.},
     TITLE = {A taxonomy of twists of supersymmetric {Y}ang--{M}ills theory},
   JOURNAL = {Selecta Math. (N.S.)},
  FJOURNAL = {Selecta Mathematica. New Series},
    VOLUME = {28},
      YEAR = {2022},
    NUMBER = {4},
     PAGES = {Paper No. 73},
      ISSN = {1022-1824},
   MRCLASS = {81T60 (14D21 70S15 81T13)},
  MRNUMBER = {4468561},
       DOI = {10.1007/s00029-022-00786-y},
       URL = {https://doi.org/10.1007/s00029-022-00786-y},
}

@article{BatalinVilkovisky,
	Author = {Batalin, Igor and Vilkovisky, Grigori},
	Doi = {10.1016/0370-2693(81)90205-7},
	Fjournal = {Physics Letters. B},
	Issn = {0370-2693},
	Journal = {Phys. Lett. B},
	Mrclass = {81E99 (58F06 81G20)},
	Mrnumber = {616572},
	Number = {1},
	Pages = {27--31},
	Title = {Gauge algebra and quantization},
	Volume = {102},
	Year = {1981}}

@article{EagerSaberiWalcher,
	Archiveprefix = {arXiv},
	Author = {Eager, Richard and Saberi, Ingmar and Walcher, Johannes},
	Eprint = {1807.03766},
	Primaryclass = {hep-th},
	Title = {Nilpotence varieties},
	Year = 2018}

@article{PridhamFMP,
	Archiveprefix = {arXiv},
	Author = {Pridham, J. P.},
	Doi = {10.1016/j.aim.2009.12.009},
	Eprint = {0705.0344},
	Fjournal = {Advances in Mathematics},
	Issn = {0001-8708},
	Journal = {Adv. Math.},
	Mrclass = {14B12 (16E35)},
	Mrnumber = {2628795},
	Mrreviewer = {Eivind Eriksen},
	Number = {3},
	Pages = {772--826},
	Primaryclass = {math.AG},
	Title = {Unifying derived deformation theories},
	Volume = {224},
	Year = {2010}}

@incollection{Toen,
	Author = {To{\"{e}}n, Bertrand},
	Fjournal = {Ast\'{e}risque},
	Isbn = {978-2-85629-855-8},
	Issn = {0303-1179},
	Journal = {Ast\'{e}risque},
	Mrclass = {14B12 (14F05 18F05 18G99)},
	Mrnumber = {3666027},
	Mrreviewer = {Montserrat Teixidor i Bigas},
	Note = {S\'{e}minaire Bourbaki. Vol. 2015/2016. Expos\'{e}s 1104--1119},
	Number = {390},
	Pages = {Exp. No. 1111, 199--244},
	Title = {Probl\`emes de modules formels},
	Year = {2017}}

@article{ElliottWilliamsYoo,
	Archiveprefix = {arXiv},
	Author = {Elliott, Chris and Williams, Brian R and Yoo, Philsang},
	Doi = {10.1016/j.geomphys.2017.08.009},
	Eprint = {1702.05973},
	Fjournal = {Journal of Geometry and Physics},
	Issn = {0393-0440},
	Journal = {J. Geom. Phys.},
	Mrclass = {81T13 (53D50)},
	Mrnumber = {3724786},
	Pages = {246--283},
	Primaryclass = {math-ph},
	Title = {Asymptotic freedom in the {BV} formalism},
	Volume = {123},
	Year = {2018}}

@article{Stoyanovsky,
  title={Quantum {L}anglands duality and conformal field theory},
  author={Stoyanovsky, Alexander},
  Archiveprefix={arXiv},
  primaryclass={math},
  eprint={0610974},
  year={2006}
}

@incollection {BBJantzen,
    AUTHOR = {Beilinson, A. and Bernstein, J.},
     TITLE = {A proof of {J}antzen conjectures},
 BOOKTITLE = {I. {M}. {G}elfand {S}eminar},
    SERIES = {Adv. Soviet Math.},
    VOLUME = {16},
     PAGES = {1--50},
 PUBLISHER = {Amer. Math. Soc., Providence, RI},
      YEAR = {1993},
   MRCLASS = {22E47 (14A22 14F05)},
  MRNUMBER = {1237825},
MRREVIEWER = {D. Mili\v{c}i\'{c}},
}

@article{Zhao1,
  title={Quantum parameters of the geometric {L}anglands theory},
  author={Zhao, Yifei},
  archiveprefix={arXiv},
  eprint={1708.05108},
  year={2017}
}

@article{Zhao2,
  title={Tame twistings and {$\Theta$}-data},
  author={Zhao, Yifei},
  archiveprefix={arXiv},
  eprint={2004.09671},
  year={2020}
}

@article {GaitsgoryLysenko,
    AUTHOR = {Gaitsgory, D. and Lysenko, S.},
     TITLE = {Parameters and duality for the metaplectic geometric
              {L}anglands theory},
   JOURNAL = {Selecta Math. (N.S.)},
  FJOURNAL = {Selecta Mathematica. New Series},
    VOLUME = {24},
      YEAR = {2018},
    NUMBER = {1},
     PAGES = {227--301},
      ISSN = {1022-1824},
   MRCLASS = {14F05 (14D24 14H60)},
  MRNUMBER = {3769731},
MRREVIEWER = {Andrew Salch},
       DOI = {10.1007/s00029-017-0360-4},
       URL = {https://doi.org/10.1007/s00029-017-0360-4},
}

@article{FrenkelGaiotto,
  title={Quantum {L}anglands dualities of boundary conditions, {D}-modules, and conformal blocks},
  author={Frenkel, Edward and Gaiotto, Davide},
  archiveprefix={arXiv},
  eprint={1805.00203},
  year={2020},
  journal = {Communications in Number Theory and Physics},
  Volume={14},
  pages={199--313}  
}

@article{ToenVezzosiFoliation,
  title={Algebraic foliations and derived geometry: the {R}iemann--{H}ilbert correspondence},
  author={To{\"e}n, Bertrand and Vezzosi, Gabriele},
  archiveprefix={arXiv},
  eprint={2001.05450},
  year={2020}
}

@article{BSY,
  title={Global shifted potentials for moduli spaces of sheaves on {CY4}},
  author={Borisov, Dennis and Sheshmani, Artan and Yau, Shing-Tung},
  archiveprefix={arXiv},
  erint={1908.00651},
  year={2019}
}

@article{DonagiPantev,
  title={Lectures on the geometric {Langlands} conjecture and non-abelian {Hodge} theory},
  author={Donagi, Ron and Pantev, Tony},
  journal={Second International School on Geometry and Physics Geometric Langlands and Gauge Theory},
  pages={129},
  year={2010},
  publisher={Citeseer}
}

@article {WittencomplexCS,
    AUTHOR = {Witten, Edward},
     TITLE = {Quantization of {C}hern-{S}imons gauge theory with complex
              gauge group},
   JOURNAL = {Comm. Math. Phys.},
  FJOURNAL = {Communications in Mathematical Physics},
    VOLUME = {137},
      YEAR = {1991},
    NUMBER = {1},
     PAGES = {29--66},
      ISSN = {0010-3616},
   MRCLASS = {58F06 (14H60 58D27 81S10 81T40 83C45)},
  MRNUMBER = {1099255},
MRREVIEWER = {N. J. Hitchin},
       URL = {http://projecteuclid.org/euclid.cmp/1104202513},
}

@book {LodayVallette,
    AUTHOR = {Loday, Jean-Louis and Vallette, Bruno},
     TITLE = {Algebraic operads},
    SERIES = {Grundlehren der Mathematischen Wissenschaften [Fundamental
              Principles of Mathematical Sciences]},
    VOLUME = {346},
 PUBLISHER = {Springer, Heidelberg},
      YEAR = {2012},
     PAGES = {xxiv+634},
      ISBN = {978-3-642-30361-6},
   MRCLASS = {18D50 (16E99)},
  MRNUMBER = {2954392},
MRREVIEWER = {Andrey Yu. Lazarev},
       DOI = {10.1007/978-3-642-30362-3},
       URL = {https://doi.org/10.1007/978-3-642-30362-3},
}

@incollection {BZFrenkelSS,
    AUTHOR = {Ben-Zvi, David and Frenkel, Edward},
     TITLE = {Geometric realization of the {S}egal-{S}ugawara construction},
 BOOKTITLE = {Topology, geometry and quantum field theory},
    SERIES = {London Math. Soc. Lecture Note Ser.},
    VOLUME = {308},
     PAGES = {46--97},
 PUBLISHER = {Cambridge Univ. Press, Cambridge},
      YEAR = {2004},
   MRCLASS = {17B67 (17B68 17B69)},
  MRNUMBER = {2079371},
       DOI = {10.1017/CBO9780511526398.006},
       URL = {https://doi.org/10.1017/CBO9780511526398.006},
}

@article{SafronovGQ,
  title={Shifted geometric quantization},
  author={Safronov, Pavel},
  eprint={2011.05730},
  archiveprefix={arXiv},
  year={2020}
}

@article{RaghavendranYoo,
	Author = {Raghavendran, Surya and Yoo, Philsang},
  eprint={1910.13653},
  archiveprefix={arXiv},
  	Year = {2019}}

\textsc{Amherst College}\\
\textsc{Department of Mathematics and Statistics, 220 South Pleasant Street,  MA 01002}\\
\texttt{celliott@amherst.edu}\\
\vspace{5pt}\\
\textsc{Seoul National University}\\
\textsc{Department of Mathematical Sciences \& Research Institute of Mathematics, Gwanak-Gu, Gwanak-Ro 1, Seoul 08826, Republic of Korea} \\
\texttt{philsang.yoo@snu.ac.kr}

\end{document}